\documentclass[letterpaper,11pt,twoside]{article}

\usepackage[left=1in, right=1in, bottom=1.25in, top=1.5in]{geometry}
\usepackage[utf8]{inputenc}
\usepackage{setspace}
\usepackage{fancyhdr}
\usepackage{amsmath}
\usepackage{amsfonts}
\usepackage{amssymb}
\usepackage{amsthm}
\allowdisplaybreaks
\usepackage{multirow}
\usepackage[T1]{fontenc}
\usepackage{xcolor}
\usepackage[mathscr]{euscript}
\usepackage{latexsym,bbm,xspace,graphicx,float,mathtools,mathdots,xspace}
\usepackage{enumitem}
\usepackage[ruled,vlined]{algorithm2e}
\usepackage{bm}
\usepackage[colorlinks,citecolor=blue,linkcolor=magenta,bookmarks=true]{hyperref}
\usepackage[nameinlink]{cleveref}

\usepackage{subfig}
\usepackage{graphicx}
\usepackage{tikz}

\usepackage{tablefootnote}

\fancypagestyle{plain}{%
\fancyhf{} % clear all header and footer fields
\fancyfoot[C]{\textbf{\thepage}} % except the center

}

\theoremstyle{plain}
\newtheorem{theorem}{Theorem}

\newtheorem{corollary}{Corollary}
\newtheorem{lemma}{Lemma}

\newtheorem{proposition}{Proposition}

\newtheorem{claim}{Claim}

\theoremstyle{definition}
\newtheorem{definition}{Definition}

\theoremstyle{remark}
\newtheorem{remark}{Remark}

\newtheoremstyle{restate}{}{}{\itshape}{}{\bfseries}{~(restated).}{.5em}{\thmnote{#3}}
\theoremstyle{restate}

\newcommand{\Mod}[1]{\ (\mathrm{mod}\ #1)}

\crefname{theorem}{Theorem}{Theorems}
\crefname{assumption}{Assumption}{Assumptions}
\crefname{corollary}{Corollary}{Corollaries}
\crefname{lemma}{Lemma}{Lemmas}
\crefname{conjecture}{Conjecture}{Conjectures}
\crefname{proposition}{Proposition}{Propositions}
\crefname{observation}{Observation}{Observations}
\crefname{claim}{Claim}{Claims}
\crefname{property}{Property}{Properties}
\crefname{op}{Open Problem}{Open Problems}
\crefname{problem}{Problem}{Problems}
\crefname{question}{Question}{Questions}

\crefname{definition}{Definition}{Definitions}
\crefname{example}{Example}{Examples}
\crefname{sketch}{Sketch}{Sketches}
\crefname{idea}{Idea}{Ideas}

\crefname{remark}{Remark}{Remarks}

\crefname{equation}{Equation}{Equations}
\crefname{figure}{Figure}{Figures}
\crefname{table}{Table}{Tables}

\newcommand{\ignore}[1]{}
\newcommand{\eps}{\epsilon}

\newcommand{\E}{\operatorname{{\bf E}}}
\newcommand{\Ex}{\mathop{{\bf E}\/}}
\renewcommand{\Pr}{\operatorname{{\bf Pr}}}
\newcommand{\Prx}{\mathop{{\bf Pr}\/}}

\newcommand{\X}{\mathcal{X}}
\newcommand{\Y}{\mathcal{Y}}
\newcommand{\rZ}{\mathcal{Z}}

\DeclareMathOperator\erf{erf}

\newcommand{\polylog}{\mathrm{polylog}}
\newcommand{\poly}{\mathrm{poly}}
\newcommand{\bcalZ}{\bm{\mathcal{Z}}}
\newcommand{\bx}{\bm{x}}
\newcommand{\by}{\bm{y}}
\newcommand{\bxi}{\bm{\xi}}

\newcommand{\R}{\mathbb R}
\newcommand{\RR}{\R_{\geq 0}}

\newcommand{\N}{\mathbb N}
\newcommand{\NN}{\N_{\geq 1}}
\newcommand{\Z}{\mathbb Z}

\renewcommand{\i}{\mathrm{i}\mkern1mu}   % for complex numbers
\renewcommand{\d}{\mathrm{d}}   % for integrals
\newcommand{\rhs}{\mathrm{RHS}} % for inequalities

\newcommand{\sig}{\mathrm{sig}}

\newcommand{\comment}[1]{}

\def\colorful{1}
\ifnum\colorful=1

\newcommand{\red}[1]{{\color{red} {#1}}}

\newcommand{\gray}[1]{{\color{gray}{#1}}}
\fi
\ifnum\colorful=0

\newcommand{\red}[1]{{{#1}}}

\newcommand{\gray}[1]{{{#1}}}
\fi

\usepackage[framemethod=TikZ]{mdframed}

\newcounter{prob}[section]
\newcounter{algo}[section]

\newenvironment{prob}[2][]{%
\refstepcounter{prob}%
\mdfsetup{%
frametitle={%
\tikz[baseline=(current bounding box.east),outer sep=0pt]
\node[anchor=east,rectangle,fill=white]
{\strut Problem~\arabic{prob}:~#1};}}%
\mdfsetup{innertopmargin=0pt,linecolor=black,%
linewidth=1pt,topline=true,%
frametitleaboveskip=\dimexpr-\ht\strutbox\relax
}
\begin{mdframed}[]\relax%
\label{#2}}{\end{mdframed}}

\newenvironment{algo}[2][]{%
\refstepcounter{algo}%
\mdfsetup{%
frametitle={%
\tikz[baseline=(current bounding box.east),outer sep=0pt]
\node[anchor=east,rectangle,fill=white]
{\strut Algorithm~\arabic{algo}:~#1};}}%
\mdfsetup{innertopmargin=0pt,linecolor=black,%
linewidth=1pt,topline=true,%
frametitleaboveskip=\dimexpr-\ht\strutbox\relax
}
\begin{mdframed}[]\relax%
\label{#2}}{\end{mdframed}}

\newcommand{\rnote}[1]{\footnote{{\bf \color{red}Rocco}: {#1}}}
\newcommand{\trnote}[1]{\footnote{{\bf \color{orange}Tim R.}: {#1}}}

\newcommand{\yjnote}[1]{\footnote{{\bf \color{blue}Yaonan}: {#1}}}

\def\calS{\mathcal{S}} \def\calB{\mathcal{B}} 

\title{Average-Case Subset Balancing Problems
}
\author{Xi Chen, Yaonan Jin, Tim Randolph, and Rocco A.~Servedio 
\\ \texttt{ \{xichen, rocco\}@cs.columbia.edu, \{yj2552, t.randolph\}@columbia.edu }
\\ Columbia University}
\date{\today}

\begin{document}

\maketitle

\begin{abstract}

Given a set of $n$ input integers, the \emph{Equal Subset Sum} problem asks us to find two distinct subsets with the same sum. In this paper we present an algorithm that runs in time $\smash{O^*(3^{0.387n})}$ in the~average case, significantly improving over the $\smash{O^*(3^{0.488n})}$ running time of the best known worst-case algorithm  \cite{mucha2019equal} and the  Meet-in-the-Middle benchmark of $O^*(3^{0.5n})$. 

Our algorithm generalizes to a number of related problems, such as the ``\emph{Generalized Equal Subset Sum}'' problem, which asks us to assign a coefficient $c_i$ from a set $C$ to each input number $x_i$ such that $\sum_{i} c_i x_i = 0$. Our algorithm for the average-case version of this problem runs in~time $|C|^{(0.5-c_0/|C|)n}$ for some positive constant $c_0$, whenever  $C=\{0, \pm 1, \dots, \pm d\}$ or $\{\pm 1, \dots, \pm d\}$~for some positive integer $d$ (with runtime  %for every $d$, with runtime 
$O^*(|C|^{0.45n})$ when $|C|<10$). %for $|C| < 10$. 
Our results extend to the~problem of finding ``nearly balanced'' solutions in which the target is a not-too-large nonzero offset $\tau$.
%rather than $0$.

Our approach relies on new structural results that characterize the probability that $\sum_{i} c_i x_i$ $=\tau$ has a solution $c \in C^n$ when $x_i$'s are chosen randomly;
%uniformly at random from $\{0,1,\dots,M-1\}^n$; 
these results may be of independent interest.\ignore{We present a new elementary approach for lower-bounding the probability of a solution for coefficient sets with 0, which gives our algorithm exponentially small failure probability in this case. We also extend the classic work of Borgs, Chayes, and Pittel \cite{BCP01} for coefficient sets without 0, which gives $o_n(1)$ failure probability.} 
%We use these structural results to prove correctness of our algorithm, 
Our algorithm is inspired by the ``representation technique''  introduced by Howgrave-Graham and Joux \cite{howgrave2010new}. 
This requires several new ideas
  to overcome preprocessing hurdles
  that arise in the representation framework,
  as well as a novel application of dynamic programming in the solution recovery phase of the algorithm.
%overcome preprocessing hurdles that arise in the representation technique, including ``instance translation'' to reduce solution size and ``shrinking'' to reduce input size, as well as a novel application of dynamic programming in the solution recovery phase of the algorithm.

%Finally, we discuss how our results imply faster algorithms for related problems such as Number Balancing with exponential precision.
\end{abstract}

\thispagestyle{empty}
\newpage
\setcounter{page}{1}

\section{Introduction}
\label{sec:intro}

The Subset Sum problem and its variants are among the most famous NP-complete problems, and the question of whether Subset Sum can be solved in time $O^*(2^{(0.5 - \delta)n})$\footnote{We use $O^*(\cdot)$ notation to suppress factors polylogarithmic in the argument of the function; so the notation ``$O^*(2^{0.5n})$'' suppresses $\poly(n)$ factors.} for some constant $\delta > 0$ is one of the major questions in exact algorithms. Despite many attempts to find a ``truly faster'' exact algorithm, Horowitz and Sahni's classic $O^*(2^{0.5n})$-time Meet-in-the-Middle algorithm for worst-case inputs remains the benchmark after almost 50 years \cite{horowitz1974computing}.  

In recent years, the apparent difficulty of the Subset Sum problem has fueled work on variants and related settings. In 2010, \cite{howgrave2010new} made a significant breakthrough by showing that Subset Sum could be solved in time $O^*(2^{0.337n})$ in the average case under a reasonable heuristic assumption.\footnote{Although the original paper claims a running time of $O^*(2^{0.311n})$, a correction of the original analysis gives this $O^*(2^{0.337n})$ runtime; see \cite[Section~2.2]{becker2011improved} for details. } Subsequent works, including \cite{becker2011improved,bohme2011,bonnetain2020improved},
%and \cite{esser2019better}
% TR: Removing Esser-May until we have a clearer idea of whether their result holds up.
refined what became known as the \emph{representation technique}, and whittled the average-case exponent down to $O^*(2^{0.283n})$
%$O^*(2^{0.255n})$ 
under heuristic assumptions. This new technique inspired a flurry of results in related settings, including better time-space tradeoffs \cite{austrin2013space, dinur2012efficient}, a faster polynomial-space algorithm \cite{bansal}, and fast algorithms for large classes of sufficiently ``random-like'' instances \cite{austrin2015subset, austrin2016dense}.

In 2019, \cite{mucha2019equal} used the representation technique to establish an $O^*(3^{0.488n})$-time \emph{worst-case} algorithm for \emph{Equal Subset Sum}, the Subset Sum variant that asks for two different input subsets with the same sum. This resolved an open question of Woeginger (\cite{woeginger2008open}), who observed that the Meet-in-the-Middle approach to Equal Subset Sum runs in time $O^*(3^{0.5n})$ and asked whether this was a barrier for Equal Subset Sum analogous to the $O^*(2^{0.5n})$ runtime barrier for Subset Sum. Equal Subset Sum is also closely related to the Number Balancing problem of \cite{karmarkar1982differencing}, which can be thought of as the optimization version of Equal Subset Sum.

Both Subset Sum and Equal Subset Sum are special cases of the following more general problem: given a multiset of input integers, find a linear combination using only coefficients from a small set $C$ that achieves a specified target value. Writing $[a:b]$ to denote the set of integers $\{a, a+1, \dots, b\}$, we define:

\begin{prob}[Generalized Subset Sum (GSS)]{prob:GSS}
    \textbf{Input.} An input range bound $M$, an input vector $\vec{x} = (x_1,  \dots, x_n) \in [0:M-1]^n$, a set $C \subset \Z$ of allowed coefficients, and a target integer $\tau$. \\
    \textbf{Output.} A coefficient vector $\vec{c} \in C^n$ such that $\vec{c}\cdot \vec{x} = \tau$, if one exists.
\end{prob}

% \begin{prob}[Generalized Equal Subset Sum (GESS)]{prob:GESS}
%     \textbf{Input.} An input range bound $M$, an input vector $\vec{x} = (x_1, x_2, \dots, x_n) \in [0:M-1]^n$, and a coefficient set $C \subset \mathbb{Z}$. \\
%     \textbf{Output.} A coefficient vector $\vec{c} \in C^n$ such that $\vec{c}\cdot \vec{x} = 0$, if one exists. We disallow the all-zeros solution.
% \end{prob}

It is natural to compare the runtime of GSS algorithms to the size of the search space $|C|^n$. There is a straightforward generalization of Horowitz and Sahni's Meet-in-the-Middle Algorithm for GSS:  partition the input into two vectors $\vec{x}_1$ and $\vec{x}_2$ of length $n/2$ each, list all partial linear combinations $\vec{c}_{1} \cdot \vec{x}_1$ and $\vec{c}_{2} \cdot \vec{x}_2$ by enumerating all $\vec{c}_{1}, \vec{c}_{2} \in C^{n/2}$, and search the lists for a pair of linear combinations such that $\vec{c}_1 \cdot \vec{x}_1 + \vec{c}_2 \cdot \vec{x}_2$ achieves the target $\tau$. This can be done in time $O^*(|C|^{n/2})$ by sorting the lists and using two pointers that walk from opposite ends of the two lists, seeking a pair that sums to the target. The $C = \{0,1\}$ case of this algorithm is Horowitz and Sahni's Meet-in-the-Middle algorithm for Subset Sum, and the $C = \{0, \pm 1\}$ case (with target $\tau=0$ and the solution $\vec{c}=\vec{0}$ disallowed) is the Meet-in-the-Middle algorithm for Equal Subset Sum. 

\ignore{
START IGNORE
\begin{algo}[Meet-in-the-Middle for GSS (informal)]{algo:mim}
    \begin{enumerate}
        \item Partition the input into two vectors $\vec{x}_1$ and $\vec{x}_2$, each of length $n/2$.
        \item List all partial linear combinations $\vec{c} \cdot \vec{x}_1$ and $\vec{c} \cdot \vec{x}_2$ by enumerating all $\vec{c} \in C^{n/2}$.
        \item Search the lists for a pair of linear combinations such that $\vec{c}_1 \cdot \vec{x}_1 + \vec{c}_2 \cdot \vec{x}_2$ achieves the target. This can be done in time $O^*(|C|^{n/2})$ by sorting the lists and using %binary search.\footnote{The search step can be replaced with two pointers that walk from opposite ends of the two lists, seeking a pair that sum to the target. Incrementing one pointer or the other based on whether the provisional sum is too large or too small saves a factor of $n$ in the runtime and is the source of the name ``Meet-in-the-Middle''. }
        two pointers that walk from opposite ends of the two lists, seeking a pair that sums to the target.
    \end{enumerate}
\end{algo}
END IGNORE}

In light of the broad body of work that has been done on the average-case version of the original Subset Sum problem, it is natural to consider the \emph{average-case} GSS problem in which the input vector $\vec{\bx}$ is uniformly random over $[0:M-1]^n$; this average-case GSS problem is the subject of the current paper. There are two natural average-case variants of the GSS problem:  in the first variant, the target value is obtained by sampling a ``hidden'' solution, and hence every instance of this average-case problem variant admits a solution.  This variant of the problem is motivated by cryptographic applications, and corresponds to the average-case variant of Subset Sum that was studied by \cite{howgrave2010new,becker2011improved} and others; we refer to it as the ``cryptographic version'' of the average-case GSS problem. In the second version, which is the one we consider in this work, the target is a fixed value which does not depend on the draw of random input $\vec{\bx}$. \ignore{.\footnote{We consider a fixed value $|\tau| = o(Mn)$. In this range, the parameter settings at which the problem is likely to have a solution are similar to the 0 target. }}This version, which we refer to as the ``balancing version'' (see the formal definition below), has more of the flavor of Equal Subset Sum. %which corresponds to a target value of 0. 
Since both yes-instances and no-instances are possible for this variant, a natural goal that arises in its study is to understand the probability (as a function of the various parameter settings) that a solution exists.  Structural questions of this type have in fact been the subject of considerable study for the special case of $C = \{\pm 1\}$; see for example \cite{borgs2001phase,lueker1998exponentially}.

%(we will discuss  although they focus on minimizing the discrepancy between two sets and are not tailored for solving the decision problem . %We present structural results and fast exponential time algorithms for the balancing variant.

%In this work, we consider the balancing variant and refer to a \emph{fixed offset} $|\tau|$ to distinguish this setting from the \emph{random target} $t$ considered in the cryptographic variant.

% \orange{ There are two interesting ``flavors'' of the average case. One is the ``crypto version'', in which a target is ``hidden'', randomly determined, and found, and the other is the ``structural/balancing'' version, in which there's a small fixed offset. $0$ the archetypal case but other small offsets work just as well, as we show. We're focused on the second case. Our structural results answer this natural question and are interesting in their own right - when we draw a bunch of numbers, what parameter settings let you combine them to make 0? (To do: See if BCP/Luecker consider motivations for their problem that might be relevant here.) }

\subsection{Our Results}

We consider two families of symmetric coefficient sets, with and without zero. For a fixed positive integer $d \in \mathbb{N}$, we define
$C(d)  = \{ {\pm 1}, \pm 2, \dots, \pm d \}$ and 
$C_0(d)  = \{0, \pm 1, \pm 2, \dots, \pm d \}$.
These two sets cover a wide range of cases, including Equal Subset Sum. We note that while the natural coefficient set corresponding to Subset Sum is $C = \{0,1\}$, an instance $\vec{x}$ with target $\tau$ can be easily converted to GSS on $C(1) = \{\pm 1\}$ by setting a new target $\tau' := 2\tau - \sum_{i \in [n]} x_i$.\footnote{In fact, our results extend to all scale multiples and translations of $C(d)$ and $C_0(d)$. For details, refer to \Cref{subsec:other-coefficients}. }

We consider algorithms for 
  the balancing variant
  of average-case GSS over either $C=C(d)$ or $C_0(d)$.
Given as input $M,\tau$ and $\vec{x}\in [0:M-1]^n$,
  such an algorithm either returns ``no solution'' or a solution
  $\vec{c}$
  that satisfies $\vec{c}\cdot \vec{x}=\tau$ (with $\vec{c}\ne \vec{0}$ when $C=C_0(d)$
  and~$\tau=0$).
  %or ``no solution'' with the following performance guarantee:  
We say an algorithm fails on $(M,\tau,\vec{x})$ if it returns ``no solution'' but indeed there is a solution; otherwise we say it succeeds.
%  otherwise it is correct .

%We give new structural results and fast exponential time algorithms for the balancing version of GSS in the average case (i.e.\ when $\vec{\bx}$ is sampled uniformly at random over $[0: M-1]^{n}$). 
%\rnote{I think maybe we should not introduce a second very similar looking problem GESS in a box given that we already introduced GSS - I tried an exposition that doesn't introduce ``GESS'' but we can revert back if people prefer.}\ignore{We refer to this problem as \emph{Generalized Equal Subset Sum} (GESS), and refer to a \emph{fixed offset} $\tau$ to distinguish our target from the randomly determined target $\tau$ used in the cryptographic variant.}

%Our main algorithmic result is a fast exponential-time algorithm for (the balancing variant of) average-case GSS:
Our main algorithmic result for (the balancing variant of) average-case GSS is as follows.

\begin{theorem}[Algorithm for Average-Case GSS]
    \label{thm:main}
    Fix any $d \in \N_{\geq 1}$ and let $C = C(d)$ or $C_0(d)$. 
    For any constant $\zeta >0$,
    there is a randomized algorithm for average-case GSS with running time
    % \xnote{I think we don't need the $O()$ below.}
    %\xnote{I think $\N$ is the set of natural numbers so we don't need the subscript $\ge 1$. \yj{Dispute... Some people think the set of natural numbers $\N$ includes $0$ but some others think not... Let add back the subscript $\geq 1$ for clearance.}}
\[
    O^*(|C|^{\Lambda(|C|) n + \zeta n})
    \footnote{Note that the runtime of our algorithm is independent of the input bound $M$. This occurs because the probability of a `yes' instance is exponentially small when $M = 2^{\omega(n)}$ (refer to \Cref{thm:optprecisionC,thm:optprecisionC0}).}         
    \quad \ \
        \text{where } \;\; \Lambda(z) = \max
        \begin{cases}
            1 - \frac{z+1}{2z}\log_{z}(z+1) + \frac{1}{z}\log_{z}(2) \\[0.8ex]
            \frac{2}{3} - \frac{z + 1}{3z} \log_{z}\big(\frac{z + 1}{2}\big) 
        \end{cases}.
\]
%For uniform random $\vec{\bx} \in [0:M-1]^n$, the algorithm is correct with probability at least $1 - e^{-\Omega(n)}$ for $C=C_0(d)$ and $1 - o_n(1)$ for $C=C(d)$.
Given any $M$ and $\tau$ with $|\tau|=o(nM)$,
  the algorithm succeeds on $(M,\tau,\vec{\bx})$ with
  probability at least $1-e^{-\Omega(n)}$ when $C=C_0(d)$ and 
  with probability at least $1-o_n(1)$ when $C=C(d)$, over the draw of $\vec{\bx}\sim [0:M-1]^n$ and 
  the randomness of the algorithm.
%\[
%        |C|^{\Lambda(|C|) n + O(\eps n)}, \quad \quad
%        \text{where } \;\; \Lambda(z) = \frac{2}{3} - \frac{z + 1}{3z} \log_{z}\Big(\frac{z + 1}{2}\Big).
%\]
%For uniform random $\vec{\bx} \in [0:M-1]^n$, the algorithm is correct with probability at least $1 - e^{-\Omega(n)}$ for $C=C_0(d)$ and $1 - o_n(1)$ for $C=C(d)$.
\end{theorem}

% We note that $\max_{z \in \NN} \Lambda(z)$, the worst constant that can arise in the exponent of our runtime, is approximately $0.3954$, occurring when $z = 8$ (corresponding to $C = C(4) = \{\pm 1, \dots, \pm 4\}$), and that $\lim_{z \rightarrow \infty} \Lambda(z) = 1/3$. 
We note that $\Lambda(z) = 0.5 - \Omega(1 / z)$, and thus our algorithm beats the Meet-in-the-Middle runtime of $O^\ast(|C|^{0.5n})$ by an exponential margin for every constant $|C|$. \Cref{fig:result} plots the function $\Lambda$ and \Cref{table:results} lists our algorithm's runtime on various coefficient sets.

As a special case of \Cref{thm:main} we obtain an average-case algorithm for Equal Subset Sum that significantly improves on the worst-case  $\smash{O^*(3^{0.488n}})$ runtime of \cite{mucha2019equal}:

\begin{corollary}[Algorithm for Average-Case Equal Subset Sum]
    There exists an algorithm that solves Average-Case Equal Subset Sum in time $O^*(3^{0.387n})$ with success probability $1 - e^{-\Omega(n)}$.
    %\xnote{I think in all these bounds with numbers in the exponent we don't need the $O^*$ because I suspect we 
    %have rounded all constants up so there is enough room to overcome any polynomials. \yj{Some exponents like $3^{0.488n}$ are rounded up, while others are not -- e.g., for the Meet-in-the-Middle we must use $O^*(|C|^{0.5n})$. \cite{howgrave2010new,becker2011improved,bansal,mucha2019equal} use the $O^*$ notation even for runtimes like $3^{0.488n}$, in order to avoid the inconvenient in switching $O$ versus $O^*$. \\
    %Although $3^{0.387n}$ is also okay here, maybe keep $O^*(3^{0.387n})$ for consistency...}}
\end{corollary}

\begin{figure}[t!]
    \centering
    \begin{tikzpicture}[thick, smooth, scale = 1.75]
    \draw[->] (0, 0) -- (8.25, 0);
    \draw[->] (0, 0) -- (0, 3.25);
    \node[above] at (0, 3.25) {\small $\Lambda(|C|)$};
    \node[right] at (8.25, 0) {\small $|C|$};
    
    \draw[thick] (1, 0) -- (1, 0.05);
    \draw[thick] (2, 0) -- (2, 0.05);
    \draw[thick] (3, 0) -- (3, 0.05);
    \draw[thick] (4, 0) -- (4, 0.05);
    \draw[thick] (5, 0) -- (5, 0.05);
    \draw[thick] (6, 0) -- (6, 0.05);
    \draw[thick] (7, 0) -- (7, 0.05);
    \draw[thick] (8, 0) -- (8, 0.05);
    
    \node[left] at (0, 0) {\small $0.2$};
    \node[left] at (0, 1) {\small $0.3$};
    \node[left] at (0, 2) {\small $0.4$};
    \node[left] at (0, 3) {\small $0.5$};
    \draw[thick] (0, 1) -- (0.1, 1);
    \draw[thick] (0, 2) -- (0.1, 2);
    \draw[thick] (0, 3) -- (0.1, 3);
    
    \node[below] at (0, 0) {\small $2$};
    \node[below] at (1, 0) {\small $3$};
    \node[below] at (2, 0) {\small $4$};
    \node[below] at (3, 0) {\small $5$};
    \node[below] at (4, 0) {\small $6$};
    \node[below] at (5, 0) {\small $7$};
    \node[below] at (6, 0) {\small $8$};
    \node[below] at (7, 0) {\small $9$};
    \node[below] at (8, 0) {\small $10$};
    
    % Draw 2/3 - (x+1)/3x * log_x(x+1/2), change x to x+2, scale by 10x and subtract 2 to make it plot correctly
    \draw[domain = 0:1.6, smooth, variable=\x, blue] plot({\x}, {
        ((2/3 - ((\x+3)/(3*(\x+2))) * ln((\x+3)/2)/ln(\x+2)) * 10) - 2
    });
    % Draw 1 - (x+1)/2x * log_x(x+1) + 1/x * log_x(2), change x -> x+2, scale by 10x and subtract 2 to make it plot
    \draw[domain = 1.6:8, smooth, variable=\x, blue] plot({\x}, {
        ((1 - ((\x+3)/(2*(\x+2))) * ln(\x+3)/ln(\x+2)  + (ln(2)/ln(\x+2))/(\x+2)) * 10) - 2
    });
    
    % New points for |C| = 2 through 10
    \fill[red] (0, 0.374 * 10 - 2) circle [radius = 1pt];
    \fill[red] (1, 0.386 * 10 - 2) circle [radius = 1pt];
    \fill[red] (2, 0.3995 * 10 - 2) circle [radius = 1pt];
    \fill[red] (3, 0.4182 * 10 - 2) circle [radius = 1pt];
    \fill[red] (4, 0.431 * 10 - 2) circle [radius = 1pt];
    \fill[red] (5, 0.4402 * 10 - 2) circle [radius = 1pt];
    \fill[red] (6, 0.4473 * 10 - 2) circle [radius = 1pt];
    \fill[red] (7, 0.4529 * 10 - 2) circle [radius = 1pt];
    \fill[red] (8, 0.4573 * 10 - 2) circle [radius = 1pt];
    
    \end{tikzpicture}
    \caption{Plot of $\Lambda$. The red points plot $\Lambda(z)$ for $z \in [2:10]$.
    }
    \label{fig:result}
\end{figure}
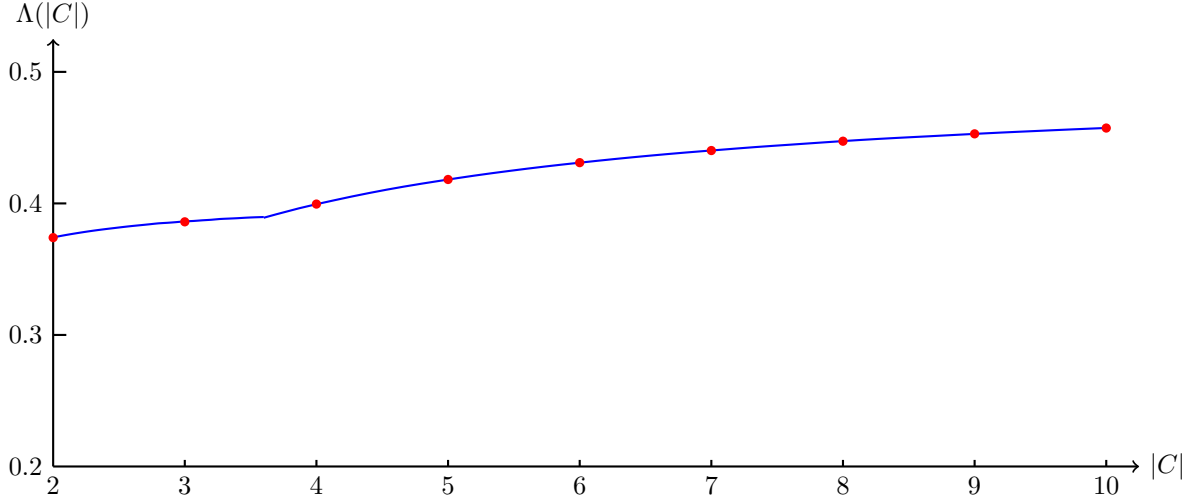

\begin{table}[t!]
\label{table:results}
\centering
\begin{tabular}{|c|c|c|c|c|c|c|}
    \hline
    \multicolumn{2}{|c|}{\rule{0pt}{12pt}$d \in \N_{\ge 1}$} & $1$ & $2$ & $3$ & $5$ & $10$ \\ [1pt]
    \hline
    \hline
    \rule{0pt}{12pt}\multirow{2}{*}{Runtime} & $C(d)$ & $O^*(|C|^{0.375n})^\dagger$ & $O^*(|C|^{0.400n})$ & $O^*(|C|^{0.432n})$ & $O^*(|C|^{0.458n})$ & $O^*(|C|^{0.479n})$ \\ [1pt]
    \cline{2-7}
    \rule{0pt}{12pt}& $C_{0}(d)$ & $O^*(|C|^{0.387n})^\ddagger$ & $O^*(|C|^{0.419n})$ & $O^*(|C|^{0.441n})$ & $O^*(|C|^{0.462n})$ & $O^*(|C|^{0.480n})$ \\ [1pt]
    \hline
    \multicolumn{7}{p{0.975\textwidth}}{\footnotesize $^\dagger$Our $C(1)$ case differs from the ``cryptographic'' average-case Subset Sum problem considered in \cite{howgrave2010new,becker2011improved,bohme2011,bonnetain2020improved}, because we consider the ``balancing'' problem with a fixed offset (for which a solution may or may not exist); see the Introduction for further discussion.} \\
    \multicolumn{7}{p{0.975\textwidth}}{\footnotesize $^\ddagger$This runtime for $C_{0}(1)$ should be contrasted with the $O^*(|C|^{0.488n})$-time worst-case runtime due to \cite{mucha2019equal}.}
\end{tabular}
\caption{Runtime of our algorithm for average-case GSS on various coefficient sets.}
\end{table}

Our algorithm has the additional property that it runs faster on \emph{dense} instances, i.e.,\ ones for which $M$ is substantially less than $|C|^n$. Intuitively, this is possible because in this regime there are likely to be many solutions.
%\ignore{ if $M < |C|^{n}$, the expected number of solutions grows exponentially. In this case, the base of our runtime scales with $M$.}

\begin{theorem}[Average-Case GSS on Dense Instances]
    \label{thm:dense-instances}
    Fix $d \in \NN$, $C = C(d)$ or $C = C_0(d)$, $M = |C|^{\alpha n + o(n)}$ for some $\alpha \in (0, 1)$ and an offset $\tau$ with $|\tau| = o(Mn)$. For any constant $\zeta > 0$, there exists an algorithm that solves average-case GSS in time
    \begin{align*}
         O^*(|C|^{\alpha \Lambda(|C|) n + \zeta n}),
    \end{align*}
    where $\Lambda$ is defined as in \Cref{thm:main}.  The algorithm succeeds with probability at least $1 - e^{-\Omega(n)}$ for $C= C_0(d)$ and with probability at least $1 - o_n(1)$ for $C=C(d)$.
\end{theorem} 

Crucial ingredients underlying our algorithms are new structural results on the probability that random GSS instances have solutions.  In the context of previous work on closely related questions  \cite{borgs2001phase,lueker1998exponentially}, we believe that our new structural results may be of independent interest; we explain our new structural results and contrast them with prior work below.

\subsection{Our Techniques}

To explain our structural and algorithmic results, we begin with some intuition for the distribution of sums that are achievable using coefficient set $C$. Consider a uniform random draw of the input vector $\vec{\bx}$ from $[0:M-1]^n$. Since the coefficient set $C$ is symmetric about 0, \ignore{by standard concentration results almost} all elements of the random set $\bm{S}_n := \{\vec{c} \cdot \vec{\bx} : \vec{c} \in C^n\}$ have magnitude $O(Mn)$, and hence intuitively $\bm{S}_n$ is ``tightly concentrated around the origin''. \ignore{viewing $\bm{S}_n$ as a $|C|^n$-element multiset, }
Consequently, when the offset $\tau$ is not too large (the case of interest for us) it is natural to expect that a solution $\vec{c} \cdot \vec{\bx} = \tau$ is likely to exist if $M \ll |C|^n$ and is unlikely to exist if $M \gg |C|^n$. Indeed, a simple counting argument establishes this for the case that $M \gg |C|^n$, but substantially more work is required to rigorously confirm the above intuition for the $M \ll |C|^n$ case.  This is accomplished by our structural results, which we describe next.

\ignore{%This is precisely the intuition confirmed by the structural theorems proved in \Cref{sec:structural}. When $M \gg |C|^n$, a simple counting argument suffices to show that solutions are unlikely. The challenge is to prove that solutions exist with (very) high probability when $M \ll |C|^n$. 
}

\subsubsection{ Structural Results }

%\rnote{Then main change I made here is to state our two structural results in more detail - it felt to me like our description of them was a little oblique so I made it more explicit. I didn't go quite all the way to a boldface theorem or boldface ``(Informal version)'' of the theorem but we could conceivably do that.}
Our structural results for coefficient set $C_0(d)$ (including zero) and for $C(d)$ (not including zero) are established using very different techniques. Consider the case of $C=C_0(d)$ first. We are able to show that for any positive constant $\eps > 0$, given $M \leq |C|^{(1-\eps)n}$ and any fixed integer offset $\tau$ with $|\tau|=o(Mn)$, the probability (over a uniform random $\vec{\bx}\sim[0:M-1]^n$) that there exists a solution $\vec{c} \cdot \vec{\bx} = \tau$ with $\vec{c} \in C^n$ is \emph{exponentially} close to 1.  See \Cref{thm:optprecisionC0} for a detailed statement of this result. This extremely high probability that there exists a solution translates directly into the extremely high success probability of our algorithms in \Cref{thm:main} and \Cref{thm:dense-instances}.

To establish this structural result, we employ a novel proof strategy that aligns with the intuition sketched earlier and is based on an iterative analysis. We define random sets $\bm{S}_1, \bm{S}_2, \dots$ where the set $\bm{S}_\ell$ for $\ell \in [n]$ is given by $\bm{S}_{\ell} := \{\vec{c}' \cdot \vec{\bx}': \vec{c}' \in C^\ell\}$ with a uniform random $\vec{\bx}' \sim [0:M-1]^\ell$. We then analyze how these set sizes $|\bm{S}_{\ell}|$ increase as a function of $\ell$ by thinking of $\bx_1,\bx_2,\dots$ as being drawn in succession.  We first argue that with very high probability $|\bm{S}_{\ell}|$ increases rapidly with $\ell$ until, at some value $\bm{L}_1 < n$, it reaches a point at which it is dense on at least one ``large'' interval that is ``close to'' the origin. We then argue that at this point it suffices to draw a few more elements to ensure that some {\em partial solution} $\sum_{i \in [\bm{L}_2]} c_{i} \bx_{i}$ will hit the offset $\tau$ with very high probability, where $\bm{L}_2 \geq \bm{L}_1$ and $\bm{L}_2$ is still less than $n$. The remaining elements $\bx_{\bm{L}_2+1},\dots,\bx_n$ are simply assigned the $0$ coefficient to complete the overall solution. The detailed proof is given in \Cref{subsec:proof_C0}.

Turning to the case of coefficient set $C=C(d)$, it is clear that the absence of the 0 coefficient is a fundamental obstacle to the previous approach, and indeed we do not know how to achieve an exponentially high success probability for $C=C(d).$  Instead, to handle this case we use a very different proof strategy that extends the approach of  \cite{borgs2001phase}, who analyzed the $C(d)$ case for $d=1$.  Their analysis (see \cite[Theorem~2.1]{borgs2001phase}) shows that for any $M\leq 2^{(1-\eps)n}$, the probability that a uniform random input $\vec{\bx} \sim [0:M-1]^n$ admits a ``perfect'' solution $\vec{c} \in \{\pm 1\}^n$ is $1-o_n(1)$, where a ``perfect'' solution is one satisfying $\vec{c} \cdot \vec{\bx}=1$ if $\sum_{i \in [n]} x_i$ is odd and satisfying $\vec{c} \cdot \vec{\bx} = 0$ if $\sum_{i \in [n]} x_i$ is even.\footnote{\cite{borgs2001phase} further established a range of more refined structural results, but this is the most relevant result for our purposes.}

We extend the \cite{borgs2001phase} analysis and establish a similar result for general coefficient sets $C=C(d)$ for arbitrary $d>1$. We show that for any constant $d>1$, if $M \leq |C|^{(1-\eps)n}$ then given any fixed integer offset $\tau$ with $|\tau|=o(Mn)$, the probability (over a uniform draw of $\vec{\bx}$ from $[0:M-1]^n$) that there exists a solution $\vec{c} \cdot \vec{\bx} = \tau$ with $\vec{c} \in C^n$ is $1-o_n(1).$
At a high level, our analysis establishing this follows the arguments of \cite{borgs2001phase}; we write the number of solutions as a random integral over all coefficient vectors, then bound suitable integrals to characterize the first moment and upper bound the second moment of the relevant random variable.  The analysis requires considerable attention to detail (the domain of integration needs to be broken into three different regions, with different arguments required for each region); see \Cref{thm:optprecisionC} for a detailed statement and \Cref{apx:discard} for the proof.

\ignore{
%The two proofs provide the reader with complementary insights into the structure of random Generalized Subset Sum instances, one combinatorial and the other analytic. The first produces a better bound on the probability of a solution ($1 - e^{-\Omega(n)}$ when $M < |C|^{(1 - \eps)n}$ for any constant $\eps > 0$, as opposed to $1 - o_n(1)$). However, the other method is not completely without reward. \cite{borgs2001phase} contend with an interesting parity issue: in the $C(1)$ case, no solution is possible if the sum of the inputs is odd. Is this a unique product of the $C(1)$ case, or do similar phenomena occur for other coefficient sets? Our extended proof explains the parity issue in terms of constructive interference of cosine functions that does not occur for $C(d)$, $d > 1$.
%\trnote{Getting wordy - maybe footnote or cut this paragraph.}
}

\begin{remark}[Comparison with \cite{lueker1998exponentially}] It is interesting to compare our structural results \Cref{thm:optprecisionC} and \Cref{thm:optprecisionC0} with the main structural result of \cite{lueker1998exponentially}. That work analyzes the probability that given a real-valued vector $\vec{\bm{X}} = (\bm{X}_{1}, \dots, \bm{X}_{n})$ drawn uniformly at random from $[-1, 1]^n$, there exists a coefficient vector $\vec{\delta} \in \{0,1\}^n$ such that $\vec{\delta} \cdot \vec{\bm{X}}$ lies within a small range $\pm \eta$ of a specified target value $z \in [-\frac{1}{2}, \frac{1}{2}]$. By scaling the $\bm{X}_i$’s, \cite{lueker1998exponentially} implies that if $M \leq \kappa^n$ where $\kappa = 2^{1 / (2 + 2\ln 2)} \approx 1.227$, then for any target $\tau \in [-\frac{M}{2}, \frac{M}{2}]$, with exponentially high probability over a uniform draw of $\bx_1,\dots,\bx_n$ each from the continuous interval $[-M, M]$, there exists a solution $\vec{c} \in C(1)^n = \{\pm 1\}^n$ for which $|\vec{c} \cdot \vec{x} - \tau| \leq 1/2$.
%By rounding the $\bm{X}_i$'s and rescaling, \cite[\orange{Corollary~2.5}]{lueker1998exponentially} implies that for some constant $0 < \kappa < 1$, if $M \leq (1+\kappa)^n$ then with exponentially high probability over a uniform draw of $\vec{\bx}$ from $[0:M-1]^n$, there exists a solution $\vec{c} \in C(1)^n = \{\pm 1\}^n$ for which $\vec{c} \cdot \vec{\bx}$ hits a prescribed \orange{offset from $\lfloor M\sum_i \bm{X}_i \rfloor$}.\rnote{Is this a correct statement - we believe Luecker gives this? \orange{TR: I checked this and I think it's mostly right. In particular, I think the `prescribed target value' is really a `prescribed target offset from the sum $\sum_i \bm{X}_i$: I think Luekcer's result says that we can hit any target in $\sum_i \bm{X}_i \pm \frac{M}{2}$ with exponentially high probability. I agree with the condition written out above with $\kappa$ ($\kappa = 1/2C \approx 0.2$, right?). } } 
While this result has a similar flavor to our \Cref{thm:optprecisionC}, our \Cref{thm:optprecisionC0}, and \cite[Theorem~2.1]{borgs2001phase}, it is incomparable to all of them (even apart from the fact that it deals with continuous rather than discrete random variables). %Unlike \Cref{thm:optprecisionC0}, it applies to the \orange{ asymmetric coefficient set $\{0,1\}$}.  
The exponentially high probability bound that it gives is similar to the \Cref{thm:optprecisionC0} bound and is better than the $1-o_n(1)$ probability bound of \cite[Theorem~2.1]{borgs2001phase}, but the quantitative bound $M \leq \kappa^n$ that it requires is worse than the 
%$M \leq 2^{(1-\eps)n}$ 
bound allowed by \cite[Theorem~2.1]{borgs2001phase}. We note that unlike the result implicit in \cite{lueker1998exponentially}, both \Cref{thm:optprecisionC} and \Cref{thm:optprecisionC0} apply to general coefficient sets $C$, and that both \Cref{thm:optprecisionC}, \Cref{thm:optprecisionC0} and \cite[Theorem~2.1]{borgs2001phase} allow the ``correct'' range of values for $M$, i.e., $M$ is allowed to be as large as $|C|^{(1-\eps)n}$. (Having this ``correct'' range of values for $M$ is essential for our ultimate purpose of developing average-case GSS algorithms. Otherwise, there would be a wide range of values for $M$ to which no strong structural result would apply.)
\end{remark}

\subsubsection{ Algorithmic Results }

Intuitively, our structural results tell us the parameter settings for which average-case GSS instances are likely to have solutions and what those solutions will look like. This allows us to prove the correctness of our algorithm, which combines new preprocessing ideas with an approach based on the representation technique.

\vspace{.1in}
\noindent
\textbf{The Representation Technique.} The representation technique works by constructing many partial \emph{candidate solutions} and filtering them in a way that reduces the search space while retaining at least one way of constructing a solution. For example, the key observation of \cite{howgrave2010new} is that any Subset Sum solution of size $n/2$ can be decomposed into $\Omega^\ast(2^{0.5n})$ ``solution pairs'', each consisting of two disjoint sets of size $n/4$. Since there are $\binom{n}{n / 4} = \Theta^*(2^{0.811n})$ input subsets of size $n/4$, this means that (under sufficiently strong assumptions on the input) hashing these subsets in a way that ensures solution pairs hash to the same value reduces the search space to $O^*(2^{0.811n - 0.5n}) = O^*(2^{0.311n})$ elements, from which a solution can be recovered using Meet-in-the-Middle.\footnote{Several challenges arise in making these ideas precise, including efficiently simulating the hash function and proving that randomness in the input results in candidate solutions that are favorably distributed with high probability. For a more detailed introduction to the representation technique for Subset Sum, see \cite[Section~2.2]{becker2011improved}.}

For $C$ larger than $\{0, 1\}$, the representation approach considers partial candidate solutions that assign nonzero coefficients to some input elements. For example, \cite{mucha2019equal} consider $C = \{0, \pm 1\}$ and partial candidate solutions that pick out approximately $n / 6$ coefficients matched with each of $1$ and $-1$. However, this approach works only for yes-instances with a significant number of $0$ coefficients. We show how to circumvent the problem for general $C$ by ``translating'' GSS instances so that any $c \in C$ can play the role of the $0$ coefficient. 

At a high level, since instances in which $M > |C|^{(1 + \eps)n}$ are extremely unlikely to have a solution, we can simply ignore such instances. On the other extreme, we show how to ``shrink'' instances for which $M < |C|^{(1-\eps)n}$ for any constant $\eps$ while preserving the existence of a solution. (The ``shrinking'' also gives an improved runtime for our algorithm). Thus, we can focus on a narrow range of values for $M$.  For such instances, we simulate a hash based on the \emph{signature}, or partial sum, of a partial candidate solution in such a way that the two halves of a genuine solution are guaranteed to hash to the same value. We prove a result (\Cref{lem:signature-distribution}) which guarantees that with high probability many partial solutions have distinct signatures and thus removes the need for any heuristic assumptions. Finally, we use a dynamic programming approach to simulate the signature-based hash, sample elements, and recover a solution. \Cref{sec:algorithm} presents our algorithm with preprocessing, correctness and runtime proofs.

% In particular, we can ignore instances for which $M \gg |C|^n$ as these are unlikely to admit solutions. When $M \ll |C|^n$, we describe a ``discarding'' procedure that allows us to effectively reduce the problem size. This involves fixing coefficients for a constant fraction of the input and considering the resulting subproblem, which remains large enough that a solution exists with high probability. We also employ the insight that the coefficient sets of GSS instances can be translated without changing the instance structure. The algorithm itself is a variant on the \emph{representation technique}, which artificially expands the search space by decomposing candidate solutions in such a way that every solution corresponds to multiple decomposed representatives in a new search space. (For an introduction to the representation technique for Subset Sum, see \cite{becker2011improved} or \cite{nederlof2020improving}.) \Cref{sec:algorithm} summarizes the ideas behind our algorithm in more detail before the formal presentation with correctness and runtime proofs.
% In particular, we decompose candidate solutions into pairs of \emph{half-partitions}, each of which assigns coefficients to roughly half of the input elements. We then simulate hashing these half-partitions in a way that narrows the search space while maintaining solution pairs.

\subsection{Organization}
\Cref{sec:notation} establishes useful notation and preliminaries.  
Our structural results are in \Cref{sec:structural} and our algorithmic results are in \Cref{sec:algorithm}.
Finally, \Cref{sec:applications} gives average-case reductions from GSS to (Generalized) Number Balancing as well as versions of our structural results in this setting.

% \rnote{ (Tim paraphrase): We should discuss each of the four items in the following footnote, but should just integrate it with our explanation above to make for a smoother flow.}
% \trnote{ Should definitely make sure we cover: MiM, Shroppel-shamir and recent space/time improvements, HGJ BCJ, and the potential for average case algorithms; Mucha et al and ESS, Structural theorems for partition (BCJ, Luecker). }

% Salvaged the following paragraph from an earlier draft. One line of research, culminating in \cite{borgs2001phase}, has focused on the moments of the solution distribution of the integer partition problem and how they vary with input parameters. Study of the Number Balancing problem has focused on minimizing the offset (or \emph{discrepancy}) of a solution efficiently. Worst-case Subset Sum is a classic NP-complete problem that still receives significant attention from researchers interesting in finding an algorithm that runs in time $o(2^{n/2})$ (see \cite{austrin2013space}, \cite{austrin2015dense}, \tr{also others}). In contrast, there exist average-case algorithms for Subset Sum that break this barrier, although some of the best rely on heuristic assumptions (see \cite{howgrave2010new}, \cite{becker2011improved}, and \cite{esser2019better}).}

\section{Notation and Preliminaries}
\label{sec:notation}

Given an $n$-dimensional vector $\vec{x}$ and $T\subseteq [n]$, we write $\vec{x}_T$ to denote the $|T|$-dimensional vector obtained as the projection of $\vec{x}$  onto indices $i \in T$.

When written without a specified base, $\log(\cdot)$ denotes the base-2 logarithm.

\vspace{.1in}
\noindent
\textbf{Parameter assumptions. } In the proofs below, we consider $M = |C|^{\Omega(n)}$. When $M = |C|^{o(n)}$, standard dynamic programming techniques solve GSS in subexponential time.

\vspace{.1in}
\noindent
\textbf{Big-O notation.} An asterisk added to big-$O$ notation ($O^*, \Omega^*$, and $\Omega^*$) indicates the suppression of factors polylogarithmic in the argument of the function. For example, $O^*(n)$ ignores $\polylog(n)$ factors and $O^*(2^n)$ ignores $\poly(n)$ factors.

\vspace{.1in}
\noindent
\textbf{Set notation.} We write $[a:b]$ for the set of integers $\{a, a+1, \dots, b\}$. This notation is simplified to $[a]$ for $\{1, 2, \dots, a\}$. We write $\Sigma(S)$ as shorthand for the sum $\sum_{s \in S} s$, $\alpha S$ as shorthand for the multiset $\{\alpha s : s \in S\}$, and $S + \alpha$ as shorthand for the multiset $\{s + \alpha: s \in S\}$. For $\eps \in (0, 1)$ and $b \geq 1$, we write $a \in b^{1\pm\eps}$ to indicate {$a \in [b^{1-\eps}, b^{1+\eps}]$.}
    
\vspace{.1in}
\noindent
\textbf{Probability notation.} Random variables like $\bm{x}_{i}$ are written in boldface. For a finite set of vectors $S$, we write  ``$\vec{\bx} \sim S$'' to indicate that vector $\vec{\bx}$ is sampled uniformly at random from $S$.

\vspace{.1in}
\noindent
\textbf{Combinatorics.} Recall the definition of multinomial coefficients:
\[
    \binom{n}{\alpha_1 n, \alpha_2 n, \dots, \alpha_m n} := \frac{n!}{(\alpha_1 n)!(\alpha_2 n)!\dots (\alpha_m n)!}
    %\yjnote{``\$\{... {\textbackslash}choose ...\}\$'' is outdated and may violate other packages. I modify it to ``\${\textbackslash}binom\{...\}\{...\}\$''.}
\]
for nonnegative $\alpha_1, \alpha_2, \dots, \alpha_m$ such that  $\sum_{i \in [m]} \alpha_i = 1$ and $\alpha_i n$'s are integers for all $i$. Stirling's approximation tells us that $n! = \Theta^*(n^ne^{-n})$. Substituting yields the helpful approximation
\begin{align}
    \label{eq:stirling}
    \binom{n}{\alpha_1 n, \alpha_2 n, \dots, \alpha_m n}
    = \Theta^*\Big(\prod_{i \in [m]} \alpha_i^{-\alpha_i n}\Big)
    = \Theta^*\big(2^{H(\alpha_1, \alpha_2, \dots, \alpha_m)n}\big),
\end{align}
where $H(\alpha_1, \alpha_2, \dots, \alpha_m) := -\sum_{i \in [m]} \alpha_{i} \log_{2}(\alpha_{i})$ denotes the entropy function. We write $H(\alpha)$ for the binary entropy function $H(\alpha, 1-\alpha)$.

%Most solutions to GESS instances will be \emph{balanced}: that is, each coefficient will be associated with approximately the same number of elements. We introduce the following definitions that will be helpful later.
%\begin{definition}
%We say that a partition of $[n]$ into $k$ sets is $\alpha$-balanced if each constituent set has cardinality at most $\frac{1 + \alpha}{k}n$.
%\end{definition}

% Likewise, we say that a candidate vector $c \in C^n$ is $\alpha$-balanced if each \emph{nonzero} coefficient appears at most $\frac{1 + \alpha}{|C|}n$ times in $c$. We say that a yes-instance of GESS on $C$ is $\alpha$-balanced if it admits an $\alpha$-balanced solution.

\section{Structural Results}
\label{sec:structural}

    This section presents our structural results, which characterize when average-case GSS instances are likely to have a solution in terms of $M$, $n$, $\tau$ and $|C|$. 
    
    %\trnote{ I was torn between talking about the techniques used to prove structural results here, as opposed to in the 'Our Techniques' section above. Currently we do a longer exposition of our algorithmic approach in Section 4 rather than 'Our Techniques' - whatever we do, we should probably try to make it consistent. }
    
    % \yj{Tim, we use different notations for the target in \Cref{sec:structural,sec:algorithm} ($\tau$) versus \Cref{apx:discard} ($t$). Maybe replacing all the $\tau$'s with $t$ is more convenient. Have we ever used $t$ for other purposes?}
    
    \subsection{ The \texorpdfstring{$C = C(1)$}{} Case }
    
    The $C = C(1)=\{\pm 1\}$ case has been studied in  previous work by Borgs, Chayes and Pittel \cite{borgs2001phase}. Their results precisely determine  the parameters of a phase transition within a subexponential window around $2^n / \sqrt{n}$ and prove that solutions exist with probability $1 - o_n(1)$ for $M$ below the window and $o_n(1)$ above. For our purposes, precisely pinning down the phase transition window is less important than establishing regions within which solutions are either very likely or very unlikely to exist, so we present a corollary of their analysis (\cite[Theorem~2.1]{borgs2001phase}) which features a larger window but more flexibility in the offset. 
 The proof can be found in \Cref{apx:corollary}.
    % \begin{corollary}[GSS Solution Probability on $C(1)$; see \cite{borgs2001phase} Theorem 2.1]
    %     \label{corr:optprecisionC1}
    %     Let $C = C(1) = \{\pm 1\}$, and fix any $\eps > 0$. Let $\vec{\bx} = (\bx_{i})_{i \in [n]}$ be uniformly sampled from $[0:M-1]^n$. For any integer target $|\tau| = o(Mn)$, the probability $\Pr_{\vec{\bx}}[\exists \vec{c} \in C^n : |\vec{\bx} \cdot \vec{c}| \in \{\tau, \tau+1\}]$
    %     \begin{itemize}
    %         \item is equal to $1 - o_n(1)$ when $M = O^*(|C|^{(1-\eps)n})$; and
            
    %         \item is at most $2|C|^n / M$ when $M \geq |C|^n$.
    %     \end{itemize}
    %     Moreover, $|\vec{\bx} \cdot \vec{c}| = s$ is possible if and only if $s$ and the 1-norm $|\vec{\bx}|_1$ have the same parity, so the parity of $|\vec{\bx}|_1$ determines a single possible target in $\{\tau, \tau+1\}$.
    % \end{corollary}
    
    \begin{corollary}[GSS Solution Probability on $C(1)$ \cite{borgs2001phase}]
        \label{corr:optprecisionC1}
        Let $C = C(1)$ and fix any positive~constant $\eps > 0$. For $\vec{\bx} = (\bx_{i})_{i \in [n]} \sim [0:M-1]^n$ and any integer $\tau$ satisfying $|\tau| = o(Mn)$, we~have
        \[
            \Prx_{\vec{\bx}}\Big[\exists\hspace{0.05cm} \vec{c} \in C^n : \vec{\bx} \cdot \vec{c} \in \{\tau, \tau+1\}\Big]\hspace{0.1cm}
            \begin{cases}
                 \ge 1 - o_n(1) & \text{if~} M \le |C|^{(1-\eps)n} \\[0.3ex]
                 \le {2|C|^n}\big/{M}     & \text{if~} M \geq |C|^n.
            \end{cases}
        \]
    \end{corollary}
    
      Note that $\vec{\bx}\cdot \vec{c}$ has the same parity as
      $\sum_{i\in [n]}\bx_i$ for every $\vec{c}\in \{\pm 1\}^n$. So 
      the parity of $\sum_{i\in [n]}\bx_i$ determines a single possible target in
      $\{\tau,\tau+1\}$.
%      $|\vec{\bx} \cdot \vec{c}| = s$ is possible if and only if $s$ and the 1-norm $|\vec{\bx}|_1$ have the same parity, so the parity of $|\vec{\bx}|_1$ determines a single possible target in $\{\tau, \tau+1\}$.    
    %\trnote{For me, the second version is still easier to read. It does have the issue Yaonan mentioned with mixing big-$0$ notation and inequalities.s Switch or not?}

    \subsection{ The \texorpdfstring{$C = C(d)$, $d > 1$}{} Case }

    The proof of \cite{borgs2001phase} can be further extended to the case of $C(d)$ with a 
      general $d>1$. We prove the following theorem in \Cref{apx:discard}:

    \begin{theorem}[GSS Solution Probability for $C(d)$, $d > 1$]
        \label{thm:optprecisionC}
        Let $C=C(d)$ for a fixed integer $d > 1$ and fix any constant $\eps>0$.
        %Fix any $d>1$ and any $\eps >0$. %, and let $C$ be the symmetric coefficient set $C = C(d).$ 
        For $\vec{\bx}  \sim [0:M-1]^n$ and any  integer $\tau$ satisfying 
        $|\tau| = o(Mn)$, we have  
        \[
            \Prx_{\vec{\bx}}\Big[\exists\hspace{0.05cm} \vec{c} \in C^n : \vec{\bx} \cdot \vec{c} = \tau\Big]\hspace{0.1cm}
            \begin{cases}
                 \ge 1 - o_n(1) & \text{if~} M \le |C|^{(1-\eps)n} \\[0.3ex]
                 \le  {|C|^n}\big/{M}     & \text{if~} M \geq |C|^n.
            \end{cases}
        \]
    \end{theorem}
    
    Generalizing the proof strategy of \cite{borgs2001phase} to $C(d)$ significantly complicates the analysis but does not substantially change the underlying intuition. Notably, \cite{borgs2001phase} confronted a parity issue in their case of $C(1)$: 
    no solution summing to a target $\tau$ that has parity different from the sum of inputs
    is possible. Our analysis explains the parity issue as a result of constructive interference in the integrand of the solution-counting function and demonstrates that no such issues occur for $C(d)$ when $d > 1$
    (compare the $\vec{\bx}\cdot \vec{c}\in \{\tau,\tau+1\}$ in \Cref{corr:optprecisionC1} and $\vec{\bx}\cdot \vec{c}=  \tau$ in \Cref{thm:optprecisionC}).
    
    This proof approach has the advantage that it neatly bounds the first and second moments of the number of solutions, not just the probability that a solution exists. However, the probability that a solution exists in the $ M\le |C|^{(1-\eps)n}$ regime will correspond to the failure probability of our algorithm later, so it is desirable to have a sharper bound than the $1 - o_n(1)$ given by \Cref{corr:optprecisionC1} and  \Cref{thm:optprecisionC}.  We achieve this in the $C = C_0(d)$ case by using a new approach described in the next subsection.
    
    \subsection{ The \texorpdfstring{$C = C_0(d)$}{} Case }

    When the coefficient set contains $0$, we analyze the probability that a solution exists by considering a process in which a set of ``achievable targets'' is grown as input numbers $\bx_1,\bx_2,\dots$ are revealed one by one. Once the set of achievable targets contains $\tau$, the existence of a solution is ensured since any remaining input elements can be assigned the $0$ coefficient. The structural result that we prove in this way is analogous to \Cref{thm:optprecisionC}, but with exponentially small failure probability for small $M$. Moreover, the argument is shorter, more elementary, and more intuitive.
    
    \begin{theorem}[GSS Solution Probability for $C = C_0(d)$]
        \label{thm:optprecisionC0}
    Let $C=C_0(d)$ for some fixed $d\in \N$, and fix any constant $\eps >0$.
    %, and let $C$ be the symmetric coefficient set $C = C_0(d).$ 
    For $\vec{\bx} \sim [0:M-1]^n$ and any integer $\tau$  satisfying $|\tau|=o(Mn)$, we have
       %\xnote{Should we use $\bx\cdot c$ instead of $|\bx\cdot c|$ in this lemma? Otherwise it should be $2|C|^n/M$ I think.}
        \[
            \Prx_{\vec{\bx}}\Big[\exists\hspace{0.05cm} \vec{c} \in C^n : \vec{\bx} \cdot \vec{c} = \tau\Big]\hspace{0.1cm}
            \begin{cases}
                 \ge 1 - e^{-\Omega(n)} & \text{if~} M  \le |C|^{(1-\eps)n} \\[0.3ex]
                 \le {|C|^n}\big/{M}     & \text{if~} M \geq |C|^n.
            \end{cases}
        \]
    \end{theorem}

    % \orange{
    %     We observe that when $ndM \leq (d+1)^n$, for any input $\vec{x}$ there exist two distinct solutions $\vec{c}_1, \vec{c}_2 \in [0:d]^n$ such that $\vec{c}_1 \cdot \vec{x} = \vec{c}_2 \cdot \vec{x}$ by the pigeonhole principle. This implies the existence of the solution $\vec{c}_1 - \vec{c}_2 \in C^n$. Our result extends this trivial case and shows that solutions are very likely as long as $M = O^*((2d+1)^{(1-\eps)n})$ for some $\eps > 0$.
    % } \trnote{We may want to take this out, because I don't think the pigeonhole principle argument works for general $\tau$. For example, the input $\vec{x} = \vec{0}$ can't be used to make a nonzero target.}

\subsection{Proof of \texorpdfstring{\Cref{thm:optprecisionC0}}{} }
\label{subsec:proof_C0}

We first give the simple upper bound  in the case when $M \geq |C|^n$, which is just a union bound.~Any fixed {\em not-all-zero} $\vec{c} \in C^{n}$ has $c_{i} \neq 0$ for some $i \in [n]$. Conditioning on any outcome of $\smash{\vec{\bx}_{[n] \setminus \{i\}}}$ from $[0: M - 1]^{n - 1}$, $\vec{\bx} \cdot \vec{c} = \tau$ holds for at most one outcome of $\bx_{i}$ and thus occurs with probability at most $1/M$. Union-bounding over all $\vec{c} \in C^{n}$ gives the claimed upper bound of $|C|^n / M$.
    
In the rest of the section, we focus on the case when $M\le |C|^{(1-\eps)n}$ for some positive~constant $\eps>0$. Let $d\in \N$ be a constant with $C=C_0(d)$,
  and let $\tau$ be the target integer with $|\tau|=o(Mn)$.
For convenience we assume in the proof that $\tau\ne 0$; it will become clear at the end
  that the same proof works for $\tau=0$.
We start with some notation and a high-level overview of the proof.

Given $\vec{x} \in [0:M-1]^{\ell}$,    
%  by $\vec{x}_{[\ell]} := ( x_{i})_{i \in [\ell]}$ %and $\vec{c}_{[\ell]} := (c_{i})_{i \in [\ell]}$ 
%  the projections of $\vec{x}$ %and $\vec{c}$ 
%  onto the first $\ell$ indices. 
we define the set $\calS(\vec{x} )$ of \emph{targets achievable with $\vec{x} $} to be
\begin{align*}
     \calS(\vec{x} ) := \big\{ \vec{x} \cdot \vec{c} : \vec{c} \in C^{\ell}\big\}.
\end{align*}
When $\ell=0$ and $\vec{x}$ is the empty vector, $\calS(\vec{x})=\{0\}$ by default.
Given $\vec{x}\in [0:M-1]^n$,
  note that $$\calS(\vec{x}_{[1]})\subseteq \calS(\vec{x}_{[2]})\subseteq  \cdots \subseteq \calS(\vec{x}_{[n]})$$ and %$\bm{S}_\ell$ depends on $\vec{\bx}_{[\ell]}$ and that 
$\vec{x}$ admits a solution if and only if $\calS(\vec{x}_{[n]})=\calS(\vec{x})$ contains $\tau$.\footnote{This is the reason why we assumed $\tau\ne 0$; when
  $\tau=0$ we need the extra condition that $0\in \calS(\vec{x})$ can be
  obtained by a not-all-zero $\vec{c}$.} 
Let 
$$
m:=\left\lceil \tau n + \log_{|C|} M\right\rceil
\quad\text{and}\quad
m':=\left\lceil \left(1-\frac{\eps}{3}\right)n\right\rceil,\quad\text{with
  $\tau=\frac{\eps^2}{256d^2\ln |C|}$.}
$$
Using $\log_{|C|}M\le (1-\eps)n$, we have $m'-m>
 \eps n/3$.
%\yjnote{\label{yjfootnote:0}
% In view of \Cref{yjfootnote:1,yjfootnote:2}, we may need to modify this place; but anyway, we still have $m'-m> \eps n/3$.}
  
  Our proof considers the evolution of the set $\bm{S}_{ \ell}:=\calS(
  \bx_1,\ldots,\bx_\ell)$ as we draw each input element $\bx_\ell\sim [0: M - 1]$ for $\ell=1,\ldots, n$ sequentially.
%Our proof of the $1 - e^{-\Omega(n)}$ lower bound for $M = O^*(|C|^{(1-\eps)n})$ proceeds in four steps.
It proceeds in three steps:
\begin{flushleft}\begin{enumerate}
%    \item Show that for any fixed GSS instance $\vec{x} \in [0: M - 1]^{n}$ and any $\ell \in [n]$, almost all achievable targets have rather small magnitude (\Cref{lem:bound_badsolns}).
    \item Show that with high probability over the draw of
    the first $m$ elements, %there exists some $m = n - \Omega(n)$ such that 
    $\bm{S}_{m}$ occupies a constant fraction of some length-$M$ interval not too far from the origin. (\Cref{lem:main}).
    
    \item Assuming the event in item 1 happens, show that with high probability over the draw of the next $m'-m$ elements, %there exists some $m' = n - \Omega(n)$ such that 
    $\bm{S}_{m'}$ occupies a constant fraction of a length-$M$ interval that  contains $\tau$. (\Cref{lem:lower_bound_on_sum}). 
    %We do this by ``translating'' the interval from the previous step.
    
    \item Assuming the event in item 2 happens, show that $\tau\in \bm{S}_n$ with
    high probability over the draw of the last $n-m'$ elements. 
    %Analyze the draw of the remaining elements $\bx_{m'+1},\dots,\bx_n$ to show that 
\end{enumerate}\end{flushleft}
Among these three steps, the first step is the most challenging.

Before stating \Cref{lem:main}, we give one more definition about
  what we meant by ``not too far from the origin.''
Given $\vec{x}\in [0:M-1]^\ell$, we define  
%We now proceed to the detailed arguments. For each $\ell \in [n]$ and any constant $\delta > 0$, we define 
the set of  \emph{large} targets achievable with $\vec{x}$ as
\[
   {\calB}(\vec{x}) := \Big\{w \in {\calS}(\vec{x}): |w| \geq \eps n M\big/8\Big\}.
\]
%and we say that an element of $\bm{S}_\ell$ that is not $\delta$-large is \emph{$\delta$-small}.
The following simple claim shows that 
  $\calB(\vec{x})$ cannot be too large:
%only a small fraction of the $|C|^\ell$ candidate solutions (choices of coefficients for $\bx_1,\dots,\bx_\ell$) can be large:

\begin{claim}%[Upper Bounding the Number of Large Candidate Solutions]
\label{lem:bound_badsolns}
For any $\ell\in [0:n]$ and $\vec{x}\in [0:M-1]^\ell$, we have
$$
    |\calB(\vec{x})|\le |C|^\ell\cdot 2\exp\left(-\frac{\eps^2 n}{128d^2}\right)
    =|C|^\ell\cdot 2|C|^{-2\tau n}.
    $$
%Fix any $d \in \N$ and let the coefficient set $C$ be $C= C_0(d).$ For any outcome of $\vec{\bm{x}}_{[\ell]} = (\bm{x}_{i})_{i \in [\ell]} \sim [0:M-1]^\ell$ and any constant $\delta > 0$, we have that
%\begin{align*}
 %   |\bm{B}_{\ell}(\delta)| \leq 2 \cdot |C|^{\left(1 %-\frac{2\delta^{2}}{|C|^2 \ln |C|}\right) \cdot \ell}.
%\%end{align*}
\end{claim}

\begin{proof}
The case with $\ell=0$ is trivial given $\calB(\vec{x})=\emptyset$.
%    Fix any outcome $x_1,\dots,x_\ell$ and 
For $\ell\ge 1$,
let independent random variables $\by_1,\dots,\by_\ell$ be such that $\Pr[\by_i = c {x}_{i}]$ = $ {1}/{|C|}$ for each $c \in C$. Letting $\by = \sum_{i \in [\ell]} \by_i$, we have 
    \begin{align*}
        \frac{|\calB(\vec{x})|}{|C|^\ell}\le 
%        = \Big|\Big\{ w \in \bm{S}_{\ell}: |w| \geq \delta n M \Big\}\Big|
 %       \leq \Big|\Big\{ w \in \bm{S}_{\ell}: |w| \geq \delta \ell M \Big\}\Big|
  %      \leq |C|^{\ell} \cdot 
  \Prx_{\by}\Big[|\by| \geq \eps nM\big/8\Big]
    \end{align*}
and the claim follows from an application of the
  Hoeffding inequality (and $\ell\le n$).
\end{proof} 

We prove \Cref{lem:main} for the first step:

\begin{lemma}[Dense Interval Lemma]
\label{lem:main}
With probability $1 - e^{-\Omega(n)}$
  over $\vec{\bx} \sim [0:M-1]^m$, 
  $\calS(\vec{\bx})\setminus \calB(\vec{\bx})$ occupies at least $\gamma$-fraction
  of some length-$M$ interval with 
  \begin{equation}\label{eq:gamma}
  \gamma := \frac{\tau^2}{48|C|^2},
  \end{equation}
  i.e.,
  there exists a length-$M$ interval
  $I$ such that 
$ 
 |I\cap(\calS(\vec{\bx})\setminus \calB(\vec{\bx})) |
\ge \gamma M. %,\quad\text{with $\gamma := \eps^{2} / (96 |C|^{2})$.}
$ % (i.e., 
%$$\bm{S}_{m} \setminus \bm{B}_{m}
%occupies at least a constant fraction  of some length-$M$ interval of integers.
    % \rnote{Was ``some length-$M$ interval has a constant density $\gamma := \eps^{2} / (96 |C|^{2})$ in the set of $\delta$-small candidate solutions $\bm{S}_{m} \setminus \bm{B}_{m}(\delta)$.'' but this was a little confusing to my ear.  We frequently use the phrase ``interval $I$ has density $\tau$ in set $S$'' --- I think what we mean by this is that $|S \cap I| \geq \tau |I|$, right? The phrase doesn't sound to me like it means that, so I rephrased throughout. \yj{The current representation sounds better to me : )}}
\end{lemma}

\begin{proof}
%Consider any outcome $\vec{x}_{[\ell]}\in [0:M-1]^\ell$ of $\vec{\bx}_{[\ell]}$,
%Let $\vec{y} \in [0:M-1]^{\ell-1}$ for some $\ell\in [m]$ where 
%  $\vec{y}$ denotes the empty vector when $\ell=1$ with 
%  $\calS(\vec{y})=\calB(\vec{y})=\{0\}$ by default.
%We say $\vec{y}$ is \emph{good} if 
%    For each $\ell \in [0: n]$, let $\bm{\X}_{\ell}$ be an indicator random variable that is $1$ if
%\begin{enumerate}
%$\calS(\vec{y}) \setminus \calB(\vec{y})$ occupies at least  $\gamma$-fraction of some length-$M$ interval;
%    \item or the set of large targets achievable is large:
    %has size at least 
%    $$
%    |\calB(\vec{y})|\ge |C|^m\cdot \exp\left(-\frac{\eps^2 n}{32d^2}\right).
%    $$
%    \end{enumerate}
%    \ignore{some length-$M$ interval has density $\gamma$ in $\bm{S}_{\ell} \setminus \bm{B}_{\ell}(\delta)$} 
%otherwise, $\vec{y}$ is \emph{bad} ($\vec{y}$ is trivially
%  bad when it is the empty string with $\ell=1$).  
 
    We use the following claim to establish the lemma:
\begin{claim}\label{claim:useful}
Let $\ell\in [m]$ and let $\vec{y}\in [0:M-1]^{\ell-1}$ be a vector such that 
  $|I\cap (\calS(\vec{y})\setminus \calB(\vec{y}))|<\gamma M$ for any length-$M$ interval
  $I$.
When $\bm{z}\sim [0:M-1]$, we have
$$
\frac{|\calS(\vec{y}\circ \bm{z})|}{|\calS(\vec{y})|}\ge (1-\tau/4)|C|
$$
with probability at least $1-\tau/4$.
\end{claim}    
 
We delay the proof of the claim and use it to prove the lemma first.    
%    , it suffices to show that
%    \[
%        \Prx_{\vec{\bx}_{[m]}}\left[\bm{\X}_{m} = 1\right] = 1 - e^{-\Omega(n)}.
%    \]
    Consider the experiment~of drawing $\bm{x}_1, \bm{x}_2, \dots, \bm{x}_m\sim [0:M-1]$ in turn, with $\bm{S}_\ell=\calS(\bm{x}_1,\ldots,\bm{x}_\ell)$ 
    and $\bm{B}_\ell=\calB(\bm{x}_1,\ldots,\bm{x}_\ell)$ for each $\ell\in [0:m]$
    (so $\bm{S}_0=\{0\}$ and $\bm{B}_0=\emptyset$).
    For each $\ell\in [m]$, let $\bm{\X}_\ell$ denote the indicator random variable that is set to $1$ if either 
    $|I\cap (\bm{S}_{\ell-1}\setminus \bm{B}_{\ell-1})|\ge \gamma M$ for some length-$M$ interval $I$ or %(\bm{x}_1,\ldots,\bm{x}_{\ell-1})$ is good
    %or 
\begin{equation}\label{hehehaha}
\frac{|\bm{S}_\ell|}{|\bm{S}_{\ell-1}|}\ge (1-\tau/4)|C|.
\end{equation}
Then conditioning on any outcome of 
  $(\bm{x}_1,\ldots,\bm{x}_{\ell-1})$ it follows from \Cref{claim:useful} that
  the probability~of $\bm{\X}_\ell=1$ is at least $1-\tau/4$.
By Chernoff bound we have the probability of $\smash{\sum_{\ell\in [m]} \bm{\X}_\ell\ge (1-\tau/2)m}$ is  $1-e^{-\Omega(n)}$.
We show that when this occurs  
  we must have $\smash{|I\cap (\bm{S}_\ell\setminus \bm{B}_\ell)|\ge \gamma M}$ for some
  length-$M$ interval $I$ for some $\ell\in [m]$, which implies the 
  same for $\bm{S}_m\setminus \bm{B}_m$ given that $\bm{S}_\ell\subseteq \bm{S}_m$.
  
To this end we assume for a contradiction that this is not the case
  for any $\ell\in [m]$.
Then every $\bm{\X}_\ell=1$ implies (\ref{hehehaha}).
Given that we always have $\bm{S}_{\ell-1}\subseteq \bm{S}_\ell$,
  we have 
\begin{align*}
        |\bm{S}_{m}|
         \geq ( (1-\tau/4)|C|)^{(1 - \tau/2)m}  
        %\cdot (1 - \eps / 4)^{(1 - \eps / 4) \cdot m}
          \ge |C|^m |C|^{-(\tau/2)m} e^{-(\tau/2)m} 
         \ge M\cdot |C|^{\tau n- \tau m}
         =M\cdot 2^{\Omega(n)},
        %\nonumber \\
%        & \geq |C|^{(1 - \eps / 4) \cdot m} \cdot e^{-(\eps / 4) \cdot m}
%        \nonumber  \\
%        & \geq |C|^{(1 - \eps / 4) \cdot (1 - \eps %/ 2) \cdot n} \cdot e^{-(\eps / 4) \cdot n}
%        \nonumber  \\
%        & \geq |C|^{(1 - \eps) \cdot n} \cdot \Big(\frac{|C|}{e}\Big)^{(\eps / 4) \cdot n}
%        \nonumber  \\
%        & = M \cdot e^{\Omega(n)}
%        \nonumber \\
%        & \geq M \cdot (2\delta n + |C|) \cdot \gamma.
%        \label{eq:main:S}
    \end{align*}
Here the second step used $1-\tau/4\ge e^{-\tau/2}$ for our small $\tau$. The third step used $m = \lceil \tau n + \log_{|C|} M \rceil$ and $e < 3 \le |C|$.
%On the other hand, by using the Hoeffding inequality and
%  a union bound, we have the probability of 
%  $\bm{B}_\ell$ satisfying (\ref{}) for some $\ell$
%  is $e^{-\Omega(n)}$.
On the other hand, by \Cref{lem:bound_badsolns},
% \yjnote{\label{yjfootnote:1}
% Confused with the following calculation, given $m = \lceil \tau n + \log_{|C|} M \rceil$. For example, when $|C| = 9 > e^{2}$, we have $|C|^{m} \cdot 2e^{-2\tau n} \geq M \cdot 9^{\tau n} \cdot 2e^{-2\tau n} > M$.}
% \yjnote{\label{yjfootnote:2}
% Instead, we require that $|\bm{S}_{m}| \geq ((1-\tau/4)|C|)^{(1 - \tau/2)m} \red{\gg} |C|^{m} \cdot 2e^{-2\tau n} \geq |\bm{B}_{m}|$. \\
% Because $(1 - \tau / 4)^{1 - \tau / 2} \geq e^{-\tau / 4}$ (\yj{this relaxation does not loss much}), it suffices to have $e^{-(\tau / 4) m} \cdot |C|^{-(\tau / 2) m} \gg 2e^{-2\tau n}$, which holds when $m < \frac{8n}{1 + 2\ln |C|}$. But in this regime, for $|\bm{S}_{m}|$ we can just prove a lower bound of $((1-\tau/4)|C|)^{(1 - \tau/2)m}$. Suppose, for example, that $|C| \geq e^{8}$ and $\eps = 1 / 2$, then this lower bound is at most $|C|^{m} \leq |C|^{\frac{8n}{1 + 2 \cdot 8}} < |C|^{(1 - \eps)n} = M$. \\
% Hence the lower bound $((1-\tau/4)|C|)^{(1 - \tau/2)m}$ is not strong enough to derive the contradiction $|\bm{S}_m\setminus \bm{B}_m|\ge M\cdot 2^{\Omega(n)}$...}
% we can choose $m = \lceil (\log_{|C|} M) + \frac{\tau n}{\ln|C|} \rceil$; then we have $|\bm{S}_{m}| \ge M \cdot e^{\tau n - \tau m} = M \cdot 2^{\Omega(n)}$ and (when $n$ is large enough) $|\bm{B}_m|\le M \cdot e^{\tau n + 1} \cdot 2e^{-2\tau n} < M$.}
$$|\bm{B}_m|\le |C|^m\cdot 2|C|^{-2\tau n}<M.
$$
As a result, we have $|\bm{S}_m\setminus \bm{B}_m|\ge
M\cdot 2^{\Omega(n)}$, a contradiction with
the trivial upper bound of $\eps nM/4$ for $|\bm{S}_m\setminus \bm{B}_m|$ by definition.
We prove Claim \ref{claim:useful} in the rest of the proof:

\begin{proof}[Proof of Claim \ref{claim:useful}]
%    Each step $\ell \in [m]$ of this experiment (together with the previous steps) determines $\bm{S}_{\ell}$, $\bm{B}_{\ell}(\delta)$ and $\bm{\X}_{\ell}$. Note that if $\bm{\X}_{\ell - 1} = 1$ for some $\ell \in [m]$, then clearly we must have $\bm{\X}_{\ell} = 1$ because (by considering $c_{\ell} = 0$) $\bm{S}_{\ell} \setminus \bm{B}_{\ell}(\delta)$ is a superset of $\bm{S}_{\ell - 1} \setminus \bm{B}_{\ell - 1}(\delta)$.\ignore{ and so to establish the lemma it is enough to show that some $\ell \leq m$ has
%    \[
%        \Prx_{\vec{\bx}_{[\ell]}}\Big[\bm{\X}_{\ell} = 1\Big] = 1 - e^{-\Omega(n)}.
%    \]
%    }
Fix an $\ell\in [m]$ and $\vec{y}\in [0:M-1]^{\ell-1}$
  such that $|I\cap ( \calS(\vec{y})\setminus \calB(\vec{y})|<\gamma M$ for 
  any length-$M$ interval $I$.
%    Fix an $\ell \in [m]$ and fix any specific outcome of $\vec{\bx}_{[\ell - 1]} \in [0: M - 1]^{\ell - 1}$. We assume for now that $\bm{\X}_{0} = \dots = \bm{\X}_{\ell - 1} = 0$.
    % \rnote{This was $\bm{\X}_{\ell}$ but I think we mean $\bm{\X}_{\ell-1}$, right?}
Let $\bm{z}\sim [0:M-1]$, and we 
  consider the quantity $ \smash{|C| \cdot |\calS(\vec{y})| - |\calS(\vec{y}\circ \bm{z})| }$ as a random variable.  It is easy to see that this random variable is nonnegative (as each element in~$\calS(\vec{y})$ gives rise to at most $|C|$ elements in $\calS(\vec{y}\circ\bm{z})$).
    %, corresponding to the $|C|$ possible coefficient choices for $c_\ell$).
 We will show that 
 % \rnote{Since we said earlier that the random variable was just over the $\bx_\ell$-randomness, and we told the reader earlier to assume that $\bm{\X}_{\ell-1} = 0$, I think we should not have the equation indicate a conditioining on $\bm{\X}_{\ell-1} = 0$.}
    \begin{align}
        \label{eq:no_collision}
        \Ex_{\bm{z}}\Big[|C| \cdot |\calS(\vec{y})| - |\calS(\vec{y}\circ\bm{z})| \ignore{\Bigm| \bm{\X}_{\ell - 1} = 0}\Big]
        \leq |\calS(\vec{y})| \cdot \Big(2|C| M \cdot \gamma +| \calB(\vec{y})|\Big) \cdot \frac{|C|^2}{M}.
    \end{align}
    \ignore{The nonnegativity holds because $\bm{S}_{\ell - 1}$ contains elements $\vec{\bx}_{[\ell - 1]} \cdot \vec{c}_{[\ell - 1]}$ for all $\vec{c}_{[\ell - 1]} \in C^{\ell - 1}$ and $\bm{S}_{\ell}$ contains elements $\vec{\bx}_{[\ell]} \cdot \vec{c}_{[\ell]}$ for all $\vec{c}_{[\ell]} \in C^{\ell}$.}
     
%\xnote{Not sure if we should write $S_{\ell-1}$ as a random variable here.}    
We show \eqref{eq:no_collision} by upper-bounding the number of ``collisions'' that occur when we derive elements in $\calS(\vec{y}\circ \bm{z})$. Namely, consider the pairs $(w, c)$ for $w \in \calS(\vec{y})$ and $c \in C$. If every two distinct~pairs $(w, c) \neq (w', c' )$ were to evaluate to $w + c \bm{z} \neq w' + c' \bm{z}$, then there is no {\em collision} and we would have $|\calS(\vec{y}\circ\bm{z})| = |C| \cdot |\calS(\vec{y})|$. 
In general $|C|\cdot |\calS(\vec{y})|-|\calS(\vec{y}\circ \bm{z})|$ is at most the number of collisions.

To bound the number of collisions, we observe that a necessary condition 
  for $w+c\bm{z}=w'+c'\bm{z}$ to happen  is
  $|w-w'|\le 2dM$ since $c, c' \in C = \{0, \pm 1, \dots, \pm d\}$ and $\bm{z}\in [0: M - 1]$.
Below we bound the number of pairs $(w,c)\ne (w',c')$ such that $|w-w'|\le 2dM$:
%The following observations together account for \eqref{eq:no_collision}:
    \begin{flushleft}\begin{itemize}
        \item Clearly, there are at most $|\calS(\vec{y})|$ possibilities for $w$ and at most $|C|^2$ possibilities for $c$ and $c'$.
        
        \item %The triple $(w, c, z)$ may collide with another triple $(w', c', z')$ only if $|w - w'| \leq 2d M$, since $c, c' \in C = \{0, \pm 1, \dots, \pm d\}$ and $z,z'\in [0: M - 1]$. 
        By assumption,  $\calS(\vec{y}) \setminus \calB(\vec{y})$ does not occupy a $\gamma$ fraction of any length-$M$ interval.\ignore{so no length-$M$ interval has density $\gamma$ in $\bm{S}_{\ell - 1} \setminus \bm{B}_{\ell - 1}(\delta)$} Hence the interval $[w - 2dM, w + 2dM]$ contains at most $4d M \cdot \gamma$ many  elements in $ \calS(\vec{y})\setminus \calB(\vec{y})$. Moreover, this interval trivially has at most $|\calB(\vec{y})|$ many elements in $\calB(\vec{y})$. So in total there could be at most $4dM \cdot \gamma + |\calB(\vec{y})|$ elements $w' \in \calS(\vec{y})$ such that $|w-w'|\le 2dM$.
        %\item For any $w, w' \in \calS(\vec{y})$ such that $|w-w'|$, there are at most $|C|^2$ pairs of colliding triples $(w, c, \bx_{\ell}) \neq (w', c', \bx_{\ell})$. This is because, for any pair of $c, c' \in C$, the collision $w + c \bx_{\ell} = w' + c' \bx_{\ell}$ happens if and only if $\bx_{\ell} = \frac{w' - w}{c - c'}$.\ignore{(This requires an integer $\frac{w' - w}{c - c'} \in [0: M - 1]$.)} Over the randomness of $\bx_{\ell} \sim [0: M - 1]$, this happens with probability at most $1 / M$ (with equality if $\frac{w' - w}{c - c'}$ is integral and probability 0 otherwise).
    \end{itemize}\end{flushleft}
As a result, the number of pairs $(w,c)\ne (w',c')$ with 
  $|w-w'|\le 2dM$ is at most
$$
|\calS(\vec{y})|\cdot |C|^2\cdot \Big(2|C| M \cdot \gamma + |\calB(\vec{y})|\Big).
$$
Then \eqref{eq:no_collision} follows as each pair leads to a 
  collision with probability at most $1/M$ over $\bm{z}\sim [0:M-1]$.
    
Using \Cref{lem:bound_badsolns} and our choice of $m$, we have 
    \begin{align}
        \label{eq:main:B}
        |\calB(\vec{y})|
        \leq |C|^m\cdot 2|C|^{-2\tau n}
        \le M\cdot 2|C|^{-\tau n}
        %= 2 \cdot |C|^{\frac{|C|^2 \ln |C| - %2\delta^{2}}{|C|^2 \ln |C|} \cdot m}
        %\leq 2 \cdot |C|^{\frac{|C|^2 \ln |C| - 2\delta^{2}}{|C|^2 \ln |C| - \delta^{2}} \cdot (1 - \eps) n}
        < |C| M \cdot \gamma.
    \end{align}
%    where the second step applies \Cref{lem:bound_badsolns}, the third step applies \eqref{eq:main:delta},
    % \rnote{I got a little confused by how this third step followed from \eqref{eq:main:delta}? \yj{For the third step, the $\lhs$ has exponent $(1 - \frac{2\delta^{2}}{|C|^2 \ln |C|}) \cdot m \leq (1 - \frac{2\delta^{2}}{|C|^2 \ln |C|}) \cdot \frac{(1 - \eps) \cdot n}{1 -\frac{\delta^{2}}{|C|^2 \ln |C|}} = \frac{|C|^2 \ln |C| - 2\delta^{2}}{|C|^2 \ln |C| - \delta^{2}} \cdot (1 - \eps) \cdot n = (1 - \frac{\delta^{2}}{|C|^2 \ln |C| - \delta^{2}}) \cdot (1 - \eps) \cdot n$, where the inequality uses \eqref{eq:main:delta}.}}
 %   and the last step holds for any large enough $n \in \NN$ given that $M = O^{*}(|C|^{(1 - \eps)n})$ while $|C|$ and $\gamma$ are constants.
    Combining \eqref{eq:no_collision} and \eqref{eq:main:B} together (and using our choice of $\gamma$ in \eqref{eq:gamma}), we have
    %know that the random variable $(|C| - \frac{|\bm{S}_{\ell}|}{|\bm{S}_{\ell - 1}|})$, conditioned on a fixed value $\vec{x}_{[\ell-1]}$ satisfying $\X_{\ell-1} = 0$,
    %\ignore{(which is random because of the randomness in the outcome of $\bx_\ell$; recall that $\bx_1,\dots,\bx_{\ell-1}$ are assumed to be fixed values)}
    %is always nonnegative and 
    %(recalling that we assume $\bm{\X}_{\ell-1}=0$) 
    %has expectation
    \begin{align*}
        \Ex_{\bm{z}}\bigg[|C| - \frac{|\calS(\vec{y}\circ \bm{z})|}{|\calS(\vec{y})|} \ignore{\biggm| \bm{\X}_{\ell - 1} = 0}\bigg]
        \leq \Big(2|C|M \cdot \gamma + |\calB(\vec{y})|\Big) \cdot \frac{|C|^2}{M}
        \leq 3|C|^{3} \cdot \gamma
        = \frac{|C| \cdot \tau^{2}}{16}.
    \end{align*}
Given that the random variable in the expectation is nonnegative, we have 
  from Markov that
%    Let $\bm{\Y}_{\ell}$ be an indicator random variable (again over the randomness of $\bx_{\ell}$) that is $1$ if $|\bm{S}_{\ell}| / |\bm{S}_{\ell - 1}| < |C| \cdot (1 - \eps / 4)$.
    % Set the constant $\lambda := 8 / \eps$, then we have $3\lambda \cdot |C|^2 \cdot \gamma = 3 \cdot (8 / \eps) \cdot |C|^{2} \cdot \eps^{2} / (96 |C|^{2}) = \eps / 4$.
%    Then Markov's inequality guarantees that for each $\ell \in [m]$, assuming that $\bm{\X}_{\ell-1}=0$, we have that
    \begin{align*}
        %\Prx_{ \bx_{\ell}}\Big[\bm{\Y}_{\ell} = 1 \ignore{\Bigm| \bm{\X}_{\ell - 1} = 0}\Big]
         \Prx_{\bm{z}}\bigg[\frac{|\calS(\vec{y}\circ \bm{z})}{|\calS(\vec{y})|} < (1 - \tau / 4)|C|  \ignore{\biggm| \bm{\X}_{\ell - 1} = 0}\bigg]
         %< \frac{|C| \cdot \eps^{2} / 32}{|C| \cdot \eps / 4}
        \le  \frac{\tau}{4}.
    \end{align*}
The claim follows. %  % 
%    Intuitively, this means that (until $\bm{\X}_{\ell - 1} = 1$) each time we draw a new element $\bx_{\ell}$ for $\ell \in [m]$, the cardinality $|\bm{S}_{\ell}|$ is likely to grow by a constant factor $\approx |C|$. 
\end{proof}    

This finishes the proof of the Dense Interval Lemma.
\end{proof}

Next we prove the lemma for Step 2:
\begin{lemma}[Interval Shifting Lemma]
\label{lem:lower_bound_on_sum}
Let $\vec{x}\in [0:M-1]^m$ be such that $\calS(\vec{x})\setminus \calB(\vec{x})$ occupies at least $\gamma$-fraction of some length-$M$
 interval.
Let $\vec{\by}\sim [0:M-1]^{m'-m}$.
Then with probability at least $1-e^{-\Omega(n)}$,
  $\calS(\vec{x}\circ\vec{\by})$ occupies $\gamma$-fraction of some
  length-$M$ interval that contains $\tau$.
\end{lemma}
\begin{proof}
Let $I=[\alpha,\alpha+M]\subseteq [-\eps nM/8,\eps nM/8]$ be a length-$M$ interval 
  such that $J:= I\cap (\calS(\vec{x})\setminus \calB(\vec{x}))$ satisfies 
  $|J|\ge \gamma M$.
Assume without loss of generality that $\tau\notin I$ and $\tau<\alpha$.
As $|\tau|=o(Mn)$, we have $0<\alpha-\tau<\eps nM/7$.
For $\by_1,\ldots,\by_{m'-m}\sim [0:M-1]$,
it follows from Hoeffding inequality (note that the sum has expectation $(m'-m)(M-1)/2\ge (M-1)\cdot \eps n/6$ as $m'-m \ge \eps n/3$) that
%\yjnote{The modified $m$ given in \Cref{yjfootnote:2} still works; see \Cref{yjfootnote:0,yjfootnote:1,yjfootnote:2}.}) that 
$$
\sum_{i\in [m'-m]}\by_i\ge \eps nM/7>\alpha-\tau.
$$
When this happens, we let $\ell$ be the smallest index such that 
$$
\beta:=\sum_{i\in [\ell]}\by_i >\alpha-\tau  
$$
so $\beta-M\le \alpha-\tau$.
On the one hand $\calS(\vec{x}\circ \vec{\by})$ occupies at least  $\gamma$-fraction
  of $[\alpha-\beta,\alpha-\beta +M]$ since it contains $J$ after shifting down by $\beta$.
On the other hand, we have
  $\tau>\alpha-\beta$ and $\tau\le \alpha-\beta+M$.
This finishes the proof of the lemma.
\end{proof}

Now we are ready to conclude the proof of \Cref{thm:optprecisionC0}. 
%Recall that $M = O^*(|C|^{(1-\eps)n})$ for some constant $\eps > 0$, and let $\delta \in (0, \eps / 4]$ be a constant chosen so as to satisfy \eqref{eq:main:delta}. Let $\eps' > 0$ be another constant chosen so that
%    $1 - \eps = (1 - \eps')(1- 5\delta)$. Thus
%    \[
%        M = O^*(|C|^{(1-\eps')(1 - 5\delta)n}).
%    \]
Combining \Cref{lem:main} and \Cref{lem:lower_bound_on_sum} we have with probability
  at least $1-e^{-\Omega(n)}$ over $\vec{\bx}\sim [0:M-1]^{m'}$,
  $\calS(\vec{\bx})$ contains at least $\gamma$-fraction of
  a length-$M$ interval that contains $\tau$. 
%that with probability $1-e^{-\Omega(n)}$, the set $\bm{S}_{n - 5 \lfloor\delta n\rfloor }$ occupies at least a $\gamma = \eps^{2} / (96 |C|^{2})$ fraction of some length-$M$ interval within $[-\delta n M, \delta n M]$.
    % \ignore{. that has at least $\gamma = \eps^{2} / (96 |C|^{2})$ in the set $\bm{S}_{n - 5 \lfloor\delta n\rfloor }$.}
Fix such a $\vec{x}\in [0:M-1]^{m'}$.
We draw the next $n-m'=\Omega(n)$ random elements $\bx_{m'+1},\ldots,\bx_n\sim
  [0:M-1]$.
%    At this point, we have at least $\lfloor \delta n \rfloor$ inputs $\bx_i$ to which coefficients have not yet been assigned. 
 {For each such $\bx_i$, there is at least a $\gamma$ probability that 
   either $w+\bx_i=\tau$ or $w-\bx_i=\tau$ for some $w\in \calS(\vec{ x})$.
When this happens for some $i$, we get a solution by setting coefficients $c_1,\ldots,c_{m'}$
  to achieve $w$ in the sum over $\vec{x}$, setting $c_i$ to be either $1$ or $-1$
  accordingly, and setting every other $c_{m'+1},\ldots,c_n$ to be $0$.
   % the outcome of that $\bx_i$ %(when it is chosen uniformly at random from $[0:M-1]$) 
%yields a collision with $\tau$ by taking the coefficient $c_i$ to be $1$ or $-1$ and taking all the other unassigned coefficients in $c_{m'+1},\ldots,c_n$ to be 0. 
The probability that all $n-m'$ such $\bx_i$'s fail to yield a solution is at most $(1 - \gamma)^{\Omega(n)} = e^{-\Omega(n)}.$} 
(Note that 
  the same proof works for the case when $\tau=0$ since during Step 3, we always get
  a not-all-zero solution with at least one $\{\pm 1\}$-coefficient when the event happens for some $\bx_i$.)
\ignore{\rnote{This was ``\gray{Drawing $\Theta_n(1/\gamma) = \Theta_n(1)$ elements ensures that we create a solution with high probability. For $i \leq \lfloor \delta n \rfloor$, let $\bm{\X}_i$ be an indicator random variable that is $1$ if drawing the $i$th extra input creates a collision with $\tau$ when added to our dense interval. As $\E[\bm{\X}_i] \geq \gamma /2$, using a Chernoff bound establishes that this event occurs with probability
    \[
        \Pr\Big[\sum_{1 \leq i \leq \lfloor \delta n \rfloor} \bm{\X}_i < \frac{1}{2} \cdot \frac{\gamma}{2} \cdot \lfloor \delta n \rfloor\Big] \leq \exp\Big(-\frac{\gamma}{24} \cdot \lfloor \delta n \rfloor\Big) = e^{-\Omega(n)}.
    \]
    }'' but I think we don't need a Chernoff bound, right?}}\ignore{At this point, we can assign any remaining elements the 0 coefficient to create a solution.}
    
    Taking a union bound over the failure probability at each of the three steps, we get that there is a solution with probability  $1 - e^{-\Omega(n)}$ when $M\le |C|^{(1-\eps)n}$.
    This finishes the proof of \Cref{thm:optprecisionC0}.

\section{Algorithmic Results}
\label{sec:algorithm}

This section proves \Cref{thm:main}.
Using structural results from the previous section, we start by showing that
  it suffices to design an algorithm for the core case when $M\in|C|^{(1\pm \eps)n}$
  for a small positive constant $\eps$.

\subsection{Reduction to the Core Problem}

We start with some notation. 
Given $C=C(d)$ or $C_0(d)$ for some fixed $d\in \mathbb{N}_{\ge 1}$,
  a \emph{solution profile} $\pi=(\pi_w)_{w\in C}$ is a tuple of 
  nonnegative integers that sum to $n$.
For an input vector $\vec{x}=(x_1,\ldots,x_n)$, a target offset $\tau$ and 
  a (solution) profile $\pi$,
  we say $\vec{c}\in C^n$ is a solution to GSS under the profile $\pi$ 
  if $\vec{c}\cdot \vec{x}=\tau$ and the number of occurrences of $w$ in vector $\vec{c}$ is $\pi_w$
  for each coefficient $w \in C$.

Our goal of this section is to prove the following theorem 
  for the problem of solving GSS under a given profile:
  
\begin{theorem}\label{finalfinalmain}
Fix any $d\in \mathbb{N}_{\ge 1}$ and let $C=C(d)$ or $C_0(d)$. 
For any sufficiently small constant $\xi>0$,
  there is a constant $\eps>0$ and a randomized algorithm with running
  time $|C|^{\Lambda(|C|)n+\xi n}$ that has the following performance guarantee.
Given any $M\in|C|^{(1\pm \eps)n}$, $\tau$ with $|\tau|=o(nM)$ and a profile $\pi$,
  the algorithm succeeds\footnote{Similar to the original GSS problem,
  an algorithm fails if it returns ``no solution'' but indeed there is a solution
  under the given profile $\pi$.} on $(M,\tau,\pi,\vec{\bx})$ with probability at least
  $1-e^{-\Omega(n)}$ (over the draw of $\vec{\bx}\sim [0:M-1]^n$ and the randomness
  of the algorithm).
\end{theorem}

We first use
  \Cref{finalfinalmain} and our structural results 
  to prove \Cref{thm:main}; the proof of \Cref{thm:dense-instances}
  is similar and can be found in \Cref{sec:appendix2}.
The intuition is that if $M\ge |C|^{(1+\eps)n}$ then 
  most likely there is no solution but if $M\le |C|^{(1-\eps)n}$,
  then we can reduce the instance size to $n'$ so that $M$ falls inside
  the window $|C|^{(1\pm \eps)n'}$ and the algorithm of 
  \Cref{finalfinalmain} applies.
  
\begin{proof}[Proof of \Cref{thm:main} assuming \Cref{finalfinalmain}]
Fix a constant $\zeta > 0$ as in the statement of \Cref{thm:main}. We will show an algorithm for average-case GSS with running time $|C|^{\Lambda(|C|)n + \zeta n}$, assuming \Cref{finalfinalmain}. Set $\xi \leq \zeta / 2$ to be sufficiently small that \Cref{finalfinalmain} holds, and let $\eps>0$ be the constant determined by $\xi$ in \Cref{finalfinalmain}. We begin by considering $C = C_0(d)$ and split into three cases based on the size of $M$.

First, if $M\ge |C|^{(1+\eps)n}$ then it follows from 
  \Cref{thm:optprecisionC0} that the probability of having a solution is $e^{-\Omega(n)}$
  and thus returning ``no solution'' achieves overall success probability
  $1-e^{-\Omega(n)}.$
Next, if $M\in|C|^{(1\pm \eps)n}$ then we run the algorithm in \Cref{finalfinalmain}
  for every profile $\pi$ and return any solution it finds;
  we return ``no solution'' if no solution is found for any profile $\pi$.
Given that there are only polynomially many profiles (since $d$ is fixed),
  the success probability remains $1-e^{-\Omega(n)}$ by a union bound
  and the running time only goes up by a polynomial factor.
%If the latter returns a solution for any profile $\pi$, 

Finally we consider the case when $M\le |C|^{(1-\eps)n}$ but $M=|C|^{\Omega(n)}$.
Let $n'=\Omega(n)$ be an integer such that $M$ is between $|C|^{(1-\eps)n'}$ and $|C|^{(1-\eps/2)n'}$, and note that $\tau$ still satisfies
  $|\tau|=o(n'M)$ given $n'=\Omega(n)$.
We run the algorithm in \Cref{finalfinalmain} on the first $n'$ input integers
  $\vec{\bx}'=(\bx_1,\ldots,\bx_{n'})$
   and try all possible profiles.
It follows from \Cref{thm:optprecisionC0} that there is a solution
  with probability at least $1-e^{-\Omega(n')}=1-e^{-\Omega(n)}$ and thus
  a solution will be found by the algorithm of \Cref{finalfinalmain}
  with probability $1-e^{-\Omega(n)}$.
Given that $C=C_0(d)$, any solution to $\vec{\bx}'$ can be extended to 
  obtain a solution to $\vec{\bx}$ by assigning the 0 coefficient to the remaining inputs.
So in this case our algorithm finds a solution to $\vec{\bx}$ with
  probability $1-e^{-\Omega(n)}$. % and of course succeeds with the same probability.
%Fix a $\eps > 0$. For $M = |C|^{\alpha n + o(n)}$ for some $\alpha \in (0, 1)$, \Cref{thm:main} follows from \Cref{thm:dense-instances}. For $M \in |C|^{(1\pm \eps)n}$, \Cref{thm:main} follows from \Cref{prop:correctness,prop:runtime}. Finally, for $M > |C|^{(1 + \eps)n}$, the answer ``no solution'' is correct with probability $1 - e^{-\Omega(n)}$ by \Cref{corr:optprecisionC1} and \Cref{thm:optprecisionC,thm:optprecisionC0}.

% TR: stopped here AM

Second, consider $C=C(d)$.
The first two cases, $M \geq |C|^{(1+\eps)n}$ and $M \in |C|^{(1 \pm \eps)n}$, are similar, except that in the second case by answering
  ``no solution'' we achieve only a success probability of $1-o_n(1)$ by 
  \Cref{corr:optprecisionC1} and \Cref{thm:optprecisionC}; the success probability
  in the first case remains $1-e^{-\Omega(n)}$. 
The main difference occurs in the last case, 
  as we can no longer extend a solution
  to $\vec{\bx}'$ to a solution to $\vec{\bx}$ by setting the coefficients of
  every other input integer to be $0$.
Instead, we can shrink the input instance as follows. Reveal the last $n - n'$ input integers one by one. For each input $\bx_i$, $i \in [n'+1:n]$, perform the following operation: if the current target $\tau$ is positive, assign $-1$ to $\bx_i$ and subtract it from $\tau$ to create a new target. If the current target $\tau$ is negative, assign $+1$ to $\bx_i$ and add it to $\tau$ to create a new target.
%Instead, starting with the input offset $\tau$, we can
%  reveal the last $n-n'$ input integers and for each $\bx_i$ of them, 
%  either the current $\tau$ is nonnegative, in which case we assign $-1$ to $\bx_i$
%  and subtract it from $\tau$, or the current $\tau$ is negative,
%  in which case we assign $+1$ to $\bx_i$.
After this procedure, we are left with $n'$ random numbers $\vec{\bx}'=(\bx_1,\ldots,\bx_{n'})$
  and a new offset $\tau'$ with $|\tau'| = o(n'M)$ such that any solution to $(\vec{\bx},\tau')$
  can be extended to a solution to $(\vec{\bx},\tau)$.
By \Cref{thm:optprecisionC},
  there is a solution with probability at least $1-o_n(1)$ over the randomness
  of $\vec{\bx}'$ in the $C = C(d)$, $d > 1$ case.
By \Cref{corr:optprecisionC1}, there is a solution with probability at least $1-o_n(1)$ over the randomness
  of $\vec{\bx}'$ in the $C = C(1)$ case if $\sum_i \bx_i$ has the same parity as $\tau$, as the shrinking procedure preserves the parity of $\sum_i \bx_i - \tau$.
This implies that running the algorithm of \Cref{finalfinalmain} 
  on $\vec{\bx}'$ and $\tau'$ for all profiles $\pi$ finds a solution
  with probability at least $1-o_n(1)$, which can be extended to a solution to $\vec{\bx}$ and $\tau$.
\end{proof}

\subsection{Algorithm Overview}
We give an overview of our algorithmic approach
  for \Cref{finalfinalmain} and the underlying definitions and ingredients required for the proof.
For the rest of the section,
  we consider $d$ to be a fixed integer with $C=C(d)$ or $C_0(d)$.
$\delta>0$ will denote a sufficiently small constant and 
  $\eps := \eps(\delta) >0$ a smaller constant that depends on $\delta$, defined later in the proof.
%Let 
%section, we consider a certain input range bound $M := M(n)$ and a coefficient set $C = C(d)$ or $C_0(d)$ for some $d \in \mathbb{N}_{\geq 1}$.
Our goal is to give an algorithm for GSS with a solution profile $\pi$ that runs in time $|C|^{\Lambda(|C|)n+O(\delta n)}$ and has success probability $1-e^{-\Omega(n)}$ when $M\in|C|^{(1\pm \eps)}$ and $|\tau|=o(nM)$.

\subsubsection{Terminology and Key Notions}

For the rest of the proof, the coefficient set $D := D(z)$ denotes a translation of the coefficient set $C$ by some \emph{nonzero} $z\in C$, i.e., $D=\{w - z: w \in C\}$.
% We denote $D$ by $C-z$. 
Note that
  we always have $0\in D$ whether $C=C_0(d)$ or $C(d)$.
%It turns out that we will focus on the following 

Given $D=D(z)$ for a certain nonzero $z\in C$, 
  we define \emph{input partitions} and \emph{half partitions}:

\begin{definition}%[Input Partitions]
%Given an input vector $\vec{x} = (x_1, \dots, x_n) \in [0:M-1]^n$, an \emph{input partition} of $\vec{x}$ is a tuple of multisets $\mathbb{S} := (S_c)_{c \in C}$ whose (multiset) union is the multiset of input elements $\{x_1, \dots, x_n\}$. 
A \emph{solution profile} with respect to $D$
  is a tuple of nonnegative integers $(\sigma_w)_{w \in D}$
  that sum to $n$.
An \emph{input partition} of $[n]$ with respect to $D$ is a tuple of 
  pairwise disjoint sets $\mathbb{S} := (S_w)_{w \in D}$ with union $[n]$.
We say $\mathbb{S}$ corresponds to the solution profile 
  $\sigma$ if $|S_w| = \sigma_w$ for all $w \in D$.
The \emph{size} of a solution profile $\sigma$ is 
  given by $|\sigma|:=n-\sigma_0$; similarly the \emph{size}
  of an input partition $\mathbb{S}$ is given by 
  $|\mathbb{S}|:=\sum_{w \in D\setminus \{0\}} |S_w|$.
\end{definition}

%Note that there is a natural bijection between input partitions

%Thus an input partition $\mathbb{S}$ of $\vec{x}$ is a solution to the instance with target $\tau$ precisely when
%\[
%    \sum_{c \in C} c \cdot \Sigma(S_c) = \tau.
%\]
%Given a coefficient set $C$, we refer to a partition $\pi := (\pi_{c})_{c \in C}$ of the integer $n$ (so $\sum_{c \in C} \pi_c = n$) into parts indexed by $c \in C$ as a \emph{solution profile.} The input partition $\mathbb{S} := (S_c)_{c \in C}$ corresponds to the solution profile $\pi$ such that $\pi_c := |S_c|$ for all $c \in C$. 

%For our purposes, it will be convenient to think of the size of solution profiles and input partitions in terms of the number of nonzero elements. We define the size of a solution profile to be $|\pi| := n - \pi_0$ and the size of an input partition to be $\sum_{c \in C \setminus \{0\}} |S_c|$.

\begin{definition}[Half-partitions]
% Given an input vector $\vec{x}=(x_1,\dots,x_n)$ and an input partition $\mathbb{S} := (S_c)_{c \in C}$ of $\vec{x}$, we say that 
Given a solution profile $\sigma$ with respect to $D$,
  a  \emph{half-partition} that corresponds to $\sigma$ is
  a tuple of sets $\mathbb{T}:=(T_w)_{w \in D\setminus \{0\}}$
  that satisfies the following conditions:
  %  of $\mathbb{S}$ is a tuple of multisets $\mathbb{T} := (T_c)_{c \in C \setminus \{0\}}$ with $T_c \subset S_c$ satisfying the following conditions:
    \begin{enumerate}
        \item $|T_w| = \sigma_w/2$ for $w \in D \setminus \{0\}$ if $\sigma_w$ is even.
        \item $|T_w| = (\sigma_w + 1)/{2}$ or $(\sigma_w - 1)/{2}$ for $w \in D \setminus \{0\}$ if $\sigma_w$ is odd.
    \end{enumerate}
Given an input partition $\mathbb{S}$ that corresponds to $\sigma$,
  we say $\mathbb{T}$ is a half-partition of $\mathbb{S}$ if
  $T_w \subseteq S_w$ for all $w \in D\setminus \{0\}$.
Two half-partitions $\mathbb{T} = (T_w)_{w \in D \setminus \{0\}}$ and $\mathbb{T}':= (T'_w)_{w \in D \setminus \{0\}}$
  of $\mathbb{S}$ are called a \emph{matching pair} of $\mathbb{S}$ if
  for every $w \in D \setminus \{0\}$, $T_w$ and $T'_w$
  are disjoint and their union is $S_w$.

Given a half-partition $\mathbb{T} = (T_w)_{w \in D \setminus \{0\}}$ and an input partition $\mathbb{S}=(S_w)_{w \in D}$,
  their \emph{signatures} under the input vector $\vec{x}$ are defined as 
$$
\sig(\mathbb{T},\vec{x}) := \sum_{w \in D\setminus \{0\}} w \cdot \Big(\sum_{i\in T_w} x_i\Big)
\quad\ \text{and}\quad\ 
\sig(\mathbb{S},\vec{x}) := \sum_{w \in D\setminus \{0\}} w \cdot \Big(\sum_{i\in S_w} x_i\Big).$$
%    The \emph{signature} of a half-partition is the value $\sig(\mathbb{T}) := \sum_{c \in C \setminus \{0\}} c \cdot \Sigma(T_c)$.
%    If $\mathbb{T} := (T_c)_{c \in C \setminus \{0\}}$ and $\mathbb{T}' := (T'_c)_{c \in C \setminus \{0\}}$ are two half-partitions of $\mathbb{S}$ such that for each $c \in C$ the multiset union of $T_c$ and $T'_c$ is $S_c$, we say that $\mathbb{T}$ and $\mathbb{T}'$ are a \emph{matching pair}.
\end{definition}

%We observe that since $C$ is either $C_0(d)$ or $C(d)$, the total number of possible solution profiles is $n^{O(d)}=\poly(n)$, and thus an exponential-time algorithm that never returns false positives can easily afford to guess a feasible solution profile for any yes-instance of GSS.
%

For a fixed input partition $\mathbb{S}$, we let $A(\mathbb{S})$ denote the set of all half-partitions of $\mathbb{S}$. 
%We also write $A(\vec{c})$ for $A(\mathbb{S})$ when $\vec{c}$ is the candidate solution corresponding to $\mathbb{S}$. 
Let $B(\sigma)$ denote the set of all half-partitions that correspond to $\sigma$. 
%\orange{(so here too, an element of $B(\pi)$ is a half-partition $\mathbb{T}$)}. 
For brevity, we write $a := a(\sigma) = |A(\mathbb{S})|$,
%\yjnote{``$a$'' is used both for elements in $C$ and the number of all half-partitions... Maybe a different notation?}
as $a$ is the same for all $\mathbb{S}$ corresponding to $\sigma$, and $b :=b(\sigma)= |B(\sigma)|$ when $\sigma$ is clear from context.

\subsubsection{High-Level Algorithm Sketch }
\label{subsubsec:sketch}

Let $\pi$ be the input solution profile with respect to $C=C_0(d)$ or $C(d)$.
Our algorithm starts by picking a nonzero $z\in C$ and translating the coefficient set to $D=C-z$.
We can use $\pi$ to induce a solution profile $\sigma$ with respect to $D$,
  where $\sigma_w = \pi_{w + z}$ for each $w \in D$.
  Our goal is now to find an input partition $\mathbb{S}$
  (with respect to $D$) that corresponds to $\sigma$ such that
\begin{equation}\label{eq:newgoal}
\sig(\mathbb{S},\vec{x})=\tau-z\cdot \sum_{i\in [n]} x_i.
\end{equation}
 
We set $z = \max_{z' \in C\setminus \{0\}} \pi_{z'}$, in order to minimize the size $|\sigma|$ of solution profile $\sigma$. 
This is natural because it allows us to focus on small-size
  candidate solutions (input partitions) after the translation.   
When $C=C(d)$, there exists a $z\in C$
  such that the translated profile $\sigma$ has size $|\sigma|\le \frac{|C|-1}{|C|}n$ by the pigeonhole principle.
And when $C=C_0(d)$, this is not possible in general 
  because $\pi_0$ can be large. (Looking ahead, we require $z \neq 0$ for our proof of \Cref{lem:signature-distribution}).
To this end, we prove the following technical lemma, which states that if $\pi_0\ge (\frac{1}{|C|}+\delta)n$ then there is no solution with probability at least $1-e^{-\Omega(n)}$.
This allows to focus on the case when $\pi_0<(\frac{1}{|C|}+\delta)n$, in which case we can choose $z \in C$ such that $\sigma$ satisfies
\begin{equation}\label{eq:sizebound}
|\sigma|\le \left(\frac{|C|-1}{|C|}+\frac{\delta}{|C|-1}\right)n.
\end{equation}
%after the translation.

\begin{lemma}
Suppose $M\in|C|^{(1\pm \eps)n}$ for some sufficiently small constant 
  $\eps = \eps(\delta) >0$.
If $\pi_0\ge (\frac{1}{|C|}+\delta)n$,
  then $\vec{\bx}\sim [0:M-1]^n$ has no solution with probability 
  at least $1-e^{-\Omega(n)}$.
\end{lemma}
\begin{proof}
This follows from taking a union bound over all solutions corresponding
  to $\pi$.
When $\pi_0\ge (\frac{1}{|C|}+\delta)n$, the number of candidate solutions is 
  at most $O^*(2^{hn})$ where 
$$
h:=H\left(\frac{1}{|C|}+\delta, \frac{1}{|C|}-\frac{\delta}{|C|-1}
,\ldots,\frac{1}{|C|}-\frac{\delta}{|C|-1}\right) = \log_2 |C| - f(\delta), 
%-\Omega(\delta^2).\xnote{I think this is true; but feel free to change it back to $f(\delta)$.}
$$
where $f$ is a nonnegative, increasing function on $[0,1]$. It follows from holding out one input element corresponding to a nonzero coefficient that each candidate solution is a solution with probability
  at most $1/M$.
Taking the union over all candidate solutions, we have that for
  $\eps < f(\delta)$, the probability of having a solution
  under the profile $\pi$ is $e^{-\Omega(n)}$.
\end{proof}

From now on, we consider the translated coefficient set $D$ and solution profile $\sigma$ that satisfy 
  the size bound in \eqref{eq:sizebound}. (Notice that $|D| = |C|$.) Our goal is to find 
  an input partition that satisfies \eqref{eq:newgoal} when it exists.
%We start by showing that if $\pi$
However, even with its relatively small size,
the profile $\sigma$ typically still corresponds to an exponential number of input partitions.
Assume that $\mathbb{S}$ is an input partition satisfying \eqref{eq:newgoal}.
At a high level, our algorithm searches for $\mathbb{S}$ by
  exploiting the fact that every input
  partition $\mathbb{S}$ with respect to $\sigma$ can be decomposed into a matching pair of half-partitions $(\mathbb{T},\mathbb{T}')$ in many ways. 
\Cref{lem:signature-distribution}, which we refer to as the \emph{Signature Distribution Lemma}, shows that with high probability over $\vec{\bx}\sim [0:M-1]^n$, 
  most half-partitions of $\mathbb{S}$ have distinct signatures.
%For example, suppose our algorithm successfully guesses a translation and a solution profile $\pi$ that corresponds to a solution $\mathbb{S}$ of the translated instance. 
Assuming this event occurs, we need to recover one matching pair of half-partitions in $A(\mathbb{S})$ from the larger search space $B(\sigma)$. (Note that our bound on the size of $\sigma$ comes into play at  this point: a profile $\sigma$ of smaller size would imply a larger ratio $a / b$ and make our search easier.) 

To recover a solution, we exploit the fact that the signatures of any 
  matching pair sum to $\tau-z\sum_i x_i$ but at the same time, by
  the Signature Distribution Lemma,
  most half-partitions in $A(\mathbb{S})$ have distinct signatures.
%  two half-partitions in a matching pair sum to the offset $\tau$. 
We (virtually) hash the half-partitions based on signature by grouping them into residue classes (which we will refer to as buckets) modulo a large random prime $\bm{p}$ with a 
  carefully picked magnitude. After the hash, most buckets contain a small fraction of the search space. We then sample pairs of buckets whose elements sum to $\tau-z\sum_i x_i \pmod{\bm{p}}$ and either subsample from  or exhaustively search over the paired buckets to recover a solution.
% \bigskip

\begin{algorithm}[t]
    \SetAlgoLined
    \SetKwFor{RepTimes}{repeat}{times}{end}
    \textbf{Constants:} $C = C(d)$ or $C_0(d)$, sufficiently
      small $\delta>0$ and smaller $\eps=\eps(\delta)>0$ \\
    \KwIn{$M \in |C|^{(1\pm\eps)n}$, $\tau $, 
      $\vec{x} \in [0:M-1]^n$, a profile $\sigma$
      of $D=C-z$ for some nonzero $z\in C$.}
    
    %\For{ $(z, \pi)$ such that $z \in C$ and $\pi$ is a solution profile for $C - z$ satisfying $|\pi|  \leq (1 - \frac{1}{|C|} + \delta)n$ }{
        %Translate %\footnote{Instance translation is formally defined in \Cref{sec:translation}.}
        %$\vec{x}$ by $z$.
        %\hfill
        %{\tt $\rhd$ ``Translation'' is defined in \Cref{sec:translation}.}
    
        Define $a := a(\sigma)$, $b := b(\sigma)$ as in \eqref{eq:a} and \eqref{eq:b}
        
        Define $P := \min({b}^{2/3}a^{-1/3}, a/n)$
        
        \RepTimes{$\poly(n)$}{
            Select a prime $\bm{p} \in [P, 2P]$ and an integer $\bm{r} \in [0:p-1]$ uniformly at random.
            
            If $P = a/n$, let $\bm{L}_1 = \{\mathbb{T} \in B(\sigma) \; | \; \sig(\mathbb{T},\vec{x}) = \bm{r} \pmod{\bm{p}}\}$. If $P = {b}^{2/3}a^{-1/3}$, let $\bm{L}_1$ be\\ \ \ \ \ a subsample of $O^*(b / \sqrt{a\bm{p}})$ elements drawn from this set with replacement.
            
            Let $\bm{L}_2$ be defined as $\bm{L}_1$ with respect to  $\{\mathbb{T} \in B(\sigma) \; | \; \sig(\mathbb{T},\vec{x}) = \tau-z\sum_i x_i-\bm{r} \pmod{\bm{p}}\}$.
        
            Sort $\bm{L}_1$ and $\bm{L}_2$ by signature and use Meet-in-the-Middle to find a disjoint pair $(\mathbb{T}_1, \mathbb{T}_2)$
            \\ \ \ \ \ satisfying $\sig(\mathbb{T}_1,\vec{x}) + \sig(\mathbb{T}_2,\vec{x}) = \tau-z\sum_i x_i$ if one exists.
        }
    %}
    \caption{Skeleton Algorithm for GSS}
    \label{alg:ATSS}
\end{algorithm}

\Cref{alg:ATSS} outlines the algorithm in pseudocode. Certain implementation details are deferred to the proofs of correctness (\Cref{prop:correctness}) and runtime (\Cref{prop:runtime}) to simplify the presentation. 
 %Strictly speaking, the algorithm solves instances of Average-Target Subset Sum (ATSS), a slight generalization of average-case GSS defined for convenience in \Cref{sec:translation}. 
 %In \Cref{sec:correctness_runtime}, we prove that \Cref{alg:ATSS} solves uniformly random instances $\vec{\bx} \sim [0:M-1]^n$ with fixed targets $\tau = o(Mn)$ with probability $1 - e^{-\Omega(n)}$ for $C_0(d)$ and $1 - o_n(1)$ for $C(d)$ (\Cref{prop:correctness}). We also show that the algorithm terminates in time $|C|^{\Lambda(|C|) n + O(\delta n)}$, where $\Lambda(z)$ is defined as in \Cref{thm:main} (\Cref{prop:runtime}).

%\subsubsection{Ingredients for the Analysis }

%Our algorithm takes as input several familiar parameters, including a coefficient set $C$, an input range bound $M$, an input vector $\vec{x}$ and an offset $\tau$. We establish results for coefficient sets $C = C(d)$ or $C_0(d)$ for all natural numbers $d$, uniformly random inputs $\vec{\bm{x}}$ and targets $\tau = o(Mn)$.

%Our algorithm also takes an input parameter $\delta$. In \Cref{subsubsec:sketch}, we emphasized the importance of a solution of small size: the smaller the solution, the faster the algorithm. However, we also
As mentioned earlier, we prove a Signature Distribution Lemma 
(\Cref{lem:signature-distribution}) in \Cref{sec:sig-distribution}
  that shows that when the instance has a solution, the corresponding 
  input partition almost always decomposes into 
%require that some solution to a given yes-instance must decompose into 
  many pairs of half-partitions with distinct signatures. (If this condition is false, most matching pairs of half-partitions hash to the same buckets and the sampling procedure fails.) 
%We prove  that with very high probability, almost all yes-instances have solutions that satisfy both the signature condition and an upper bound on size of $1 - \frac{1}{|C|} + \delta$, where $\delta$ is an arbitrarily small constant. For $\delta < 1/|C|$, this increases the runtime exponent by a linear factor that can be made arbitrarily small. The values $\delta$ and $|C|$ determine a small constant $\eps$ such that \Cref{alg:ATSS} succeeds on $M \in |C|^{(1 \pm \eps)n}$. Thus our other new parameter $\eps$ is a fixed constant that depends on $\delta$. (The preprocessing steps outlined in \Cref{subsubsec:sketch} let us assume an arbitrarily small constant $\eps$ without loss of generality.)

%The outer loop of our algorithm iterates over every possible translation $z \in C$ and every solution profile $\pi$ on the translated solution set, which takes $\poly(n)$ iterations. 
It remains to choose a value for the endogenous parameter $P$ that determines the magnitude of our modulus $\bm{p}$.
%that will allow us to recover a solution when we have correctly guessed a translation and a solution profile $\pi$ satisfying $|\pi| \leq (1 - \frac{1}{|C|} + \delta)n$. 
Recall that $a := a(\sigma) = |A(\mathbb{S})|$ counts the number of half-partitions of any input partition $\mathbb{S}$  corresponding to $\sigma$. We observe that
\begin{align}
    \label{eq:a}
    a(\sigma) = \Theta^*(2^{|\sigma|}).
\end{align}
Likewise, recall that $b  := b (\sigma)$ denotes the maximum number of 
  half-partitions corresponding to $\sigma$. 
We have $b (\sigma)=O^*(2^{hn})$ where (writing $|\sigma|=\alpha n$)
%When $\delta$ is sufficiently small and $s \leq (1 - \frac{1}{|C|} + \delta)n$, we have that
\begin{align} 
    \label{eq:b}
h := {H\left( 
    \frac{\alpha}{2(|C|-1)},\ldots, \frac{\alpha}{2(|C|-1)},1-\frac{\alpha}{2}\right)}
    %&= 2^{(\log_2 (2|C|) - \frac{|C|+1}{2|C|}\log_2(|C|+1))n + O(\delta n)}
    =\frac{\alpha}{2}\cdot \log_2 \frac{2(|C|-1)}{\alpha}
      +\left(1-\frac{\alpha}{2}\right)\cdot \log_2 \frac{2}{2-\alpha}.
\end{align}
%When $\delta$ is sufficiently small and 
%  $s\le (1-1/|C|-\delta)n$, we have 
%\begin{align}
   % \label{eq:b}
%    b^*=b^*(|\sigma|) = O^*\left(2^{H(\frac{1}{2|C|} + \frac{\delta}{2|C|-2}, \dots, \frac{1}{2|C|} + \frac{\delta}{2|C|-2}, \frac{|C| + 1}{2|C|} - \frac{\delta}{2})}\right) 
%    H\le   { \log_2 (2|C|) - \frac{|C|+1}{2|C|}\log_2(|C|+1) {\color{red}-} O(\delta )}.
%\end{align}
%\xnote{Should be $-O(\delta)$ instead of $+O(\delta)$ right?
%    the more biased the smaller it becomes.}
%$$
%(\frac{s}{2}(\log_2(\frac{2(|C|-1)}{s}) + \frac{2-s}{2} \log_2(\frac{2}{2-s}) - %s)n.
%$$
%This follows from the observation that $b^*(s)$ is maximized when $\sigma_0 = (\frac{1}{|C|}-\delta)n$ and $\pi_c = (\frac{1}{|C|} + \frac{\delta}{|C|-1})n$ for all other coefficients $c \in C$. %\Cref{eq:b} results from using \Cref{eq:stirling} to bound the number of half-partitions that result in this case.

We prove in \Cref{sec:sig-distribution} that almost all yes-instances have solutions that decompose into half-partitions with $\Omega(a)$ distinct signatures. In this case, we expect most residue classes to contain $\Omega(a/\bm{p})$ half-partitions that are part of matching pairs and $O(b/\bm{p})$ half-partitions in total. Setting $P = a/n$ thus ensures that every pair of residue classes $\bm{r} \pmod{\bm{p}}$ and $\tau -z\sum_i x_i \bm{r} \pmod{\bm{p}}$ contains $\Omega(n)$ matching pairs in expectation. By repeatedly sampling bucket pairs and using the Meet-in-the-Middle algorithm to search for a matching pair, we recover a solution.

In fact, for $|C| \leq 3$, including the case $C = C_0(1) = \{0, \pm 1\}$ that is our primary interest, it is possible to implement a slightly more efficient sampling procedure. In these cases, we set $P = {b}^{2/3}a^{-1/3} < a$, so that each residue class pair contains an exponential number of matching pairs in expectation. After choosing a bucket pair, we then subsample each bucket to create solution lists that are likely to contain at least one matching pair by a birthday paradox argument. We then use the Meet-in-the-Middle algorithm on the subsampled lists to recover a solution as before.

\subsubsection{Section Outline}

In the next subsection we prove the Signature Distribution Lemma.
%Formally, this lemma shows that when $\vec{\bx}\sim [0:M-1]^n$, with very high probability either
%  $\vec{\bx}$ has no solution or there is a solution
%  $\mathbb{S}$ that has many half-partitions with distinct signatures.
%The second one is concerned with the so-called \emph{pseudosolutions}:
%  a pair of half-partitions $(\mathbb{T},\mathbb{T}')$ under $\sigma$
%  are called a pseudosolution if their signatures sum to the right
%  number $\tau-z\sum_i x_i$ but they don't necessarily form a solution 
%  (because they may have $T_a\cap T_a'\ne \emptyset$ for some $a\in D\setminus %\{0\}$).
%We show that when $\vec{\bx}\sim [0:M-1]$, most likely the number 
%  of pseudosolutions is small.
We prove the correctness of the algorithm in \Cref{prop:correctness},
  where we show that for any $\vec{x}$ that  satisfies the 
    conditions of the Signature Distribution Lemma,
  \Cref{alg:ATSS} recovers $\mathbb{S}$ with high probability.
Finally, the proof of \Cref{prop:runtime} analyzes the runtime of our algorithm and fills in several implementation details omitted from the pseudocode.

\subsection{Signature Distribution Lemma}
\label{sec:sig-distribution}

For our approach to work, we need the half-partitions of some solution $\mathbb{S}$ to have many distinct signatures. While we do not know how to guarantee this for worst-case inputs, the following lemma gives what we need for average-case inputs. Essentially, it says that with very high probability, a random instance $\vec{\bx}\sim [0:M-1]^n$ either has no solutions or has a solution $\mathbb{S}$ such that different half-partitions have many different signatures.

For $C = C_0(1)$, there is a combinatorial argument that shows that signatures are well-distributed in the worst case (see \cite{mucha2019equal} for details).\footnote{In fact, this argument can be generalized to $C_0(d)$. However, it relies crucially on the existence of the $0$ coefficient and a coefficient set that is symmetric about 0, and thus cannot be easily extended to $C(d)$ or translated instances. } We employ a probabilistic argument that holds for general $C$.

\begin{lemma}[Signature Distribution Lemma]
    \label{lem:signature-distribution}
    Let $C=C_0(d)$ or $C(d)$ and $M\in|C|^{(1\pm \eps)n}$ for some constant $\eps>0$. Let 
      $\sigma$ be a solution profile with respect to $D=C-z$ for some
      $z\in C\setminus \{0\}$ with $|\sigma|$ satisfying 
      \[
        |\sigma| \leq \Big(\frac{|C| - 1}{|C|} + \frac{\delta}{|C|-1}\Big)n.
    \]
    With probability $1-e^{-\Omega(n)}$, $\vec{\bx}\sim[0:M-1]^n$
      either has no solution $\mathbb{S}$ to $\sig(\mathbb{S},\vec{\bx})=\tau-z \sum_{i } \bx_i$, or there is a solution $\mathbb{S}$
      such that half-partitions of $\mathbb{S}$ have at least 
      $\Omega(a(\sigma))$ many distinct signatures.
    %      Fix $d \in \NN$ and let $C=C_0(d)$ or $C=C(d)$. For any constant $\delta \in (0, \frac{1}{|C|})$, there exists a small constant $\eps > 0$ such that the following holds for $M \in |C|^{(1 \pm \eps)n}$: With probability $1 - e^{-\Omega(n)}$, an ATSS instance with $\vec{\bx}$ sampled uniformly at random from $[0:M-1]^n$ and target $\tau = o(Mn)$ satisfies one of the following two conditions:
    %\begin{enumerate}
    %    \item $\vec{\bx}$ is a no-instance.
    %    \item There exists $z \in C$ such that the $z$-translation of $\vec{\bx}$ has a solution $\vec{c}$, corresponding to a solution profile $\pi^z$, satisfying
    %    \begin{enumerate}
    %        \item the half-partitions of $\vec{c}$ have $\Omega(a(\pi^z))$ distinct signatures and
    %        \item $|\pi^z| < (\frac{|C| - 1}{|C|} + \delta)n$.
    %    \end{enumerate}
    %\end{enumerate}
\end{lemma}
\begin{proof}
Let $a:=a(\sigma)$ as in \eqref{eq:a}.
For the lemma to fail, there must be an input partition $\mathbb{S}=(S_w)_{w \in D}$
  that corresponds to $\sigma$ such that (1) $\sig(\mathbb{S},\vec{\bx})
  =\tau-z\sum_i\bx_i$ and (2) half-partitions of $\mathbb{S}$ have 
  $o_n(a)$ distinct signatures under $\vec{x}$.
Now fix any input partition $\mathbb{S}$ that corresponds to $\sigma$.
We will bound the number of $\vec{x}\in [0:M-1]^n$ such that 
both of the conditions above hold by $O(aM^{n-2})$.
Since there are no more than $|C|^n$ many input partitions $\mathbb{S}$,
  the number of such $\vec{x}$ is thus at most
    \begin{align*}
        % \label{eq:badsolutions1}
        |C|^n \cdot O(a M^{n-2}) 
        &= O^*(M^n \cdot |C|^{2\eps n} \cdot 2^{(-\frac{1}{|C|} + \frac{\delta}{|C|-1})n})  \\
        &\leq O^*(M^n \cdot 2^{-(\frac{1}{|C|} - \frac{\delta}{|C|-1} - 2\eps\log_2 |C|)n}),
    \end{align*}
    given that $a  = \Theta^*(2^{|\sigma|}) = O^*(2^{(1-\frac{1}{|C|} + \frac{\delta}{|C|-1})n})$ and $2^n\le |C|^n \leq M \cdot |C|^{\eps n}$. This is an exponentially small fraction of the input space, which has size $M^n$, when both
    $\eps$ and $\delta$ are sufficiently small.
    %the instance space given that 
    %\begin{align}
    %    \frac{1}{|C|} > 2\log_2(|C|)\eps + \frac{\delta}{|C|-1},
    %\end{align}
    %which clearly  holds for sufficiently small $\eps > 0$ given our initial assumption on $\delta$.
%Let $a=a(\pi)$.

Fixing an input partition $\mathbb{S}=(S_w)_{w \in D}$, we proceed to
  bound the number of $\vec{x}$ that satisfy both conditions above.
We start by showing that the number of $\vec{x}$ such that 
  half-partitions of $\mathbb{S}$ have $o(a)$ distinct signatures 
  is at most $O(aM^{n-1})$.

    To see this, consider any of the $a^2$ pairs of distinct half-partitions of 
      $\mathbb{S}$. With respect to any pair of distinct half-partitions, there exists an index $k\in [n]$ that appears in exactly one of them. For any fixed $\vec{x}_{[n]\setminus\{k\}}$, there is at most one choice of $x_k$ that results in both half-partitions having the same signature.
    As a result, the number of $\vec{x}$ that can lead to the same signature
      is at most $M^{n-1}$.
      %      either (1) an index $k \in [n]$ such that the $k$th input element is contained in exactly one of the two half partitions or (2) an index $k \in [n]$ such that the $k$th input element is assigned a different coefficient in each half partition. In either case, for any fixed choice of the other $n-1$ input elements $x_i, i \neq k$, it is clear that at most one choice $x_k \in [0:M-1]$ causes the signatures of the two half-partitions to collide. 
Union-bounding the number of signature-collisions over $a^2$ distinct pairs of half-partitions of $\mathbb{S}$ and $M^n$ instances gives at most $a^2 M^{n-1}$ signature-collisions in total over all instances $\vec{x}$.
This implies that the number of $\vec{x}$ that have 
  $o(a)$ distinct half-partition signatures, which happens only if there occur 
  $\Omega(a)$ signature-collisions, is at most $O(aM^{n-1})$.

 %   Suppose for contradiction that the half-partitions in $A(\vec{c}-z)$ have $o(a(\pi^z))$ distinct signatures on some input $\vec{x}$. Assigning $a(\pi^z)$ half-partitions to $o(a(\pi^z))$ signatures implies $\omega(a(\pi^z))$ signature-collisions by the pigeonhole principle, reasoning from the case in which half-partitions are evenly distributed.
    
%    Thus at most $O(a(\pi^z) M^{n-1})$ instances have $o(a(\pi^z))$ distinct half-partition signatures corresponding to $\pi^z$. 
%Let $\mathbb{S} := (S_c)_{c \in C}$ be the input partition corresponding to $\vec{c}$. 
Next we observe that the condition of $\vec{x}$ corresponding to $o(a)$ distinct half-partition signatures is independent of $\vec{x}_{S_0}$, as these elements do not affect half-partition signatures at all. On the other hand,
  $\vec{x}_{S_0}$ determines whether the equation
\begin{equation}
\sig(\mathbb{S},\vec{x})=\tau-z\sum_{i\in [n]} x_i
    \label{eq:s0-sig}
\end{equation}
holds because its entries appear on the right hand side with a nonzero coefficient $z$.
%  Recall that $\vec{c}$ is a solution for $\vec{x}$ if and only if
%    \begin{align}
%        \label{eq:muc-z}
%        (\vec{c} - \vec{\mu}_C) \cdot \vec{x} = \tau.
%    \end{align}
By holding out any input element indexed by $S_0$, we can see that among the $O(aM^{n-1})$ many input vectors $\vec{x}$ that have $o(a)$
  distinct half-partition signatures, only $1/M$-fraction of them 
  can satisfy \eqref{eq:s0-sig}.
Therefore the number of $\vec{x}$ that satisfies both conditions is at most $O(aM^{n-2})$.

This finishes the proof of the lemma.
%    Because $\mu_C \neq z$ by assumption, \eqref{eq:muc-z} holds for at most a $1/M$-fraction of the choices for $S_z$. (Again, this follows from holding out one element of $S_z$.)
 %   Thus for every $\vec{c} \in C^n$ such that $\pi_{\mu_C} \leq (\frac{1}{|C|} + \delta)n$, $\vec{c}$ solves at most $O(a(\pi^z) M^{n-2})$ instances such that the $z$-translation of $\vec{\bx}$ corresponds to a solution profile $\pi^z$ with $o(a(\pi^z))$ distinct signatures. Recall that to fail both conditions in the Lemma statement, an instance $\vec{x}$ must be solved by some candidate solution $\vec{c} \in C^n$ such that for every $z$-translation of $\vec{c}$ satisfying $|\pi^z| < (\frac{|C| - 1}{|C|} + \delta)n$, the half-partitions of $\vec{c}$ have $o(a(\pi^z))$ distinct signatures. The number of instances that fail both conditions in the lemma statement due to candidate solutions $\vec{c} \in C^n$ in Case~1 is thus 
\end{proof}

\subsection{Correctness and Runtime }
\label{sec:correctness_runtime}

\begin{proposition}[Correctness of \Cref{alg:ATSS}]
Under the setting of \Cref{finalfinalmain},
%    Fix $C = C_0(d)$ or $C(d)$, $\delta < \frac{1}{|C|}$, $\eps := \eps(\delta)$ satisfying \Cref{lem:signature-distribution}, $M \in |C|^{(1 \pm \eps)n}$ and $\tau =  o(Mn)$. 
\Cref{alg:ATSS} 
  succeeds on $(M,\tau,\pi,\vec{\bx})$ with probability at least
  $1-e^{-\Omega(n)}$.
  %(over the draw of $\vec{\bx}\sim [0:M-1]^n$ and the randomness
  %of the algorithm).
%solves instances of ATSS sampled uniformly from $[0:M-1]^n$ on $C_0(d)$ with probability $1 - e^{-\Omega(n)}$ and on $C(d)$ with probability $1 - o_n(1)$.
    \label{prop:correctness}
\end{proposition}
\begin{proof}
By the Signature Distribution Lemma, we may assume that either 
  $\vec{x}$ has no solution or there is an input partition $\mathbb{S}$
  such that $\sig(\mathbb{S},\vec{x})=\tau-z\sum_{i\in [n]}x_i$
  and half-partitions of $\mathbb{S}$ have $\Omega(a)$ many distinct signatures.
We assume the latter case because \Cref{alg:ATSS} never returns
  false solutions.

    We proceed to argue that for most primes $p \in [P, 2P]$, the signatures of $A(\mathbb{S})$ are distributed over many residue classes modulo $p$. To see this, consider any two distinct signatures $s_1, s_2$ of elements of $A(\mathbb{S})$. Because $|s_1 - s_2| < n |C| M = 2^{O(n)}$ and $P = 2^{\Omega(n)}$, the value $|s_1 - s_2|$ is divisible by $O_n(1)$ prime factors larger than $P$. Moreover, for sufficiently large $P$, there exist at least $P / \log(P)$ prime integers in the interval $[P, 2P]$. (This follows from the Prime Number Theorem. See, e.g., \cite[Lemma~2.1]{mucha2019equal}.) Thus for a prime $\bm{p}$ sampled uniformly at random from all primes in the interval $[P, 2P]$ we have
    \[
        \Prx_{\bm{p}}\big[\bm{p} \; \text{divides} \; |s_1 - s_2|\big] = O(\log(P) / P).
    \]
    
    Let $a'=\Omega(a)$ be the number of distinct signatures of $A(\mathbb{S})$.
    The expected number of pairs of half-partitions whose signatures collide mod $\bm{p}$ is $O(a'^2 \log(P) / P)$ by linearity of expectation. 
    %\gray{Moreover, the actual number of collisions is $O(a'^2\log(P) / P)$ with constant probability over the choice of $\bm{p}$ by Markov's inequality. Thus}\rnote{Would it be okay for us to replace this gray part just with the phrase ``By Markov's inequality''?} 
    By Markov's inequality, there exists a constant $\rho_1$ such that a randomly selected prime $\bm{p} \in [P, 2P]$ corresponds to at most $\rho_1 a'^2 \log(P)/P$ signature-collisions with probability at least $1/2$. Call such a prime $\bm{p} \in [P, 2P]$ with at most $\rho_1 a'^2\log(P) / P$ signature-collisions \emph{a good prime} $\bm{p}$. 
    
    It remains to show that for a good prime $p$, choosing a uniform random residue $\bm{r} \in [p]$ will allow us to recover $\mathbb{S}$ with inverse polynomial probability. With respect to a good prime $p$, we define a \emph{good residue} $\bm{r}$ to be one that satisfies the following two conditions:
    
    \begin{flushleft}\begin{enumerate}
        \item The signatures of at least $a' / 2p$ half-partitions in $A(\mathbb{S})$ are equal to $\bm{r} \pmod{p}$.
        \item The signatures of at most $\rho_2 bn^2 / p$ half-partitions in $B(\sigma)$ are equal to either $\bm{r} \pmod{p}$ or $\tau-\bm{r}-z\sum_ix_i \pmod{p}$, where $b:=|B(\sigma)|$, for a constant $\rho_2$ defined below.
        % \item The residue classes defined by $r$ and $t - r$ create fewer than $b / p$ \emph{pseudosolutions}: pairs of half-partitions whose signatures add to $0$ but contain at least one common element.
    \end{enumerate}\end{flushleft}
    
    Fix a good prime $p$. We first prove that $\Omega(p/n^2)$ residue classes each contain at least $a'/2p$ half-partition signatures corresponding to $\mathbb{S}$. To see this, assume for contradiction that at most $p/n^2$ residue classes contain at least $a'/2p$ solution half-partitions (we call such residue classes ``large''). Under this assumption, residue classes that are not large contain fewer than $a'/2p$ solution half-partitions, so at least $a'/2$ half-partitions must fall into large residue classes. However, reasoning from the case in which signatures are evenly distributed among large residue classes, this implies $\Omega(a'^2 n^2 / p)$ signature-collisions, which contradicts the assumption that $p$ is good.
    
    %Recall that $b := b(\pi)$ denotes the number of half-partitions of all input partitions corresponding to $\pi$. 
    Consider the $\lceil p / 2 \rceil$ pairs of residue classes defined by $(r, \tau - r-z\sum_ix_i)$ for some residue $r \in [p]$. Upon choosing a good $p$, Markov's inequality guarantees that most of the $\Omega(p/n^2)$ residue class pairs that contain $a'/2p$ matching pairs contain $O(bn^2 / p)$ half-partitions in total. Thus there exists a constant $\rho_2$ such that $1/2$ of these residue class pairs contain $\rho_2bn^2 / p$ half-partitions in total. In other words, the chance that uniform random $\bm{p}$, $\bm{r}$ are both good is $\Omega(n^{-2})$. We can inflate this probability to $1 - e^{-\Omega(n)}$ by choosing independent pairs $(\bm{p},\bm{r})$ a polynomial number of times as is done in \Cref{alg:ATSS}.
    
    We conclude by proving that \Cref{alg:ATSS} recovers a solution with very high probability conditioned on the choice of a good $p$ and $r$. First, consider the case in which $a/n < {b}^{2/3}a^{-1/3}$, in which case the algorithm sets $P = a/n$. Thus the residue class pair defined by $(r, \tau-r-z\sum_{i}x_i)$ contains $a'/2p = \Omega(a'n / a) = \Omega(n)$ matching pairs; searching with Meet-in-the-Middle guarantees recovery.
    
    Second, consider the case in which $a/n > {b}^{2/3}a^{-1/3}$, in which case  $P = {b}^{2/3}a^{-1/3}$. Let 
    \[
        \kappa > \frac{a'}{2p} = \frac{a'}{2} \cdot \frac{a^{1/3}}{{b}^{2/3}} = \Omega\left(\frac{a^{4/3}}{{b}^{ 2/3}}\right) = \Omega(n)
    \]
    denote the number of matching pairs in the residue class pair defined by $(r, \tau - r-z\sum_i x_i)$. In this case, we create $\bm{L}_1$ and $\bm{L}_2$ by subsampling  $$2\rho_2 b n^3 / \sqrt{ap} = O^*(b / \sqrt{ap})$$ 
    %= O^*(a^{1/3}b^{1/3})$$ 
    elements from each residue class uniformly at random with replacement.
    
    In expectation, the number of
    matching pairs contained in $\bm{L}_1$ and $\bm{L}_2$ is thus at least
    \[
        \frac{2\rho_2 b n^3}{\sqrt{ap}} \cdot \frac{\kappa p}{\rho_2 b n^2} \ge 2n\kappa\sqrt{\frac{p}{a}}. 
    \]
    Applying a Chernoff bound implies that both lists contain at least $n\kappa\sqrt{\frac{p}{a}} = \Omega(n^{3/2})$ half-partitions that are members of some matching pair with probability $1 - e^{-\Omega(n)}$ as $a / n > p$ by definition.
    
    In other words, with very high probability $\bm{L}_1$ and $\bm{L}_2$ each contain 
    \[
        n\kappa\sqrt{\frac{p}{a}} \geq n\frac{a'}{2p}\sqrt{\frac{p}{a}} = \Omega(n\sqrt{\kappa})
    \]
    uniform random elements sampled from matching pair sets of size $\kappa$. This yields a matching pair in our subsample with probability $1 - e^{-\Omega(n)}$ by the birthday paradox. Searching with Meet-in-the-Middle guarantees recovery.
\end{proof}

\begin{proposition}[Runtime of \Cref{alg:ATSS}]
    \label{prop:runtime}
    When fully specified as below, \Cref{alg:ATSS} runs in time
    $|C|^{\Lambda(|C|) n + O(\delta n)},$
    where $\Lambda(z)$ is defined as in \Cref{thm:main}.
\end{proposition}
\begin{proof}
%    Both outer loops of our algorithm run a polynomial number of times, so to complete the proof it suffices to show an efficient implementation of the inner loop. 
For fixed prime $p$ and residue class $r$, we start by describing how to
  generate $\bm{L}_1$ and $\bm{L}_2$ efficiently.
    
    To sample half-partitions by the residue class of their signature, we construct a $n \times p$ dynamic programming table in which each cell $(i, j) \in [n] \times [p]$ stores the number of coefficient vectors $\vec{c} \in D^i$ such that $\vec{c} \cdot x_{[i]} = j \pmod{p}$. This number is stored as $\poly(n)$ different values indicating how many coefficient vectors match each partial solution profile $\pi = (\pi_w)_{w \in D}$ that partitions the integer $i$. %(Even if $0 \not\in C$, we use the $0$ coefficient to indicate elements not contained in half-partitions.) 
    Each cell $(i, j)$ can be filled by consulting the cells $(i-1, j - w x_i \pmod{p})$ for each $w \in D$. As a result, constructing the table takes time $O^*(p)$.
    
    Our dynamic programming table allows us to sample uniformly at random from the set of half-partitions that correspond to $\sigma$ and have signatures that fall into the residue class $r \pmod{p}$. To do this, we begin at cell $(n, r)$ and consider only those coefficient vectors indexed by solution profiles corresponding to half-partitions. We sample a coefficient $w \in D$, weighting by the number of half-partitions that assign the coefficient $w$ to $x_n$. We then consider cell $(n-1, r - w x_n \pmod{p})$ and continue sampling until a single half-partition is recovered.
    
    Once the table is created, we consider two cases based on the choice of $P$.
    
    \begin{flushleft}\begin{itemize}
        \item \textbf{Case~1: $P = a/n$.} In this case, we sample the entire residue class:
    \begin{align*}
      \bm{L}_1 &= \{\mathbb{T} \in B(\sigma) \mid \sig(\mathbb{T},\vec{x}) = r \Mod{p}\}\ \ \quad \text{and}\\
      \bm{L}_2 &= \{\mathbb{T} \in B(\sigma) \mid \sig(\mathbb{T},\vec{x}) = \tau - r-z\textstyle{\sum}_i x_i \Mod{p}\}.
     \end{align*}
     Under the assumption that $p$ is a good prime and $r$ is a good residue class, this takes time $\poly(n) \cdot \rho_2 bn^2 / p = O^*(b/p)$. (If there are more than $\rho_2bn^2/p$ elements in our residue class, we know that $r$ is not a good residue and abort the loop.)
        \item \textbf{Case~2: $P = {b}^{2/3}a^{-1/3}$.} In this case, we sample $O^*(b / \sqrt{ap})$ half-partitions drawn without replacement to create $\bm{L}_1$ and $\bm{L}_2$. 
    \end{itemize}\end{flushleft}
    
    The final step of Case~1 is to search $\bm{L}_1 \times \bm{L}_2$ for a matching pair. The Meet-in-the-Middle procedure takes time $O^*(\max(|\bm{L}_1|, |\bm{L}_2|, \text{ps}(\bm{L}_1, \bm{L}_2))$, where $\text{ps}(\bm{L}_1, \bm{L}_2)$ counts the number of \emph{pseudosolutions}: pairs of half-partitions $(\mathbb{T}_1, \mathbb{T}_2) \in \bm{L}_1 \times \bm{L}_2$ such that $\sig(\mathbb{T}_1) + \sig(\mathbb{T}_2) = \tau-z\sum_i x_i$ but $(\mathbb{T}_1, \mathbb{T}_2)$ is not necessarily a matching pair (due to overlapping elements). Each pseudosolution must be checked and rejected. Using linearity of expectation over all distinct pairs of half-partitions in $B(\sigma)$, the expected number of pseudosolutions over a random input $\vec{\bm{x}}$ and choice of $p, r$ is 
    \[
        O\left(\frac{b^2}{Mp}\right) = \frac{b}{p} \cdot e^{-\Omega(n)},
    \]
    where it follows from \eqref{eq:b} that
      $b$ is smaller than $M\in|C|^{(1\pm \eps)n}$ by $e^{\Omega(n)}$.
    Hence by Markov's inequality, the probability that processing pseudosolutions takes time $\Omega(b/p)$, longer than the time to sample the residue class, is exponentially small. (If this occurs, the algorithm aborts and returns ``no solution''.) In Case~2, the expected number of pseudosolutions
  after subsampling is 
$$
O\left(\frac{b^2}{Mp}\cdot \left(\frac{b/\sqrt{ap}}{b/p}\right)^2\right)
\le O\left(\frac{b}{p}\cdot \frac{p}{a}\right)\cdot e^{-\Omega(n)}
\le O\left(\frac{b}{a} \right)\cdot e^{-\Omega(n)}
$$
so processing pseudosolutions takes time $O(b/a)$ with very high probability.
    
    The total runtime is thus the time it takes to create the table plus the time to sample and search for a matching pair. 
    %Evaluating the quantities $b^{2/3}a^{-1/3}$ and $a/n$ at different $|C|$ reveals that Case~1 holds for $|C| > 3$ and Case~2 holds for $|C| \leq 3$. 
    %\xnote{Should we include the calculation about $b$?
    %  Also it is not clear to me how to get expressions of
    %  runtime below.}
    %\trnote{ It's clearly true that for any fixed solution size $\pi$, $b$ is upper-bounded by the solution profile that evenly divides inputs among the coefficients. We should probably write out the calculation that shows $b/a$ is increasing in the size of the solution for $|C| > 3$, which implies the upper bound below. }
    In Case~1, our runtime is:
    %\trnote{
     %   Write down the reasoning here: We can write an upper bound on $b$ in terms of $\pi$. I think if you let $s = |\pi|/n$, you get
      %  \[
       %     b/a = O^*(2^{(\frac{s}{2}(\log_2(\frac{2(|C|-1)}{s}) + \frac{2-s}{2} \log_2(\frac{2}{2-s}) - s)n})
    %    \]
     %   The exponent of this expression is 
      %  \[
       %     (\frac{s}{2}(\log_2(\frac{2(|C|-1)}{s}) + \frac{2-s}{2} \log_2(\frac{2}{2-s}) - s)n.
       % \]
    %    We want to show it increases in $s$.
    %}
    \begin{equation}
        \label{eq:case1runtime}
        O^* (\max(p, b/p) ) = O^* (\max(a,b/a) ) = O^*(b/a),
    \end{equation}
    where the first equality follows from the fact that $p = \Theta(a/n)$ in Case~1 and the second equality follows from the fact that $a / n \leq b^{2/3}a^{-1/3}$ in Case~1.
 In Case~2, our runtime is
    \begin{equation}
        \label{eq:case2runtime}
        O^*(\max(p, b/\sqrt{ap},b/a)) = O^*(b^{2/3}a^{-1/3}). 
    \end{equation}
    Thus the algorithm has running time at most
    $O^*(\max(b/a,  b^{2/3}a^{-1/3}).$
    
    Let $\alpha := |\sigma|/n$. Recall from \eqref{eq:a} and \eqref{eq:b} that $a(\sigma) = \Theta^*(2^{|\sigma|})$ and that $b(\sigma) = O^*(2^{hn})$, where
    \[
    h = {H\left( 
    \frac{\alpha}{2(|C|-1)},\ldots, \frac{\alpha}{2(|C|-1)},1-\frac{\alpha}{2}\right)}
    %&= 2^{(\log_2 (2|C|) - \frac{|C|+1}{2|C|}\log_2(|C|+1))n + O(\delta n)}
    =\frac{\alpha}{2}\cdot \log_2 \left(\frac{2(|C|-1)}{\alpha}\right)
      +\left(1-\frac{\alpha}{2}\right)\cdot \log_2 \left(\frac{2}{2-\alpha}\right).
    \]
    It follows from these two equations that $b^{2/3}a^{-1/3} \leq a/n$ when $|C| \leq 3$, and thus we are always in Case~2 when $|C| \leq 3$. 
    
    %(note that when $b = \Omega^*(a^2)$ the first term is larger and otherwise the second term is larger).
      
    Assuming \eqref{eq:sizebound} without loss of generality, we have that  $\alpha$ satisfies $\alpha\le 1-1/|C|+O(\delta)$.
    Using \eqref{eq:a} and \eqref{eq:b}, we have the following bounds on $b/a$ and $b^{2/3}a^{-1/3}$: 
        %= 2^{\frac{n}{|C|}}|C|^n(|C|+1)^{-\frac{|C|+1}{2|C|}n} \cdot 2^{O(\delta n)}.
    $$
        \frac{b}{a}\le 2^{H_1(\alpha)n},\quad \text{where $H_1(\alpha)=   \frac{\alpha}{2}\cdot \log_2\left(\frac{2|C|-2}{\alpha}\right) + \frac{2-\alpha}{2}\cdot \log_2\left(\frac{2}{2-\alpha}\right) - \alpha  $}
       %|C|^{(1 - \frac{|C|+1}{2|C|}\log_{|C|}(|C|+1) + \frac{1}{|C|}\log_{|C|}(2))n + O(\delta n)}.
    $$
    and
    $$
    \frac{b^{2/3}}{a^{1/3}}\le 2^{H_2(\alpha)n},\quad\text{where  $H_2(\alpha)=   \frac{ \alpha}{3}\cdot \log_2\left(\frac{2|C|-2}{\alpha}\right) + \frac{2-\alpha}{3}\cdot \log_2\left(\frac{2}{2-\alpha}\right) -\frac{ \alpha}{3}. $}
    $$
    When $|C|\ge 4$, in the range
    $(0,1-1/|C|+O(\delta)]$ the maximum of $H_1$ and $H_2$ (achieved by $H_1$) is
    $$
    H_1\big(1-1/|C|+O(\delta)\big)\le \log_2 |C|+\frac{1}{C}-\frac{|C|+1}{2|C|}\log_2 (|C|+1)+O(\delta^2).
    $$
    When $|C|\le 3$, in the range $(0, 1-1/|C|+O(\delta)]$,
    the maximum (achieved by $H_2$) is 
%    and the function in $\alpha$ is increasing in the range when
%      $|C|\ge 4$ and thus, is maximized at $1-1/|C|+O(\delta)$ with
    $$
    H_2\big(1-1/|C|+O(\delta)\big)\le \frac{2}{3}\log_2 |C| -\frac{|C|+1}{3|C|}\log_2 \left(\frac{|C|+1}{2}\right)+O(\delta^2).
    $$
  Therefore, the runtime of the algorithm is $|C|^{\Lambda(n)n + O(\delta n)}$.
\end{proof}   
    %\trnote{ for |C| > 3.623, b/a is indeed greater than a. Not sure if we need to go into more detail here. Another (possible) concern to the reader is that the equation we give corresponds to b/a for the largest possible value of $\pi$. As solution size goes down, both b and a decrease but the quantity b/a decreases overall. One way to see this is to observe that for any $|\pi|$, $a \approx 2^{|\pi|}$ and $b \approx< 2^{H(\frac{|\pi|}{2n(|C|-1)}, \dots, 1 - \frac{|\pi|}{2n})n}$. Thus if we decrease $|\pi|$ by $\eps n$, the denominator shrinks by a factor like $2^{\eps n}$ and the numerator shrinks by a factor more or less like $|C|^{\eps n}$. }

        %= 2^{\frac{1}{3|C|}(1 + |C| + 2|C|\log_2|C| - (|C|+1)\log_2(|C|+1))n+O(\delta n)}.
        %= 2^{\frac{n}{3}(1 + \frac{1}{|C|})}|C|^{\frac{2n}{3}}(|C|+1)^{-\frac{|C|+1}{3|C|}n} \cdot 2^{O(\delta n)}.
       % $$
%        = |C|^{(\frac{2}{3} - \frac{|C| + 1}{3|C|} \log_{|C|}(\frac{|C| + 1}{2}))n + O(\delta n)}
   %$$

\section{Generalizations and Related Problems}
\label{sec:applications}

    \subsection{Other Coefficient Sets}
    \label{subsec:other-coefficients}
    
    \Cref{thm:main} applies to $C = C(d)$ and $C = C_0(d)$, but our results can be applied to GSS on scale multiples and translations of $C(d)$ and $C_0(d)$. To see this, observe that for any integers $\alpha, \beta$, an instance of GSS on the coefficient set $\alpha C + \beta$ is equivalent to GSS on $C$ with target $\alpha^{-1}(t - \beta \sum_i x_i)$. %(Here and elsewhere, $\alpha C + \beta$ is shorthand for $\{ \alpha c + \beta : c \in C\}$.)
    
    \subsection{Number Balancing}

    The Number Balancing problem attempts to divide $n$ real numbers in $[0,1]$ into two sets in a way that minimizes the difference between the two sums. The problem can be thought of as the optimization version of GSS on $C = \{0, \pm 1\}$. We introduce the following generalized version.
    
    \stepcounter{prob}\stepcounter{prob}
    \begin{prob}[Generalized Number Balancing (GNB)]{prob:GNB}
    \textbf{Input.} A vector $\vec{y} = (y_1, y_2, \dots, y_n) \in [0, 1]^n$, a coefficient set $C \subset \Z$, and a precision $\delta >0$. \\
    \textbf{Output.} A coefficient vector $\vec{c} \in C^n$ that satisfies $|\vec{c}\cdot \vec{y}| \leq \delta$, or ``no solution'' if no solution exists.
    \end{prob}
    
    In the average-case version of this problem, we consider inputs sampled uniformly at random from $[0,1]$. In \cite{karmarkar1982differencing}, Karmarkar and Karp demonstrate a linear-time algorithm for worst-case Number Balancing that achieves precision $\delta = n^{-\Omega(\log(n))}$. However, a solution with exponentially small precision always exists by the pigeonhole principle. Scaling an average-case GNB instance $\vec{\bm{y}}$ by $\delta^{-1}$ and truncating the result yields a vector of integers that can be interpreted as an instance $\vec{\bm{x}}$ of GSS sampled uniformly from $[0:\delta^{-1}-1]^n$. A solution to $\vec{\bm{x}}$ on $C$ with target $\tau = 0$ is then a solution to $\vec{y}$ on $C$ with precision $\delta n$. This insight yields the following corollary to \Cref{thm:optprecisionC0} (as well as analogues for \Cref{thm:optprecisionC} and \Cref{corr:optprecisionC1}, omitted here for brevity.) %\rnote{ Could write this in a way that emphasizes the pigeonhole principle - there's always a solution for delta bigger than $1 / (d+1)^n$. Xi: corollary insight applies to GESS. }

    \begin{corollary}[Optimal Precision for GNB with $C = C_0(d)$]
        \label{thm:optprecisionGNB-C0}
        Fix $C = C_0(d)$ and any $\eps > 0$, and consider $\vec{\by} \sim [0,1]^n$. Then we have
        \begin{align*}
            \Prx_{\vec{\by}}\big[\exists \vec{c} \in C^n : |\vec{\by} \cdot \vec{c}| < \delta n\big]
            \begin{cases}
                 = 1 - e^{-\Omega(n)} & \text{if~} \delta = \Omega^*(|C|^{-(1-\eps)n}) \\[0.5ex]
                 \le \delta |C|^n     & \text{if~} \delta \leq |C|^{-n}.
            \end{cases}
        \end{align*}
    \end{corollary}
    
    Moreover, the reduction from GNB to GSS allows us to use our algorithm to solve average-case GNB on symmetric coefficient sets.
    
    \begin{corollary}[Algorithm for Average-Case GNB]
        For any $\alpha \in (0, 1)$, $C = C(d)$ or $C = C_0(d)$, and any constant $\eps > 0$, there exists an algorithm that solves average-case GNB with precision $|C|^{-\alpha n}$ in time
        \[
            O(|C|^{\alpha \Lambda(|C|)n + \eps n}), % = O(|C|^{0.396\alpha n}),
        \]
        where $\Lambda(n)$ is defined as in \Cref{thm:main}. For uniform random $\vec{\by} \in [0,1]^n$, the algorithm is correct with probability at least $1 - e^{-\Omega(n)}$ for $C = C_0(d)$ and $1 - o_n(1)$ for $C = C(d)$.
    \end{corollary}
    \begin{proof}
        We can convert $\vec{\by} \sim [0,1]^n$ into $\vec{\bx} \sim [0:n|C|^{\alpha n}-1]^n$ by scaling and then truncating the input. (Note that this preserves uniform sampling.) Our structural results then guarantee the existence of a solution to this GSS instance with probability $1 - e^{-\Omega(n)}$ or $1 - o_n(1)$ depending on whether $C = C_0(d)$ or $C=C(d)$.\footnote{If $C = \{\pm 1\}$ and $|\vec{\bx}|_1$ has odd parity, a solution that achieves target $\tau = 1$ is fine.} If a solution exists, it corresponds to a GNB solution with precision $|C|^{-\alpha n}$ and can be recovered by \Cref{alg:ATSS} with probability $1 - e^{-\Omega(n)}$ in time $O(|C|^{\alpha \Lambda(|C|)n + \eps n})$ by \Cref{thm:dense-instances}.
    \end{proof}
    
    %\yjnote{ We've slightly modified the big-O notation in the statement of Theorem 2, so we should make sure that what we write in Section 4 and 5 is consistent. }

\bibliographystyle{alpha}
\bibliography{main}
    
\appendix

\section{Proof of \texorpdfstring{\Cref{corr:optprecisionC1}}{}}
\label{apx:corollary}

    Let $C = C(1)$. We first consider the upper bound  when $M \geq |C|^n = 2^n$.  Given any fixed $\vec{c} \in C^n$, conditioned on any values for $\bx_1, \bx_2, \dots, \bx_{n-1}$, the probability that $\bx_n$ is such that $\vec{\bx} \cdot \vec{c} \in \{\tau, \tau + 1\}$ is at most $2 / M$. Union-bounding over all coefficient vectors gives the result.

 {
    We proceed to prove the lower bound for $M \leq |C|^{(1-\eps)n} = 2^{(1-\eps)n}$ for some $\eps > 0$. We start with the case in which $|\tau| \leq M$ and then extend our analysis to all $\tau$ such that $|\tau| = o(Mn)$.
    
    Fix a target $\tau$ such that $|\tau| \leq M$.
    Let $\rZ_{n,\tau}$ denote the expected number of solutions $\vec{c}\in C^n$ of $\vec{c}\cdot \vec{\bx}=\tau$ over $\vec\bx \sim [0:M-1]^n$.}
    In this case, \cite{borgs2001phase} Proposition 3.1 together with
      $M\le 2^{(1-\eps)n}$ implies the following bounds on $\rZ_{n,\tau}$:%\xnote{There are two cases depending on whether $\tau$ is $0$; also I think their Proposition 3.1 works if $|\tau|\le C_0 M$ for some constant $C_0$.} 
        \[
            \E\big[\rZ_{n,\tau}\big] = \rho_n(1 + O(n^{-1})) \quad\text{and}\quad
            \E\big[\rZ_{n,\tau}^2\big]= 2\rho_n^2(1+O(n^{-1})).
            %\times 
            %\begin{cases}
            %    1 \text{ if } \tau = 0\\
            %    2 \text{ if } \tau \neq 0.
            %\end{cases}
        \]
        where $\rho_n$ is defined as
        $$
        \rho_n:=\sqrt{\frac{3}{2\pi n}} \cdot \frac{2^n}{M}.
        $$
        Note that we cannot directly apply Chebyshev's inequality to obtain
          concentration of $\rZ_{n,\tau}$ because of the extra factor of $2$ in 
          $\E[\rZ_{n,\tau}]^2$.
        (The factor of $2$ is there because of the observation that $\rZ_{n,\tau}$
        can have a solution only when the sum of $\vec{\bx}$ is even, which happens
        with probability $1/2$ over $\vec{\bx}$.)
        
        Since we are interested in the probability of having $\vec{c}\in C^n$
          with $\vec{c}\cdot \vec{x}\in \{\tau,\tau+1\}$, we have
        $$
            \E\big[\rZ_{n,\tau}+\rZ_{n,\tau+1}\big] = 2\rho_n(1 + O(n^{-1})) \quad\text{and}\quad
            \E\big[(\rZ_{n,\tau}+\rZ_{n,\tau+1})^2\big]= 4\rho_n^2(1+O(n^{-1})) ,
            %\times 
            %\begin{cases}
            %    1 \text{ if } \tau = 0\\
            %    2 \text{ if } \tau \neq 0.
            %\end{cases}
        $$where the second equation used the fact that $\E[\rZ_{n,\tau}\rZ_{n,\tau+1}]=0$,
        because for any $\vec{x}$ one can never have a solution for both $\tau$ and $\tau+1$
        due to the parity issue.
%        As $M \leq 2^{(1-\eps)n}$, this quantity is $2^{\Omega(n)}$. The same proposition implies that $\E[\rZ_{n,\tau}^2] = \orange{2^{\Omega(n)}}$. 
        It follows from Chebyshev's inequality that $\rZ_{n,\tau}+\rZ_{n,\tau+1}>0$
        with probability $1-o_n(1)$ (see the beginning of the proof of \Cref{thm:optprecisionC}).
        
%        Following the arguments in the proof of \cite[Theorem~2.1]{borgs2001phase} with $\rZ_{n,\tau}$ instead of $\rZ_{n,1}$ and $\rZ_{n,0}$ yields the result.
 %       \tr{Write out more details.}
        %\xnote{If we move the proof to the appendix maybe we can give more details? something similar to page 31 in the appendix.}
        
        To extend the result to a larger range of offsets, we fix $|\tau| = o(Mn)$, and assume without loss of generality that $\tau$ is positive. Consider the experiment in which input elements $\bx_1,\bx_2,\ldots$ are revealed~one by one and 
        assigned the ``$-1$'' coefficient.
        We stop at $\bx_i$ when $\bx_1+\ldots+\bx_i\in [\tau,\tau+M]$ for the first time.
%        and stop when their sum falls inside $[\tau-M/2,\tau+M/2]$. %Fixing each coefficient is equivalent to creating a new average-case GSS instance with fewer inputs and a smaller target $\tau$. 
        Standard concentration arguments show that this process 
          stops with less than $\eps n/2$ elements revealed with high probability (using $|\tau|=o(Mn)$).
        When this happens, we reduce the problem to a new GSS 
          instance with target $|\tau'|\le M$ and $n'\ge (1-\eps/2)n$ elements. 
        The latter implies that $M\le |C|^{(1-\eps/2)n'}$  so
          this reduces to our earlier analysis (with $\eps$ there replaced by $\eps/2$).
%          \orange{As $M \leq |C|^{(1-\eps)n} = |C|^{(1-\eps)n' + o(n)} \leq |C|^{(1-\eps/2)n'}$, our new instance reduces to the previous case.}

\section{Proof of \texorpdfstring{\Cref{thm:dense-instances}}{}}
\label{sec:appendix2}

The proof of \Cref{thm:dense-instances} is similar to the proof of \Cref{thm:main} and follows a similar intuition: if $M \leq |C|^{(1-\eps)n}$ for any $\eps > 0$, then we can reduce the instance size to $n'$ so that $M$ falls inside the window $|C|^{(1 \pm \eps)n'}$ and the algorithm of \Cref{finalfinalmain} applies. Moreover, shrinking the instance size results in a faster running time.

\begin{proof}[Proof of \Cref{thm:dense-instances}]
Fix $d \in \N_{\geq 1}$, $C = C(d)$ or $C = C_0(d)$, and $\tau$ such that $|\tau| = o(Mn)$. Consider $M = |C|^{\alpha n + o(n)}$ for some $\alpha \in (0,1)$. Fix a constant $\zeta$ as specified in the statement of \Cref{thm:dense-instances}. We proceed to show an algorithm for average-case GSS with running time $|C|^{\alpha\Lambda(|C|)n + \zeta n}$. 

As $M = |C|^n \cdot 2^{-\Omega(n)}$, we can use the procedures described in the proof of \Cref{thm:main} to create a new instance $\vec{\bx}'$ of average-case GSS with $n'$ elements, where $n'$ satisfies
\begin{equation}
    \label{eq:thm2eps1}
        |C|^{(1-\eps_1)n'} \leq M \leq |C|^{(1-\eps_1/2)n'}
\end{equation}
for a constant $\eps_1 > 0$ that can be made arbitrarily small.

In the $C = C_0(d)$ case, we can simply set $\vec{\bx}' = \vec{\bx}_{[n']}$. In the $C = C(d)$ case, we perform the shrinking operation described in the proof of \Cref{thm:main}, assigning $+1$ and $-1$ coefficients to the elements of $\vec{\bx}_{[n'+1:n]}$ and adjusting the target $\tau$ accordingly. In both cases, we create a new instance $\vec{\bx}'$, uniformly distributed over $[0:M-1]^{n'}$, with a new target $\tau'$ such that $|\tau'| = o(n'M)$, and for which any solution can be easily converted to a solution for $\vec{\bx}$.

Set $\xi \leq \zeta/2$ to be sufficiently small that \Cref{finalfinalmain} holds on $\vec{\bx}'$, and let $\eps > 0$ be the constant determined by $\xi$ in \Cref{finalfinalmain}. Run the algorithm in \Cref{finalfinalmain} on $\vec{\bx}'$ for every profile $\pi$ and return any solution found (or ``no solution'' if no solution is found for any profile $\pi$.) Because there are polynomially many solution profiles, this takes time
\[
    O^*(|C|^{\Lambda(|C|)n' + \xi n})
    = O^*(|C|^{(\alpha \Lambda(|C|)/(1-\eps_1) + \zeta/2)n + o(n)}),
\]
where we use the fact that $n' \leq \alpha / (1-\eps_1)$, which follows from \Cref{eq:thm2eps1}. For sufficiently small $\eps_1$, this is dominated by $|C|^{\alpha\Lambda(|C|)n + \zeta n}$.

Taking a union bound over the chance that the algorithm in \Cref{finalfinalmain} fails on any solution profile yields a success probability of $1 - e^{-\Omega(n)}$ on $\vec{\bx}'$. In the $C = C_0(d)$ case, $\vec{\bx}'$ has a solution with probability $1 - e^{-\Omega(n)}$ by \Cref{thm:optprecisionC0}, and thus we solve $\vec{\bx}$ with probability $1 - e^{-\Omega(n)}$. By \Cref{thm:optprecisionC}, there is a solution with probability at least $1-o_n(1)$ over the randomness of $\vec{\bx}'$ in the $C = C(d)$, $d > 1$ case. By \Cref{corr:optprecisionC1}, there is a solution with probability at least $1-o_n(1)$ over the randomness
  of $\vec{\bx}'$ in the $C = C(1)$ case if $\sum_i \bx_i$ has the same parity as $\tau$, as the shrinking procedure preserves the parity of $\sum_i \bx_i - \tau$. Thus we solve $\vec{\bx}$ with probability $1 - o_n(1)$ in this case.
\end{proof}

% {\color{red}
%     Fix\xnote{Need to make a pass; should be similar to the proof of theorem 1.} $C = C(d)$ or $C_0(d)$ and $M = |C|^{\alpha n + o(n)}$ for some $\alpha \in (0, 1)$, an offset $\tau$ such that $\tau = o(Mn)$, and a constant $\eps > 0$. Consider instances $\vec{\bm{x}}$ sampled uniformly at random from $[0:M-1]^n$.
    
%     The proof proceeds in two steps: First, we use the Shrinking Lemma (\Cref{lem:shrink}) to create a new instance $\vec{\bm{x}}'$ of average-case GSS. The new instance has $n'$ elements, where $n'$ satisfies
%     \[
%         |C|^{(1-\eps_1)n'} \leq M \leq |C|^{(1-\eps_1/2)n'}
%     \]
%     for a constant $\eps_1 > 0$ that can be arbitrarily small.
    
%     We then use \Cref{alg:ATSS} to solve the instance in time $|C|^{\Lambda(|C|)n' + O(\delta n)}$, for some constant $\delta > 0$ that can be arbitrarily small. By \Cref{prop:runtime} the runtime is 
%     \[
%         |C|^{\Lambda(|C|)n' + O(\delta n)} \leq |C|^{\alpha \Lambda(|C|) n + \eps_1\Lambda(|C|)n + O(\delta n)}.
%     \]
%     Choosing sufficiently small values $\eps_1$ and $\delta$ completes the proof.
% }

\section{ Proof of \texorpdfstring{\Cref{thm:optprecisionC}}{} }
\label{apx:discard}

Here the challenge is again to prove the lower bound; the proof of the upper bound on solution probability is trivial and follows the same argument as that in the proof of \Cref{thm:optprecisionC0}.

To prove the lower bound, consider a coefficient set $C = C(d)$ for a fixed constant $d > 1$, $M = O^*(|C|^{(1-\eps)n})$ for a fixed constant $\eps > 0$, and a target $\tau$ satisfying $|\tau| = O(M)$. (We expand our discussion to all $\tau$ such that $|\tau| = o(Mn)$ at the end of the subsection.) Define
\begin{align}
    \label{eq:rho}
    \rho_{n} := \frac{|C|^{n+1/2}}{M\sqrt{2\pi n \kappa \sum_{c \in C} c^{2}}}
    \qquad \text{where} \qquad
    \kappa := \Ex_{\bxi \sim [0: M - 1]} \Big[\frac{\bxi^{2}}{M^{2}}\Big] = \frac{1}{3} - \frac{1}{2M} + \frac{1}{6M^{2}}.
\end{align}
Because $|C|=2d$ and $\sum_{c \in C} c^{2} = \frac{2}{3} d(d+1)(2d+1)$ are fixed constants, we have $\rho_{n} = \Theta(\frac{|C|^{n}}{M \sqrt{n}})$.
Moreover, because $M = O^{*}(|C|^{(1 - \eps)n})$ for a constant $\eps > 0$, we have $\rho_{n} = |C|^{\Omega(n)}$.

Let $\bcalZ := \bcalZ(M, C, \tau)$ be the random variable that counts the number of solution vectors $\vec{c} \in C^{n}$ for a random ATSS instance $\vec{\bx} \sim [0: M - 1]^{n}$ with target value $\tau$. Our proof of the lower bound on solution probability when $M = O^*(|C|^{(1-\eps)n})$ generalizes \cite[Proposition~3.1]{borgs2001phase}, and consists of three parts:
\begin{enumerate}
    \item \Cref{lem:discard-sub1} proves that $\E_{\vec{\bx}}[\bcalZ] = \rho_{n} \cdot (1 \pm o_{n}(1))$.
    \item \Cref{lem:discard-sub2} proves that $\E_{\vec{\bx}}[\bcalZ^{2}] \leq \rho_{n}^{2} \cdot (1 + o_{n}(1))$.
    \item In the proof of \Cref{thm:optprecisionC},
    % \trnote{Changed this to refer to the relevant Theorem in outline.tex. }
    we use the preceding lemmas to show $\bcalZ = \rho_{n} \cdot (1 \pm o_{n}(1))$ with probability $1 - o_{n}(1)$. The bound for $M = O^*(|C|^{(1-\eps)n})$ follows.
\end{enumerate}

\begin{figure}[t]
    \centering
    \subfloat[$f$ with $M=8$, $C = C(2)$.
    \label{subfig:f2}]{ \includegraphics[]{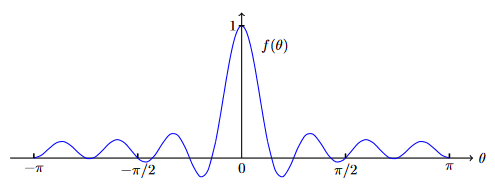} } \\
    
    \subfloat[$g$ with $M=8$, $C = C(2)$.
    \label{subfig:g2}]{ \includegraphics[]{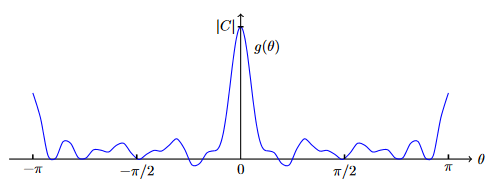} }
    \caption{Example plots of $f$ and $g$. As $|C|$ increases, so does the complexity of the oscillation.}
    \label{fig:f-g2}
\end{figure}

The first two steps are accomplished by expressing $\bcalZ$ as an integral written in terms of a certain function $g: [-\pi,\pi] \to \R$ to the $n$th power. In both cases, we show that the mass of $g^{n}$ is highly concentrated in a region near the origin and tightly bound the value of the function in this region. In particular, we will employ the function $f: [-\pi,\pi] \to \R$,
\begin{align}
    \label{eq:f-def}
    f(\theta) := \Ex_{\bxi \sim [0: M - 1]}[\cos(\theta \bxi)]
    = \frac{1}{M} \sum_{j \in [0: M - 1]} \cos (j \theta)
    = \frac{1}{M}\Big(\frac{\sin((M - 1/2)\theta)}{2 \sin(\theta/2)} +\frac{1}{2}\Big),
\end{align}
which is defined almost identically\footnote{The difference occurs because \cite{borgs2001phase} consider inputs drawn uniformly from the set $[M]$, while we consider inputs drawn uniformly from the set $[0: M - 1]$.} to the function $f$ that occurs in the proof of \cite[Proposition 3.1]{borgs2001phase}. The last equality corresponds to \cite[Equation~(3.11)]{borgs2001phase}.

\begin{proof}[Proof of \Cref{thm:optprecisionC} (Assuming \Cref{lem:discard-sub1,lem:discard-sub2})]
    We have% \rnote{added an extra line to make it easier for me to follow}
    \begin{align*}
        \E_{\vec{\bx}}\Big[\big|\bcalZ / \rho_{n} - 1\big|\Big]
        & \leq \sqrt{\E_{\vec{\bx}}\Big[\big(\bcalZ / \rho_{n} - 1\big)^{2}\Big]} \\
        & = \frac{1}{\rho_{n}} \sqrt{\E_{\vec{\bx}}[\bcalZ^{2}] - \E_{\vec{\bx}}[\bcalZ]^{2}   + \Big(\E_{\vec{\bx}}[\bcalZ] - \rho_{n}\Big)^{2}} \\
        & \leq \frac{1}{\rho_{n}} \sqrt{\E_{\vec{\bx}}[\bcalZ^{2}] - \E_{\vec{\bx}}[\bcalZ]^{2}} + \frac{1}{\rho_{n}} \Big|\E_{\vec{\bx}}[\bcalZ] - \rho_{n}\Big| \\
        % & \leq \frac{o(\rho_{n})}{\rho_{n}} + \frac{o(\rho_{n})}{\rho_{n}} \\
        & = o_{n}(1).
    \end{align*}
    Here the first step holds since any random variable $\bm{\mathcal{X}}$ satisfies that $\E[\bm{\mathcal{X}}]^{2} \leq \E[\bm{\mathcal{X}}^2]$, the second step is elementary algebra, the third step holds since $\sqrt{a + b} \leq \sqrt{a} + \sqrt{b}$ for any $a, b \in \RR$, and the last step uses \Cref{lem:discard-sub1} that $\E_{\vec{\bx}}[\bcalZ] = \rho_{n} \cdot (1 \pm o_{n}(1))$ and \Cref{lem:discard-sub2} that $\E_{\vec{\bx}}[\bcalZ^{2}] \leq \rho_{n}^{2} \cdot (1 \pm o_{n}(1))$.
    
    Our instance has a solution unless $|\bcalZ / \rho_{n} - 1| > 1$, which occurs with probability at most $o_n(1)$ by Markov's inequality. Accordingly, for $C = C(d)$, $d > 1$, $|\tau| = O(M)$, and $M = O^*(|C|^{(1-\eps)n})$ for any constant $\eps > 0$, we have
    \[
        \Pr_{\vec{\bx}}\big[\exists \vec{c} \in C^n : \vec{\bx} \cdot \vec{c} = \tau\big] = 1 - o_n(1).
    \]
   
    It remains to consider the full range of targets $\tau$.  Consider a randomly sampled GSS instance as above but with $|\tau| = o(Mn)$. Without loss of generality, consider $\tau > 0$ and consider the experiment in which we sample input elements one by one and assign each the coefficient $-1$. Each step of this process effectively creates a new instance with one fewer input element and a smaller target $\tau'$. It is a simple exercise to show that when $\tau = o(Mn)$, our new instance satisfies $\tau' = O(M)$ with high probability after $o(n)$ steps. This establishes the lower bound on the probability that a solution exists in \Cref{thm:optprecisionC}.
    
    Finally, we consider the upper bound on solution probability. For any coefficient vector $\vec{c} \in C^n$, we observe that it solves at most a $1 / M$ fraction of GSS instances. This follows from the fact that for any fixed set of $n-1$ input elements in $[0:M-1]^{n-1}$, there is at most one choice for the last element such that $\vec{c}$ solves the instance. Union-bounding over all coefficient vectors yields that the number of instances that admit any solution is at most $|C|^n \cdot M^{n-1}$. Dividing by the size of the instance space yields
    \[
        \Pr_{\vec{\bx}}\big[\exists \vec{c} \in C^n : \vec{\bx} \cdot \vec{c} = \tau\big] \leq \frac{|C|^n}{M}. \qedhere
    \]
\end{proof}

% We also define the functions
% \rnote{Since the lemma statements don't include $g$ or $G$, maybe we can just define them within the lemmas that use them. Also since the lemmas are quite long perhaps we could have a subsection or subsubsection for each one?}
% \yjnote{Maybe use \red{$G(\theta_{1}, \theta_{2})$} both for notational brevity and for distinguishing it from $g(\theta)$? I go that route in what follows.}

\subsection{ Expectation of \texorpdfstring{\boldmath$\mathcal{Z}$}{} }
\label{sec:discard-ez}

\begin{lemma}[Expectation of $\bcalZ$]
    \label{lem:discard-sub1}
    Fix a coefficient set $C = C(d) = \{\pm 1, \pm 2, \dots, \pm d\}$ for some integer $d > 1$, $M = O^*(|C|^{(1 - \eps)n})$ for some constant $\eps > 0$, and a target $\tau = O(M)$. For $\rho_n$ and $f$ defined as in \cref{eq:rho} and \cref{eq:f-def}, we have
    \[
        \E_{\vec{\bx}}[\bcalZ]
        = \frac{1}{2\pi} \int_{-\pi}^{\pi} \cos(\tau\theta) \cdot g(\theta)^{n} \cdot \d\theta
        = \rho_{n} \cdot (1 \pm o_{n}(1)),
    \]
    where the function $g(\theta) := \sum_{c \in C} f(c\theta)$.
\end{lemma}

% \rnote{The proof of the lemma is relatively long (about 3 pages) - can we break it into a few claims somehow?}
% \trnote{ Yaonan's volunteered to do this. }

    For any choice of $\vec{\bx} \sim [0: M - 1]^{n}$, we observe that
    \begin{align*}
        \bcalZ
        & = \sum_{\vec{c} \in C^{n}} \frac{1}{2\pi} \int_{-\pi}^{\pi} e^{\i\theta(\vec{\bx} \cdot \vec{c} - \tau)} \cdot \d\theta \\
        & = \frac{1}{2\pi} \int_{-\pi}^{\pi} e^{-\i \theta \tau} \sum_{\vec{c} \in C^{n}} \Big(\prod_{j \in [n]} e^{\i \theta c_j \bx_j}\Big) \cdot \d\theta \\
        & = \frac{1}{2\pi} \int_{-\pi}^{\pi} e^{-\i \theta \tau} \prod_{j \in [n]} \Big(\sum_{c \in C} e^{\i \theta c \bx_j}\Big) \cdot \d\theta,
    \end{align*}
    where the first step chooses an integral to act as an indicator for the event $\{\vec{\bx} \cdot \vec{c} = \tau\}$, the second step factors the dot product,\footnote{We have assumed $C = C(d) = \{\pm 1, \pm 2, \dots, \pm d\}$, so we need not worry about excluding the $\vec{\mu}(C)$ solution. However, we can generalize this to the case where $\vec{\mu}(C) \in C^{n}$ by subtracting one from this equation.} and the last step exchanges the sum and the product (which is enabled by the special form of the integrand).
    
    We note that $\sum_{c \in C} e^{\i c \theta \bx_j}
    = \sum_{c \in [d]} (e^{\i c \theta \bx_j} + e^{-\i c \theta \bx_j})
    = \sum_{c \in [d]} 2 \cos(c \theta \bx_j)
    = \sum_{c \in C} \cos(c \theta \bx_j)$,
    by the fact that $C = C(d) = \{\pm 1, \pm 2, \dots, \pm d\}$, the Euler formula, and cancelling sines.
    This identity allows us to simplify our equation for $\bcalZ$ to
    \begin{align}
        \label{eq:z-identity}
        \bcalZ = \frac{1}{2\pi} \int_{-\pi}^{\pi} \cos(\tau\theta) \prod_{j \in [n]} \Big(\sum_{c \in C} \cos(c \theta \bx_j)\Big) \cdot \d\theta.
    \end{align}
    Since all the $\bx_j \sim [0: M - 1]$ are i.i.d., we can break down the expectation $\E_{\vec{\bx}}[\bcalZ]$ as follows:
    \begin{align*}
        \E_{\vec{\bx}}[\bcalZ]
        & = \frac{1}{2\pi} \int_{-\pi}^{\pi} \cos(\tau\theta) \cdot \bigg(\Ex_{\bxi \sim [0: M - 1]} \Big[\sum_{c \in C} \cos(c \theta \bxi)\Big] \bigg)^{n} \d\theta \\
        & = \frac{1}{2\pi} \int_{-\pi}^{\pi} \cos(\tau\theta) \cdot g(\theta)^{n} \cdot \d\theta,
    \end{align*}
    where the last step follows because $f(c \theta) = \Ex_{\bxi \sim [0: M - 1]}[\cos(c \theta \bxi)]$ (see \eqref{eq:f-def}) and $g(\theta) = \sum_{c \in C} f(c\theta)$. Later we will see that the mass of this integral is highly concentrated around the origin%(contrast \Cref{cla:discard-sub1-1} with \Cref{cla:discard-sub1-2,cla:discard-sub1-3})
    , where $\cos(\tau\theta) \cdot g(\theta)^{n} \approx |C|^{n}$.
    
    % We begin by showing that, away from the origin, $g(\theta)$ is significantly less than $|C|$.
    
    % Inspired by \cite{borgs2001phase}, we will show that the mass of this integral is mostly very close to $0$ and leverage this fact to show tight bounds on the expectation. Our process will be slightly more involved because $g$ is a more complicated function than $f$. In outline, we will show:
    % \begin{enumerate}
    %     \item For any integer $c$, $f(cx)$ is less than $1$ for all but about a $1/M$ fraction of $[-\pi, \pi]$.
    %     \item $g(x)$ is close to $|C|$ when $|x|$ is less than about $1/(|C|M)$. Elsewhere, there exists $\eps$ such that $g(x) < |C|^{(1-\eps)}$ for all but about a $1/M$-fraction of the interval $[\pi, \pi]$. Moreover, $|C|^{(1-\eps)n}$ on the remainder.
    % \end{enumerate}
    % From the second fact we will conclude that our expectation $\E_{\bx}[Z]$ is like $|C|^{n} / M$.
    
    Fix any $a$ and $b$ satisfying $|C|^{2} / \eps < a < b$. In the rest of this proof, let us consider sufficiently large $n \geq 6$ and $M \geq 20$ such that
    \begin{align}
        \label{eq:M-bound}
        \sqrt{64\ln(n) / n} \leq 1 \leq |C|
        \qquad \mbox{and} \qquad
        M \geq \max \Big\{8d \cdot b, ~ \frac{a \cdot b}{2b - 2a}\Big\}.
    \end{align}
    Then the quantity $b_{0}$ below is well defined (as $0 \leq b / (2M) \leq 1 / d \leq 1$) and satisfies $b \leq b_{0} \leq (\pi / 2) \cdot b$.
    \begin{align}
        \label{eq:b0}
        b_{0} := 2M \cdot \sin^{-1}\Big(\frac{b}{2M}\Big).
    \end{align}

    To evaluate the expectation $\E_{\vec{\bx}}[\bcalZ]$, we readopt the approach by \cite{borgs2001phase} and split its formula into the following three parts. (Based on \eqref{eq:M-bound} and \eqref{eq:b0}, we notice that $\sqrt{64\ln(n) / n} \leq |C| < b \leq b_{0}$ and $\frac{b_{0}}{M} \leq \frac{(\pi / 2) \cdot b}{8d \cdot b} \leq \frac{\pi}{16}$. Hence, for each of these parts, the interval of integration is well defined.)
    \begin{align}
        \E_{\vec{\bx}}[\bcalZ]
        & = \frac{1}{2\pi} \int_{|\theta| \leq \frac{\sqrt{64\ln(n) / n}}{M}} \cos(\tau\theta) \cdot g(\theta)^{n} \cdot \d\theta
        \label{eq:ez-part1}\tag{Part~I} \\
        & \qquad + \frac{1}{2\pi} \int_{|\theta| \in [\frac{\sqrt{64\ln(n) / n}}{M},\; \frac{b_{0}}{M}]} \cos(\tau\theta) \cdot g(\theta)^{n} \cdot \d\theta
        \label{eq:ez-part2}\tag{Part~II} \\
        & \qquad + \frac{1}{2\pi} \int_{|\theta| \in [\frac{b_{0}}{M},\; \pi]} \cos(\tau\theta) \cdot g(\theta)^{n} \cdot \d\theta
        \label{eq:ez-part3}\tag{Part~III}.
    \end{align}
    Below let us quantify \eqref{eq:ez-part1}, \eqref{eq:ez-part2}, and \eqref{eq:ez-part3} one by one.

    \begin{claim}
    \label{cla:discard-sub1-1}
    $|\eqref{eq:ez-part1}| = \rho_{n}(1 \pm o(1))$.
    \end{claim}
    
    \begin{proof}
    Following \cite{borgs2001phase}, we set $y := M \theta$ and consider $|y| \leq \sqrt{64\ln(n) / n} = o_{n}(1)$ small enough. In this range, we observe that
    % \rnote{Should we quantify this more precisely the way we do for \Cref{cla:discard-sub2-1}?}
    \begin{align}
        f(\theta) = f(y / M)
        & = \Ex_{\bxi \sim [0: M - 1]} \big[\cos(\bxi(y/M))\big]
        \nonumber \\
        & = \Ex_{\bxi \sim [0: M - 1]} \bigg[1 - \frac{\bxi^{2}}{2} (y / M)^{2} \pm O(y^{4})\bigg]
        \nonumber \\
        & = 1 - \frac{\kappa}{2} y^{2} \pm O(y^4)
        \nonumber \\
        & = \Big(1 - \frac{\kappa}{2} y^{2}\Big) \cdot (1 \pm O(y^4)),
        \label{eq:fbound3}
    \end{align}
    where the second step uses the Taylor series $\cos(z) = 1 - \frac{1}{2}z^{2} \pm O(z^{4})$, the third step follows from the definition $\kappa = \Ex_{\bxi \sim [0: M - 1]} [\frac{\bxi^{2}}{M^{2}}]$, and the last step converts the (additive) error term into a multiplicative form.
    % \yjnote{Regarding the \green{$\pm$} sign above -- indeed the $O(y^{4})$ quantity must be positive. But we still prefer ``\green{$\pm$}'', to simplify the calculation in what follows. We go that route in what follows.}
    We can write $g(y / M) = \sum_{c \in C} f(c y / M)$ in a similar form:
    \begin{align*}
        g(y / M)
        & = \Big(|C| - \frac{\kappa}{2} \Big(\sum_{c \in C} c^{2}\Big) y^{2}\Big) \cdot (1 \pm O(y^4)) \\
        & = |C| \cdot \exp\bigg(-\frac{\kappa}{2|C|} \Big(\sum_{c \in C} c^{2}\Big) y^{2}\bigg) \cdot (1 \pm O(y^{4})),
    \end{align*}
    where the last step uses the Taylor series $e^{-z} = 1 - z \pm O(z^{2})$. It follows that
    \begin{align}
        \label{eq:gbound3}
        g(y / M)^{n} = |C|^{n} \cdot \exp\bigg(-\frac{n \kappa}{2|C|} \Big(\sum_{c \in C} c^{2}\Big) y^{2}\bigg) \cdot (1 \pm O(n y^{4})).
    \end{align}
    
    % \rnote{Regarding the green $\green{+}$ sign above - is it clear that the $O(ny^4)$ quantity will be positive or could it be negative? If it could be negative should we write $\green{\pm}$ instead?}
    % \rnote{The ``$\red{\frac{1}{2\pi }}$'' in \Cref{eq:pickle} was ``$\red{\frac{1}{2\pi M}}$'', I removed the $M$ from the denominator}
    
    Based on \eqref{eq:gbound3}, we can evaluate \eqref{eq:ez-part1} as follows:
    \begin{align*}
        \eqref{eq:ez-part1}
        & = \frac{1}{2\pi M} \int_{|y| \leq \sqrt{64\ln(n) / n}} \cos(ty / M) \cdot g(y / M)^{n} \cdot \d y \\
        & = \frac{1}{2\pi M} \int_{|y| \leq \sqrt{64\ln(n) / n}} (1 \pm O(y^{2})) \cdot g(y / M)^{n} \cdot \d y \\
        & = \frac{|C|^{n}}{2\pi M} \int_{|y| \leq \sqrt{64\ln(n) / n}} (1 \pm O(y^{2} + n y^{4})) \cdot \exp\bigg(-\frac{n \kappa}{2|C|} \Big(\sum_{c \in C} c^{2}\Big) y^{2}\bigg) \d y \\
        & = \frac{|C|^{n}}{2\pi M} \cdot (1 \pm o_{n}(1)) \cdot \int_{|y| \leq \sqrt{64\ln(n) / n}} \exp\bigg(-\frac{n \kappa}{2|C|} \Big(\sum_{c \in C} c^{2}\Big) y^{2}\bigg) \d y \\
        & = \frac{|C|^{n}}{2\pi M} \cdot (1 \pm o_{n}(1)) \cdot \sqrt{\frac{2\pi |C|}{n \kappa \sum_{c \in C} c^{2}}} \cdot \erf\left(\sqrt{\frac{32 \kappa}{|C|} \Big(\sum_{c \in C} c^{2}\Big) \ln(n)}\right) \\
        & = \rho_{n} \cdot (1 \pm o_{n}(1)) \cdot \erf\left(\sqrt{\frac{32 \kappa}{|C|} \Big(\sum_{c \in C} c^{2}\Big) \ln(n)}\right).
    \end{align*}
    Here the first step changes variable $y = M \theta$. The second step follows as $\cos(z) = 1 - O(z^{2})$ and $|t| = O(M)$. The third step substitutes \eqref{eq:gbound3}. The fourth step follows as $O(y^{2} + n y^{4}) = o_{n}(1)$ when $|y| \leq \sqrt{64\ln(n) / n}$. The fifth step resolves the integral by using the Gaussian error function $\erf(z)$. And the last step uses the definition of $\rho_{n}$ (see \eqref{eq:rho}).
    
    % \Cref{eq:dill} follows from splitting the interval of integration at $n^{-1/2} \ln(n)$, plugging in the Taylor series $\cos(z) = 1 - \frac{z^{2}}{2} + O(z^4)$ using the assumption $t = O(M)$ on the first component, \green{and applying \eqref{eq:fbound4.0} to the second component}\rnote{I didn't 100\% follow this justification - was the equation reference supposed to say ``and applying \eqref{eq:fbound4.0} to the second component''? I didn't fully see how that would yield it either}. \Cref{eq:gherkin} substitutes \eqref{eq:gbound3}, \green{absorbs the $O(y^{2})$ term into the $o(n^{-1})$ }\rnote{I wasn't sure what this meant exactly or what happened to the ``$1-O(y^{2})$'' part of the first integral}, and evaluates the second integral, \green{which is absorbed by the $o(n^{-1})$ term}. \Cref{eq:half-sour} integrates the function $\exp(-zny^{2})$, where $z = \frac{\kappa \sum_{c \in C} c^{2}}{2|C|}$. We conclude by observing that \green{$\erf(\Omega(\log n)) = 1 - o(n^{-1})$}.\yjnote{$\int x^{2} e^{cx} \d x = e^{cx}\left(\frac{x^{2}}{c}-\frac{2x}{c^{2}}+\frac{2}{c^3}\right)$}
    
    % \tr{ In equation 25, we probably want like $(1 \pm O(\log^{2} n / n))$, which should propagate down. }
    
    To finish the proof of \Cref{cla:discard-sub1-1}, it remains to reason about the Gaussian error function $\erf(z)$. It is known that $|1 - \erf(z)| \leq e^{-z^{2}}$ for any $z \in \R_{\geq 0}$. As a consequence, we have
    \begin{align}
        \label{eq:erf}
        \erf\left(\sqrt{\frac{32 \kappa}{|C|} \Big(\sum_{c \in C} c^{2}\Big) \ln(n)}\right)
        = 1 \pm n^{-32 \kappa |C|^{-1} (\sum_{c \in C} c^{2})}
        = 1 \pm n^{-32 \kappa}
        = 1 \pm o_{n}(1),
    \end{align}
    where the second step holds because $C = C(d) = \{\pm 1, \pm 2, \dots, \pm d\}$ and then $|C|^{-1} (\sum_{c \in C} c^{2}) \geq 1$, and the last step holds because $\kappa = \frac{1}{3} - \frac{1}{2M} + \frac{1}{6M^{2}} = \Omega_{n}(1)$.

    Combining everything together completes the proof of \Cref{cla:discard-sub1-1}.
    \end{proof}
    
    \begin{claim}
    \label{cla:discard-sub1-2}
    $|\eqref{eq:ez-part2}| = o(\rho_{n})$.
    \end{claim}
    
    \begin{proof}
    Once again we set $y := M \theta$. Recall that \eqref{eq:ez-part2} considers the range $|\theta| \in [\frac{\sqrt{64\ln(n) / n}}{M},\; \frac{b_{0}}{M}]$, namely $|y| \in [\sqrt{64\ln(n) / n},\; b_{0}]$. In this range, it turns out that
    \begin{align}
        \label{eq:f-bound4}
        |f(c \theta)| = |f(c y / M)| \leq \frac{1}{\sqrt[n]{n}},
    \end{align}
    for any (nonzero) coefficient $c \in C$. Assuming the truth of \eqref{eq:f-bound4} for the moment, we have
    \begin{align}
        \label{eq:g-bound4}
        |g(\theta)| = |g(y / M)|
        \leq \sum_{c \in C} |f(c y / M)|
        \leq \frac{|C|}{\sqrt[n]{n}},
    \end{align}
    for any $|y| \in [\sqrt{64\ln(n) / n},\; b_{0}]$. Hence, the magnitude of \eqref{eq:ez-part2} is at most
    \begin{align*}
        |\eqref{eq:ez-part2}|
        & = \bigg| \frac{1}{2\pi M} \int_{|y| \in [\sqrt{64\ln(n) / n},\; b_{0}]} \cos(t y / M) \cdot g(y / M)^{n} \cdot \d y \bigg| \\
        & \leq \frac{1}{2\pi M} \cdot 2b_{0} \cdot \Big( \max_{|y| \in [\sqrt{64\ln(n) / n},\; b_{0}]} |g(y / M)| \Big)^{n} \\
        & \leq \frac{1}{2\pi M} \cdot 2b_{0} \cdot \frac{|C|^{n}}{n} \\
        & = \frac{b_{0} |C|^{n}}{\pi M n} \\
        & = o(\rho_{n}),
    \end{align*}
    where the first step changes variables using the equality $y = M \theta$, the third step applies \eqref{eq:g-bound4}, and the last step holds since $b_{0} \leq \frac{\pi}{2} \cdot b$ is a constant (see \eqref{eq:b0}) and $\rho_{n} = \Theta(\frac{|C|^{n}}{M \sqrt{n}})$ (see \eqref{eq:rho}).

    It remains to verify \eqref{eq:f-bound4}. Before doing so, we observe that for any $|y| \leq b_{0}$,
    \begin{align}
        \label{eq:ez1-f-domain}
        \Big|\frac{c y}{2M}\Big|
        \leq \frac{|c| \cdot b_{0}}{2M}
        \leq \frac{|c| \cdot (\pi / 2) \cdot b}{16d \cdot b}
        \leq \frac{|c| \cdot (\pi / 32)}{d}
        \leq \frac{\pi}{32},
    \end{align}
    where the second step holds because $b_{0} \leq (\pi / 2) \cdot b$ (see \eqref{eq:b0}) and $M \geq 8d \cdot b$ (see \eqref{eq:M-bound}), and the last step holds because $c \in C = C(d) = \{\pm 1, \pm 2, \dots, \pm d\}$.

    Thus for any $|y| \leq b_{0}$, we can upper-bound $|f(c \theta)| = |f(c y / M)|$ as follows.
    \begin{align}
        |f(c \theta)| = |f(c y / M)|
        & = \Big|\frac{\sin(c y - c y / (2M))}{2 M \sin(c y / (2M))} + \frac{1}{2M} \Big|
        \nonumber \\
        & = \Big| \frac{\sin(c y)}{2M \tan(c y / (2M))} + \frac{1 - \cos(c y)}{2M} \Big|
        \nonumber \\
        & \leq \Big|\frac{\sin(c y)}{c y}\Big| + \frac{1 - \cos(c y)}{2M}
        \nonumber \\
        & \leq \Big|\frac{\sin(c y)}{c y}\Big| + \frac{1 - \cos(c y)}{40}.
        \label{eq:f-bound5}
    \end{align}
    Here the first step changes variables using the equality $y = M \theta$. The second step applies the identity $\sin(\alpha - \beta) = \sin(\alpha)\cos(\beta) - \sin(\beta)\cos(\alpha)$. The third step holds since $|\tan(z)| \geq |z|$ for any $|z| \leq \frac{\pi}{2}$ (notice that $|\frac{c y}{2M}| \leq \frac{\pi}{4} \leq \frac{\pi}{2}$, see \eqref{eq:ez1-f-domain}). And the last step follows since $M \geq 20$ (see \eqref{eq:M-bound}).

    We prove \eqref{eq:f-bound4} for any nonzero $c \in C$ and any $|y| \in [\sqrt{64\ln(n) / n},\; b_{0}]$ in two cases.
    
    \vspace{.1in}
    \noindent
    {\bf Case~I: $\sqrt{64\ln(n) / n} \leq |y| \leq \frac{\pi}{2|c|}$.}
    We know that $|\sin(z) / z| \leq e^{-z^{2} / 6}$ for any $|z| \leq \frac{\pi}{2}$ and $\cos(z) \geq 1 - \frac{1}{2}z^{2}$ for any $z \in \R$. Applying both facts to \eqref{eq:f-bound5} gives
    \begin{align*}
        |f(c \theta)| = |f(c y / M)|
        & \leq \Big|\frac{\sin(c y)}{c y}\Big| + \frac{1 - \cos(c y)}{40} \\
        & \leq e^{-\frac{1}{6} c^{2} y^{2}} + \frac{1}{80}c^{2} y^{2} \\
        & \leq e^{-\frac{1}{8} c^{2} y^{2}} \\
        & \leq e^{-\frac{1}{8} y^{2}} \\
        & \leq \frac{1}{\sqrt[n]{n}},
    \end{align*}
    where the third step holds since $e^{-z^{2} / 6} + \frac{1}{80}z^{2} \leq e^{-z^{2} / 8}$ for any $|z| \leq \frac{\pi}{2}$, the fourth step holds since $c \in C$ is a nonzero integer, and the last step holds when $|y| \geq \sqrt{64\ln(n) / n}$.
    
    \vspace{.1in}
    \noindent
    {\bf Case~II: $\frac{\pi}{2|c|} \leq |y| \leq b_{0}$.}
    Following \eqref{eq:f-bound5}, in this range we have
    \begin{align*}
        |f(c \theta)| = |f(c y / M)|
        & \leq \Big|\frac{\sin(c y)}{c y}\Big| + \frac{1 - \cos(c y)}{40} \\
        & \leq \frac{|\sin(c y)|}{\pi / 2} + \frac{1 + |\cos(c y)|}{40} \\
        & \leq \frac{2}{\pi} + \frac{1}{20} \\
        & \leq \frac{1}{\sqrt[n]{n}},
    \end{align*}
    where the second step applies $|y| \geq \frac{\pi}{2|c|}$, the third step holds because $|\sin(c y)| \leq 1$ and $|\cos(c y)| \leq 1$, and the last step holds because $\frac{2}{\pi} + \frac{1}{20} \approx 0.6866$ and $\frac{1}{\sqrt[n]{n}} = 1 - o_{n}(1)$.
    
    Combining both cases together gives \eqref{eq:f-bound4}. This completes the proof of \Cref{cla:discard-sub1-2}.
    \end{proof}

    \begin{claim}
    \label{cla:discard-sub1-3}
    $|\eqref{eq:ez-part3}| = o(\rho_{n})$.
    \end{claim}
    
    \begin{proof}
    Recall that \eqref{eq:ez-part3} considers the range $|\theta| \in [\frac{b_{0}}{M},\; \pi]$, in which we have
    \begin{align}
        \label{eq:fbound1}
        |f(\theta)|
        = \Big|\frac{1}{M} \Big(\frac{\sin((M - 1/2)\theta)}{2 \sin(\theta/2)} +\frac{1}{2}\Big)\Big|
        \leq \frac{1}{b} + \frac{1}{2M}
        \leq \frac{1}{a}.
    \end{align}
    Here the first inequality holds as $|\sin((M - 1/2)\theta)| \leq 1$ and $|\sin(\theta / 2)| \geq |\sin(\frac{b_{0}}{2M})| = \frac{b}{2M}$ (see \eqref{eq:b0}), and the second inequality holds as $M \geq \frac{a b}{2b - 2a}$ (see \eqref{eq:M-bound}).

    For any $c \in C$ and $|\theta| \in [\frac{b_{0}}{M},\; \pi]$, we know from \eqref{eq:f-def} that $|f(c\theta)| = |\E_{\bxi \sim [0: M - 1]}[\cos(c \theta \bxi)]| \leq 1$. (Recall that $C = C(d) \supseteq \{\pm 1\}$.) Applying this fact and \eqref{eq:fbound1}, for any $|\theta| \in [\frac{b_{0}}{M},\; \pi]$ we have
    \begin{align}
        \label{eq:gbound2}
        |g(\theta)|
        \leq |f(\theta)| + |f(-\theta)| + \sum_{c \in C \setminus \{\pm 1\}} |f(c \theta)|
        \leq \frac{2}{a} + (|C| - 2)
        = |C|^{(1 - \eps_{1})}
    \end{align}
    %\rnote{Is it worth reminding the reader that the last inequality holds because we have $a \geq C^2/\eps$ (all we need is that $a$ is bounded away from 1 I guess)?}
    for some constant $\eps_{1} := \eps_{1}(|C|, a) \in (0, 1)$, where the last step holds since $a > |C|^{2} / \eps > 1$ (see \eqref{eq:M-bound}). Further, we know that for any integer $c \in C$ the function $f_{c}(\theta) := f(c\theta)$ is $2\pi$-periodic (see \eqref{eq:f-def}). Thus \eqref{eq:fbound1} implies that 
    \begin{align}
        \label{eq:fbound2}
        \Big|\Big\{\theta \in [-\pi, \pi] : |f(c\theta)| > \frac{1}{a} \Big\}\Big| \leq \frac{2b_0}{M}.
    \end{align}
    This, given that $g(\theta) = \sum_{c \in C} f(c \theta)$, further gives
    \begin{align}
        \label{eq:gbound1}
        \Big|\Big\{\theta \in [-\pi, \pi] : |g(\theta)| > \frac{|C|}{a} \Big\}\Big| \leq \frac{2b_0 |C|}{M}.
    \end{align}
    
    By considering separately the subregion on which $|C| / a < |g(\theta)| \leq |C|^{(1 - \eps_{1})}$ and the subregion on which $|g(\theta)| \leq |C| / a$, the magnitude of \eqref{eq:ez-part3} is at most
    % \red{the contribution to \Cref{eq:ez1} from $|\theta| \in [\frac{b_{0}}{M},\pi]$ is at most}\rnote{We ultimately are trying to prove a (near)-equality in  \Cref{lem:discard-sub1}, not a one-sided inequality. So should we write on the LHS below $| \frac{1}{2\pi} \int_{|\theta| \in [\frac{b_{0}}{M}, \pi]} \cos(\tau\theta) g(\theta)^{n} \d\theta|$ rather than $ \frac{1}{2\pi} \int_{|\theta| \in [\frac{b_{0}}{M}, \pi]} \cos(\tau\theta) g(\theta)^{n} \d\theta$ --- it seems a priori the quantity could be negative, and if it had a huge magnitude and were negative that would be a problem}
    \begin{align*}
        |\eqref{eq:ez-part3}|
        & = \bigg| \frac{1}{2\pi} \int_{|\theta| \in [\frac{b_{0}}{M},\; \pi]} \cos(\tau\theta) \cdot g(\theta)^{n} \cdot \d\theta \bigg| \\
        & \leq \frac{1}{2\pi} \int_{|\theta| \in [\frac{b_{0}}{M},\; \pi]} |g(\theta)|^{n} \cdot \d\theta \\
        & \leq \frac{1}{2\pi} \bigg(|C|^{(1-\eps_{1})n} \cdot \frac{2b_0 |C|}{M} + \Big(\frac{|C|}{a}\Big)^{n} \cdot 2\pi\bigg) \\
        & = O\Big(\frac{|C|^{(1-\eps_1)n}}{M}\Big) + O\Big(\frac{1}{|C|^{n}}\Big),
    \end{align*}
    where the third step applies \eqref{eq:gbound2} and \eqref{eq:gbound1}, and the fourth step holds since (for the first term) both $b_{0}$ and $|C|$ are constants and (for the second term) $a > |C|^{2} / \eps > |C|^{2}$ (see \eqref{eq:M-bound}).
    
    Recall \eqref{eq:rho} that $\rho_{n} = |C|^{\Omega(n)}$.
    % \rnote{of course the lower bound on $\rho_n$ is much stronger. Maybe we should comment earlier on, right after \Cref{eq:rho} when we state that $\rho_n$ is $\Theta(|C|^n/M\sqrt{n}$, that $\rho_n$ is $\exp(\Omega(n))$?}
    So the above two terms are upper-bounded by $O(|C|^{(1-\eps_1)n} / M) = \rho_{n} \cdot O(\sqrt{n} \cdot |C|^{-\eps_{1} n}) = o(\rho_{n})$ and $O(|C|^{-n}) = o_{n}(1) = o(\rho_{n})$, respectively.
    
    This completes the proof of \Cref{cla:discard-sub1-3}.
    \end{proof}

% {\color{red}Maybe continue from (25) as follows:
% \begin{align*}
% &\le \frac{1}{2\pi} \int_{|\theta|\in [\frac{b_{0}}{M},\pi]}
% \left[\min\left(\frac{|C|}{a},\frac{|C|\alpha}{M|\theta|}\right)\right]^{n} \d\theta \\
% &\le \frac{2\alpha a}{M}\cdot \left(\frac{1}{a}\right)^{n}+
% \int_{|\theta|\ge \alpha a/M}
% \left(\frac{|C|\alpha}{M|\theta|}\right)^{n} \d\theta = O\left(\frac{|C|^{n}}{Ma^{n}}\right).
% \end{align*}
% }
    
\begin{proof}[Proof of \Cref{lem:discard-sub1}]
    Putting \Cref{cla:discard-sub1-3,cla:discard-sub1-2,cla:discard-sub1-1} together completes the proof of \Cref{lem:discard-sub1} as follows:
    \begin{align*}
        \E_{\vec{\bx}}[\bcalZ]
        & ~ = ~ \eqref{eq:ez-part1} ~ + ~ \eqref{eq:ez-part2} ~ + ~ \eqref{eq:ez-part3} \\
        & ~\in ~ \eqref{eq:ez-part1} ~ \pm ~ (|\eqref{eq:ez-part2}| ~ + ~ |\eqref{eq:ez-part3}|) \\
        & ~ = ~ \rho_{n} \cdot (1 \pm o_{n}(1)).
        \qedhere
    \end{align*}
\end{proof}

\subsection{Upper Bound on the Second Moment of \texorpdfstring{\boldmath{$\mathcal{Z}$}}{}}
\label{sec:discard-ez2}

\begin{lemma}[Second Moment of {\boldmath{$\mathcal{Z}$}}]
    \label{lem:discard-sub2}
     Fix a coefficient set $C = C(d) = \{\pm 1, \pm 2, \dots, \pm d\}$ for some integer $d > 1$, $M = O^*(|C|^{(1 - \eps)n})$ for some constant $\eps > 0$, and a target $\tau = O(M)$. For $\rho_n$ and $f$ defined as in \cref{eq:rho} and \cref{eq:f-def}, we have
    \[
        \E_{\vec{\bx}}[\bcalZ^{2}]
        \leq \frac{1}{4\pi^{2}} \int_{-\pi}^{\pi} \int_{-\pi}^{\pi} |G(\theta_{1}, \theta_{2})|^{n} \cdot \d\theta_{1} \d\theta_{2}
        \leq \rho_{n}^{2}\cdot (1 + o_{n}(1)).
    \]
    where the function $G(\theta_{1}, \theta_{2}) := \sum_{(c_{1}, c_{2}) \in C^{2}} f(c_{1} \theta_{1} + c_{2} \theta_{2})$ for any $(\theta_{1}, \theta_{2}) \in [-\pi, \pi]^{2}$.
\end{lemma}

    For any choice of $\vec{\bx} \sim [0: M - 1]^{n}$, we deduce from \eqref{eq:z-identity} that
    \begin{align*}
        \bcalZ^{2}
        & = \bigg(\frac{1}{2\pi} \int_{-\pi}^{\pi} \cos(\tau\theta) \prod_{j \in [n]} \Big(\sum_{c \in C} \cos(c \theta \bx_j)\Big) \cdot \d\theta\bigg)^{2} \\
        & = \frac{1}{4\pi^{2}} \int_{-\pi}^\pi \int_{-\pi}^\pi \cos(t\theta_{1}) \cos(t\theta_{2}) \prod_{j \in [n]} \Big(\sum_{(c_{1}, c_{2}) \in C^{2}} \cos(c_{1} \theta_{1} \bx_j) \cos(c_{2} \theta_{2} \bx_j)\Big) \cdot \d\theta_{1} \d\theta_{2} \\
        & \leq \frac{1}{4\pi^{2}} \int_{-\pi}^\pi \int_{-\pi}^\pi \prod_{j \in [n]} \bigg|\sum_{(c_{1}, c_{2}) \in C^{2}} \cos(c_{1} \theta_{1} \bx_j) \cos(c_{2} \theta_{2} \bx_j)\bigg| \cdot \d\theta_{1} \d\theta_{2} \\
        & = \frac{1}{4\pi^{2}} \int_{-\pi}^\pi \int_{-\pi}^\pi \prod_{j \in [n]} \bigg|\sum_{(c_{1}, c_{2}) \in C^{2}} \Big(\cos(c_{1} \theta_{1} \bx_j) \cos(c_{2} \theta_{2} \bx_j) - \sin(c_{1} \theta_{1} \bx_j) \sin(c_{2} \theta_{2} \bx_j)\Big)\bigg| \cdot \d\theta_{1} \d\theta_{2} \\
        & = \frac{1}{4\pi^{2}} \int_{-\pi}^\pi \int_{-\pi}^\pi \prod_{j \in [n]} \bigg|\sum_{(c_{1}, c_{2}) \in C^{2}} \cos(c_{1} \theta_{1} \bx_j + c_{2} \theta_{2} \bx_j)\bigg| \cdot \d\theta_{1} \d\theta_{2},
    \end{align*}
    where the last second step holds since $C = C(d) = \{\pm 1, \pm 2, \dots, \pm d\}$ is a symmetric set and $\sin(z)$ is an odd function, and the last step uses the identity $\cos(\alpha + \beta) = \cos(\alpha) \cos(\beta) - \sin(\alpha) \sin(\beta)$.
    
    Since all the $\bx_j \sim [0: M - 1]$ are i.i.d., the following holds for the second moment $\E_{\vec{\bx}}[\bcalZ^{2}]$:
    \begin{align}
        \E_{\vec{\bx}}[\bcalZ^{2}]
        & \leq \frac{1}{4\pi^{2}} \int_{-\pi}^\pi \int_{-\pi}^\pi \bigg|\sum_{(c_{1}, c_{2}) \in C^{2}} \Ex_{\bxi \sim [0: M - 1]}\Big[\cos(c_{1} \theta_{1} \bxi + c_{2} \theta_{2} \bxi)\Big]\bigg|^{n} \cdot \d\theta_{1} \d\theta_{2}
        \nonumber \\
        % & = \frac{1}{4\pi^{2}} \int_{-\pi}^\pi \int_{-\pi}^\pi \bigg|\sum_{(c_{1}, c_{2}) \in C^{2}} \Ex_{\bxi \sim [0: M - 1]}\Big[\frac{\cos((c_{1} \theta_{1} + c_{2} \theta_{2}) \bxi) + \cos((c_{1} \theta_{1} - c_{2} \theta_{2}) \bxi)}{2}\Big]\bigg|^{n} \cdot \d\theta_{1} \d\theta_{2}
        % \nonumber \\
        & = \frac{1}{4\pi^{2}} \int_{-\pi}^\pi \int_{-\pi}^\pi \bigg|\sum_{(c_{1}, c_{2}) \in C^{2}} f(c_{1} \theta_{1} + c_{2} \theta_{2})\bigg|^{n} \cdot \d\theta_{1} \d\theta_{2}
        \nonumber \\
        & = \frac{1}{4\pi^{2}} \int_{-\pi}^\pi \int_{-\pi}^\pi |G(\theta_{1}, \theta_{2})|^{n} \cdot \d\theta_{1} \d\theta_{2},
        \label{eq:z2-1}
    \end{align}
    where the second step
    % applies the identity $\cos(\alpha) \cos(\beta) = \frac{1}{2}\cos(\alpha + \beta) + \frac{1}{2}\cos(\alpha - \beta)$, and the third step
    uses the definition $f(\theta) = \Ex_{\bxi \sim [0: M - 1]}[\cos(\theta \bxi)]$ (see \eqref{eq:f-def}).
    
    \begin{remark}
    We will show that the mass of \eqref{eq:z2-1} is concentrated around the origin, where $|G(\theta_{1}, \theta_{2})|^{n} \approx |C|^{2n}$. Notably, this is \emph{not} true when $C = C(1)$, as considered in \cite{borgs2001phase}. In the $C = C(1)$ case, mass is also concentrated at the points $[\pm \pi, \pm \pi]$, the corners of the region of integration. This is because the function $G(\theta_1, \theta_2) = \sum_{(c_1, c_2) \in C^2} f(c_1\theta_1 + c_2\theta_2)$ consists of four simple summands that interfere constructively at this point. However, when $C = C(d)$ for any $d > 1$, interference from other summands ensures that this phenomenon does not happen.
    \end{remark}

    We reuse the same constants $a$, $b$ and $b_{0}$ (namely $|C|^{2} / \eps < a < b$ and $b_{0} = 2M \cdot \sin^{-1}(\frac{b}{2M})$) introduced in the proof of \Cref{lem:discard-sub1}. Define the subregions $R'_{\theta} \subseteq R''_{\theta} \subseteq [-\pi, \pi]^{2}$ by letting
    \begin{align*}
        R'_{\theta} := \Big[-\frac{\sqrt{64\ln(n)}}{M \sqrt{n}},\; \frac{\sqrt{64\ln(n)}}{M \sqrt{n}}\Big]^{2}
        \qquad \text{and} \qquad
        R''_{\theta} := \Big[-\frac{b_{0}}{M},\; \frac{b_{0}}{M}\Big]^{2}.
    \end{align*}
    Similar to the approach of \cite{borgs2001phase}, we split the formula \eqref{eq:z2-1} into the following three parts.
    \begin{align}
        \E_{\vec{\bx}}[\bcalZ]
        & \leq \frac{1}{4\pi^{2}} \iint_{(\theta_{1}, \theta_{2})\in R'_{\theta}} |G(\theta_{1}, \theta_{2})|^{n} \cdot \d \theta_{1} \d \theta_{2}
        \label{eq:ez2-part1}\tag{Part~1} \\
        & \qquad + \frac{1}{4\pi^{2}} \iint_{(\theta_{1}, \theta_{2})\in R''_{\theta} \setminus R'_{\theta}} |G(\theta_{1}, \theta_{2})|^{n} \cdot \d \theta_{1} \d \theta_{2}
        \label{eq:ez2-part2}\tag{Part~2} \\
        & \qquad + \frac{1}{4\pi^{2}} \iint_{(\theta_{1}, \theta_{2})\in [-\pi,\; \pi]^{2} \setminus R''_{\theta}} |G(\theta_{1}, \theta_{2})|^{n} \cdot \d \theta_{1} \d \theta_{2}
        \label{eq:ez2-part3}\tag{Part~3}.
        % \\
        % & \qquad + \frac{1}{4\pi^{2}} ???
        % \label{eq:ez2-part13}\tag{Part~4}.
    \end{align}
    We evaluate \eqref{eq:ez2-part1} and \eqref{eq:ez2-part2} respectively in \Cref{cla:discard-sub2-1,cla:discard-sub2-2}. \eqref{eq:ez2-part3} is rather tricky; we first prove three auxiliary results (\Cref{cla:discard-sub2-3,cla:discard-sub2-4,cla:discard-sub2-5}) and then evaluate it in \Cref{cor:discard-sub2}.

    \begin{claim}
    \label{cla:discard-sub2-1}
    $\eqref{eq:ez2-part1} = \rho_{n}^{2} \cdot (1 \pm o_{n}(1))$.
    \end{claim}

    % \rnote{Should we stipulate something about the range of values of $y_1,y_2$ in the following sequence of equalities here?}

    \begin{proof}
        (This proof is similar to the proof of \Cref{cla:discard-sub1-1}.)
        Following \cite{borgs2001phase}, we set $y_{1} := M \theta_{1}$ and $y_{2} := M \theta_{2}$. Thus we consider small enough $|y_{1}|, |y_{2}| \leq \sqrt{64\ln(n) / n} = o_{n}(1)$. In this range,

        \begin{align*}
            G(y_{1} / M, y_{2} / M)
            & = \sum_{(c_{1}, c_{2}) \in C^{2}} f(c_{1} y_{1} / M + c_{2} y_{2} / M) \\
            & = \sum_{(c_{1}, c_{2}) \in C^{2}} \Big(1 - \frac{\kappa}{2} (c_{1} y_{1} + c_{2} y_{2})^{2}\Big) \cdot (1 \pm O(y_{1}^{4} + y_{2}^{4})) \\
            & = \Big(|C|^{2} - |C| \cdot \frac{\kappa \sum_{c \in C} c^{2}}{2} \cdot (y_{1}^{2} + y_{2}^{2})\Big) \cdot (1 \pm O(y_{1}^{4} + y_{2}^{4})) \\
            & = |C|^{2} \cdot \exp\bigg(-\frac{\kappa \sum_{c \in C} c^{2}}{2|C|} (y_{1}^{2} + y_{2}^{2})\bigg) \cdot (1 \pm O(y_{1}^{4} + y_{2}^{4})),
        \end{align*}
        where the first step changes variables $y_{1} = M \theta_{1}$ and $y_{2} = M \theta_{2}$, the second step holds since we have $f(y / M) = (1 - \frac{\kappa}{2} y^{2}) \cdot (1 \pm O(y^4))$ for small $|y|$ (see \eqref{eq:fbound3}), the third step is elementary algebra (notice that $C = C(d) = \{\pm 1, \pm 2, \dots, \pm d\}$ is a symmetric set, so the crossing terms $2 c_{1} c_{2} y_{1} y_{2}$ get cancelled), and the last step uses the approximation $e^{-z} = (1 - z) \cdot (1 \pm O(z^{2}))$ for small $|z|$.

        As a consequence, the following holds for any $|y_{1}|, |y_{2}| \leq \sqrt{64\ln(n) / n}$.
        \begin{align}
            |G(y_{1} / M, y_{2} / M)|^{n}
            & = |C|^{2n} \cdot \exp\bigg(-\frac{n \kappa \sum_{c \in C} c^{2}}{2|C|} (y_{1}^{2} + y_{2}^{2})\bigg) \cdot (1 \pm O(n y_{1}^{4} + n y_{2}^{4}))
            \nonumber \\
            & = |C|^{2n} \cdot \exp\bigg(-\frac{n \kappa \sum_{c \in C} c^{2}}{2|C|} (y_{1}^{2} + y_{2}^{2})\bigg) \cdot (1 \pm o_{n}(1)).
            \label{eq:G-bound1}
        \end{align}
        This allows us to bound the magnitude of \eqref{eq:ez2-part1} as follows:
        \begin{align*}
            \eqref{eq:ez2-part1}
            & = \frac{1}{4\pi^{2} M^{2}} \iint_{|y_{1}|, |y_{2}| \leq \sqrt{64\ln(n) / n}} |G(y_{1} / M, y_{2} / M)|^{n} \cdot \d y_{1} \d y_{2} \\
            & = \frac{|C|^{2n}}{4\pi^{2} M^{2}} \cdot (1 \pm o_{n}(1)) \cdot \iint_{|y_{1}|, |y_{2}| \leq \sqrt{64\ln(n) / n}} \exp\bigg(-\frac{n \kappa \sum_{c \in C} c^{2}}{2|C|} (y_{1}^{2} + y_{2}^{2})\bigg) \cdot \d y_{1} \d y_{2} \\
            & = \frac{|C|^{2n}}{4\pi^{2} M^{2}} \cdot (1 \pm o_{n}(1)) \cdot \frac{2\pi |C|}{n \kappa \sum_{c \in C} c^{2}} \cdot \erf\left(\sqrt{\frac{32 \kappa}{|C|} \Big(\sum_{c \in C} c^{2}\Big) \ln(n)}\right)^{2} \\
            & = \frac{|C|^{2n}}{4\pi^{2} M^{2}} \cdot (1 \pm o_{n}(1)) \cdot \frac{2\pi |C|}{n \kappa \sum_{c \in C} c^{2}} \cdot (1 \pm o_{n}(1))^{2} \\
            & = \rho_{n}^{2} \cdot (1 \pm o_{n}(1)).
        \end{align*}
        Here the first step changes variables $y_{1} = M \theta_{1}$ and $y_{2} = M \theta_{2}$. The second step substitutes \eqref{eq:G-bound1}. The third step resolves the integral by using the Gaussian error function $\erf(z)$. The fourth step uses the approximation of $\erf(z)$ in \eqref{eq:erf}, and the last step uses the definition of $\rho_{n}$ (see \eqref{eq:rho}).
        
        This completes the proof of \Cref{cla:discard-sub2-1}.
    \end{proof}

    \begin{claim}
        \label{cla:discard-sub2-2}
        $\eqref{eq:ez2-part2} = o(\rho_{n}^{2})$.
    \end{claim}
    
    \begin{proof}
    (This proof is similar to the proof of \Cref{cla:discard-sub1-2}.)
    Once again, we set $y_{1} := M \theta_{1}$ and $y_{2} := M \theta_{2}$. Then \eqref{eq:ez2-part2} corresponds the range $(y_{1}, y_{2}) \in R''_{y} \setminus R'_{y}$, where
    \begin{align*}
        R'_{y} := [-\sqrt{64\ln(n) / n},\; \sqrt{64\ln(n) / n}]^{2}
        \qquad \mbox{and} \qquad
        R''_{y} := [-b_{0},\; b_{0}]^{2}.
    \end{align*}
    For each pair $(c_{1}, c_{2}) \in C^{2}$, we define the function $G_{c_{1}, c_{2}}(y_{1}, y_{2}) := \frac{1}{2}(f(\frac{c_{1} y_{1}}{M} + \frac{c_{2} y_{2}}{M}) + f(\frac{c_{2} y_{1}}{M} - \frac{c_{1} y_{2}}{M}))$. We will show that for any pair $(c_{1}, c_{2}) \in C^{2}$ and any point $(y_{1}, y_{2}) \in R''_{y} \setminus R'_{y}$, we have
    \begin{align}
        \label{eq:G-bound2}
        |G_{c_{1}, c_{2}}(y_{1}, y_{2})| \leq \frac{1}{\sqrt[n]{n^{2}}}.
    \end{align}
    Assuming \eqref{eq:G-bound2} for the moment, for any point $(y_{1}, y_{2}) \in R''_{y} \setminus R'_{y}$ we have
    \begin{align}
        |G(\theta_{1}, \theta_{2})|
        = |G(y_{1} / M, y_{2} / M)|
        & = \bigg|\sum_{(c_{1}, c_{2}) \in C^{2}} f(c_{1} y_{1} / M + c_{2} y_{2} / M)\bigg|
        \nonumber \\
        & = \bigg|\sum_{(c_{1}, c_{2}) \in C^{2}} G_{c_{1}, c_{2}}(y_{1}, y_{2})\bigg|
        \nonumber \\
        & \leq \frac{|C|^{2}}{\sqrt[n]{n^{2}}},
        \label{eq:G-bound3}
    \end{align}
    where the second step holds since, regarding the definition of $G_{c_{1}, c_{2}}$, the mapping $(c_{1}, c_{2}) \mapsto (c_{2}, -c_{1})$ is one-to-one, and the last step uses \eqref{eq:G-bound2}.
    
    Hence, we can upper-bound \eqref{eq:ez2-part2} as follows.
    \begin{align*}
        \eqref{eq:ez2-part2}
        & = \frac{1}{4\pi^{2} M^{2}} \iint_{y_{1}, y_{2} \in R''_{y} \setminus R'_{y}} |G(y_{1} / M + y_{2} / M)|^{n} \cdot \d y_{1} \d y_{2} \\
        & \leq \frac{1}{4\pi^{2} M^{2}} \cdot |R''_{y}| \cdot \Big( \max_{y_{1}, y_{2} \in R''_{y} \setminus R'_{y}} |G(y_{1} / M + y_{2} / M)| \Big)^{n} \\
        & \leq \frac{1}{4\pi^{2} M^{2}} \cdot |R''_{y}| \cdot \frac{|C|^{2n}}{n^{2}} \\
        & = \frac{b_{0}^{2} |C|^{2n}}{\pi^{2} M^{2} n^{2}} \\
        & = o(\rho_{n}^{2}),
    \end{align*}
    where the first step changes variables $y_{1} = M \theta_{1}$ and $y_{2} = M \theta_{2}$, the third step substitutes \eqref{eq:G-bound3}, the fourth step is elementary algebra (notice that $|R''_{y}| = 4b_{0}^{2}$), and the last step holds because $b_{0} \leq \frac{\pi}{2} \cdot b$ is a constant (see \eqref{eq:b0}) and $\rho_{n} = \Theta(\frac{|C|^{n}}{M \sqrt{n}})$ (see \eqref{eq:rho}).

    It remains to verify \eqref{eq:G-bound2}. Before doing so, we observe that for any point $(y_{1}, y_{2}) \in R''_{y} \setminus R'_{y}$,
    \begin{align*}
        \Big|\frac{c_{1} y_{1} + c_{2} y_{2}}{2M}\Big|
        \leq \frac{|c_{1}| + |c_{2}|}{2M} \cdot b_{0}
        \leq \frac{|c_{1}| + |c_{2}|}{16d \cdot b} \cdot (\pi / 2) \cdot b
        \leq \frac{|c_{1}| + |c_{2}|}{d} \cdot (\pi / 4)
        \leq \frac{\pi}{2}.
    \end{align*}
    where the second step holds because $b_{0} \leq (\pi / 2) \cdot b$ (see \eqref{eq:b0}) and $M \geq 8d \cdot b$ (see \eqref{eq:M-bound}), and the last step holds because $c \in C = C(d) = \{\pm 1, \pm 2, \dots, \pm d\}$. Since we always have $|\frac{c_{1} y_{1} + c_{2} y_{2}}{2M}| \leq \frac{\pi}{2}$ in the considered range, the arguments for \eqref{eq:f-bound5} can be reused, resulting in
    \begin{align}
        |f(c_{1} y_{1} / M + c_{2} y_{2} / M)|
        \leq \Big|\frac{\sin(c_{1} y_{1} + c_{2} y_{2})}{c_{1} y_{1} + c_{2} y_{2}}\Big| + \frac{1 - \cos(c_{1} y_{1} + c_{2} y_{2})}{40}.
        \label{eq:f-bound7}
    \end{align}
    Similarly, for any point $(y_{1}, y_{2}) \in R''_{y} \setminus R'_{y}$, we also have $|\frac{c_{2} y_{1} - c_{1} y_{2}}{2M}| \leq \frac{\pi}{2}$ and
    \begin{align}
        |f(c_{2} y_{1} / M - c_{1} y_{2} / M)|
        \leq \Big|\frac{\sin(c_{2} y_{1} - c_{1} y_{2})}{c_{2} y_{1} - c_{1} y_{2}}\Big| + \frac{1 - \cos(c_{2} y_{1} - c_{1} y_{2})}{40}.
        \label{eq:f-bound8}
    \end{align}

    To prove \eqref{eq:G-bound2} for any pair $(c_{1}, c_{2}) \in C^{2}$ and any point $(y_{1}, y_{2}) \in R''_{y} \setminus R'_{y}$, let us do case analysis.
    
    \vspace{.1in}
    \noindent
    {\bf Case~1: When $|c_{1} y_{1} + c_{2} y_{2}| \geq \frac{\pi}{2}$.}
    We deduce that
    \begin{align*}
        |G_{c_{1}, c_{2}}(y_{1}, y_{2})|
        & = \frac{1}{2} \cdot \Big|f(c_{1} y_{1} / M + c_{2} y_{2} / M) + f(c_{2} y_{1} / M - c_{1} y_{2} / M)\Big| \\
        & \leq \frac{1}{2} \cdot \Big(\rhs \mbox{ of } \eqref{eq:f-bound7} + 1\Big) \\
        & \leq \frac{1}{2} \cdot \Big(\frac{2}{\pi} + \frac{1}{20} + 1\Big) \\
        & \leq \frac{1}{\sqrt[n]{n^{2}}},
    \end{align*}
    where the second step uses \eqref{eq:f-bound7} and the fact that $|f(\theta)| \leq 1$ for any $\theta \in \R$ (see \eqref{eq:f-def}), the third step holds since $|\frac{\sin(z)}{z}| + \frac{1 - \cos(z)}{40} \leq \frac{2}{\pi} + \frac{1}{20}$ when $|z| \geq \frac{\pi}{2}$ (see the proof of \Cref{cla:discard-sub1-2}, {\bf Case~II}), and the last step holds since $\frac{1}{2} \cdot (\frac{2}{\pi} + \frac{1}{20} + 1) \approx 0.8433$ and $\frac{1}{\sqrt[n]{n^{2}}} = 1 - o_{n}(1)$.

    \vspace{.1in}    
    \noindent
    {\bf Case~2: When $|c_{2} y_{1} - c_{1} y_{2}| \geq \frac{\pi}{2}$.}
    Here we can reapply the arguments for {\bf Case~1}.

    \vspace{.1in}    
    \noindent
    {\bf Case~3: When $|c_{1} y_{1} + c_{2} y_{2}| \leq \frac{\pi}{2}$ and $|c_{2} y_{1} - c_{1} y_{2}| \leq \frac{\pi}{2}$.}
    Combining \eqref{eq:f-bound7} and \eqref{eq:f-bound8} gives
    \begin{align}
        |G_{c_{1}, c_{2}}(y_{1}, y_{2})|
        & = \frac{1}{2} \cdot \Big|f(c_{1} y_{1} / M + c_{2} y_{2} / M) + f(c_{2} y_{1} / M - c_{1} y_{2} / M)\Big|
        \nonumber \\
        & \leq \frac{1}{2} \cdot \Big(\rhs \mbox{ of } \eqref{eq:f-bound7} + \rhs \mbox{ of } \eqref{eq:f-bound8}\Big)
        \nonumber \\
        & \leq \frac{1}{2} \cdot \Big(e^{-\frac{1}{8}(c_{1} y_{1} + c_{2} y_{2})^{2}} + e^{-\frac{1}{8}(c_{2} y_{1} - c_{1} y_{2})^{2}}\Big),
        \label{eq:G-bound4}
    \end{align}
    where the third step holds since $|\frac{\sin(z)}{z}| + \frac{1 - \cos(z)}{40} \leq e^{-z^{2} / 8}$ for any $|z| \leq \frac{\pi}{2}$ (see the proof of \Cref{cla:discard-sub1-2}, {\bf Case~I}; notice that $|\frac{c_{1} y_{1} + c_{2} y_{2}}{2M}| \leq \frac{\pi}{2}$ and $|\frac{c_{2} y_{1} - c_{1} y_{2}}{2M}| \leq \frac{\pi}{2}$).

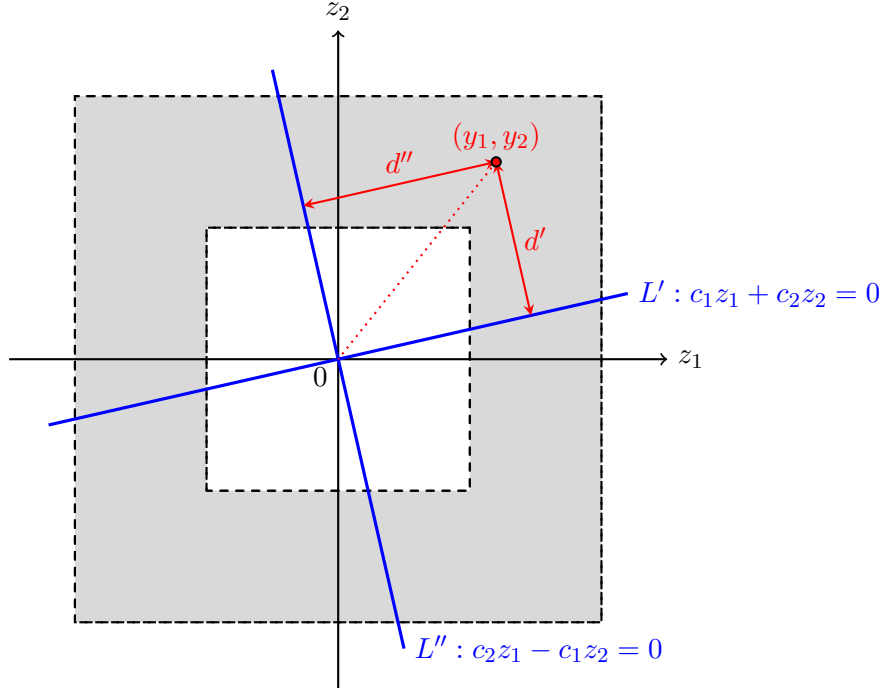
\begin{figure}[t]
    \centering
    \begin{tikzpicture}[thick, smooth, scale = 1.75]
    \draw [dashed, fill = gray!30] (-2, -2) rectangle (2, 2);
    \draw [dashed, fill = white] (-1, -1) rectangle (1, 1);
    \draw[->] (-2.5, 0) -- (2.5, 0);
    \draw[->] (0, -2.5) -- (0, 2.5);
    \node[right] at (2.5, 0) {$z_{1}$};
    \node[above] at (0, 2.5) {$z_{2}$};
    \draw[dashed] (2, 2) -- (-2, 2) -- (-2, -2) -- (2, -2) -- (2, 2);
    \draw[dashed] (1, 1) -- (-1, 1) -- (-1, -1) -- (1, -1) -- (1, 1);
    \draw[blue, very thick] (-2.2, -0.5) -- (2.2, 0.5) (-0.5, 2.2) -- (0.5, -2.2);
    \draw[red, dotted] (1.2, 1.5) -- (0, 0);
    \draw[stealth-stealth, red] (1.2, 1.5) -- (-0.2652, 1.1670);
    \draw[stealth-stealth, red] (1.2, 1.5) -- (1.4652, 0.3330);
    \node[right, red] at (1.3326, 0.9165) {$d'$};
    \node[above, red] at (0.4674, 1.3335) {$d''$};
    \node[above, red] at (1.2, 1.5) {$(y_{1}, y_{2})$};
    \draw[black, fill = red] (1.2, 1.5) circle [radius = 1pt];
    \node[right, blue] at (2.2, 0.5) {$L': c_{1} z_{1} + c_{2} z_{2} = 0$};
    \node[right, blue] at (0.5, -2.2) {$L'': c_{2} z_{1} - c_{1} z_{2} = 0$};

    \node[anchor = 45] at (0, 0) {$0$};
    \end{tikzpicture}
    \caption{Demonstration for the formula $\eqref{eq:G-bound4} \geq |G_{c_{1}, c_{2}}(y_{1}, y_{2})|$. The gray region represents the considered range $(y_{1}, y_{2}) \in R''_{y} \setminus R'_{y}$, where $R'_{y} = [-\sqrt{64\ln(n) / n},\; \sqrt{64\ln(n) / n}]^{2}$ and $R''_{y} = [-b_{0},\; b_{0}]^{2}$. The Pythagorean theorem guarantees that $d'^{2} + d''^{2} = y_{1}^{2} + y_{2}^{2}$.}
    \label{fig:discard-sub2-2}
\end{figure}

    Let us give a geometric interpretation for \eqref{eq:G-bound4}, which is demonstrated in \Cref{fig:discard-sub2-2}. Consider two straight lines $L': c_{1} z_{1} + c_{2} z_{2} = 0$ and $L'': c_{2} z_{1} - c_{1} z_{2} = 0$. By elementary algebra, we can see that the Euclidean distance $d' := d'(y_{1}, y_{2})$ (resp.\ the Euclidean distance $d'' := d''(y_{1}, y_{2})$) from a given point $(y_{1}, y_{2}) \in \R$ to straight line $L'$ (resp.\ straight line $L''$) is given by
    \begin{align}
    \label{eq:G-bound5}
        d' = \frac{|c_{1} y_{1} + c_{2} y_{2}|}{\sqrt{c_{1}^{2} + c_{2}^{2}}}
        \qquad \mbox{and} \qquad
        d'' = \frac{|c_{2} y_{1} - c_{1} y_{2}|}{\sqrt{c_{1}^{2} + c_{2}^{2}}}.
    \end{align}
    Moreover, $L'$ and $L''$ are perpendicular (by construction), so the Pythagorean theorem ensures that $d'^{2} + d''^{2} = y_{1}^{2} + y_{2}^{2}$. Accordingly, the larger distance $\max\{d',\; d''\}$ is lower bounded by
    \begin{align}
    \label{eq:G-bound6}
        \max\{d',\; d''\} \geq \sqrt{(d'^{2} + d''^{2}) / 2} = \sqrt{(y_{1}^{2} + y_{1}^{2}) / 2} \geq \sqrt{32\ln(n) / n},
    \end{align}
    where the first step uses the AM-GM inequality, and the last step holds since $(y_{1}, y_{2}) \notin R'_{y} = [-\sqrt{64\ln(n) / n},\; \sqrt{64\ln(n) / n}]^{2}$ (see \Cref{fig:discard-sub2-2}).

    Putting everything together results in
    \begin{align*}
        |G_{c_{1}, c_{2}}(y_{1}, y_{2})|
        \leq \eqref{eq:G-bound4}
        & = \frac{1}{2} \cdot \Big(e^{-\frac{c_{1}^{2} + c_{2}^{2}}{8} \cdot d'^{2}} + e^{-\frac{c_{1}^{2} + c_{2}^{2}}{8} \cdot d''^{2}}\Big) \\
        & \leq \frac{1}{2} \cdot \Big(e^{-\frac{1}{4} \cdot d'^{2}} + e^{-\frac{1}{4} \cdot d''^{2}}\Big) \\
        & \leq \frac{1}{2} \cdot \Big(e^{-\frac{1}{4} \cdot 0} + e^{-\frac{1}{4} \cdot 32\ln(n) / n}\Big) \\
        & = \frac{1}{2} \cdot \Big(1 + \frac{1}{\sqrt[n]{n^{8}}}\Big) \\
        & \leq \frac{1}{\sqrt[n]{n^{2}}},
    \end{align*}
    where the first step applies \eqref{eq:G-bound5}, the second step follows since the coefficients $c_{1}, c_{2} \in C$ are nonzero integers, the third step applies \eqref{eq:G-bound6} and the fact that $\min\{d',\; d''\} \geq 0$, the fourth step is elementary algebra, and the last step holds whenever $n \geq 6$ (see \eqref{eq:M-bound}).

    Combining all the three cases together gives \eqref{eq:G-bound2}. This completes the proof of \Cref{cla:discard-sub2-2}.
    \end{proof}
    
    Below we use \Cref{cla:discard-sub2-3,cla:discard-sub2-4,cla:discard-sub2-5} (as auxiliaries) to evaluate \eqref{eq:ez2-part3} in \Cref{cor:discard-sub2}. This evaluation together with \Cref{cla:discard-sub2-1,cla:discard-sub2-2} immediately gives \Cref{lem:discard-sub2}.

    \begin{claim}
        \label{cla:discard-sub2-3}
        There exists some constant $\eps_{2} := \eps_{2}(|C|, a) > 0$ such that, for any $(\theta_{1}, \theta_{2})\in [-\pi,\; \pi]^{2} \setminus R''_{\theta}$,
        \begin{align*}
            |G(\theta_{1}, \theta_{2}))| \leq |C|^{2} - \frac{1 - 1 / a}{2} \leq |C|^{(2 - \eps_{2})}.
        \end{align*}
    \end{claim}
    
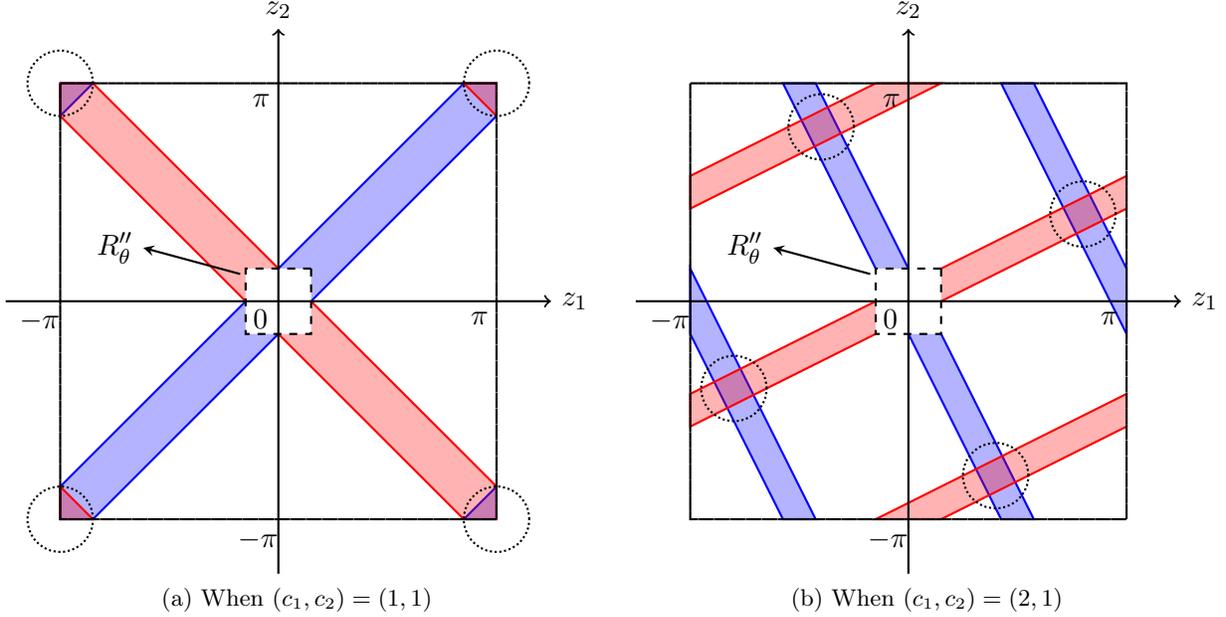
\begin{figure}[t]
    \centering
    \subfloat[When $(c_{1}, c_{2}) = (1, 1)$
    \label{subfig:discard-sub2-3:1}]{
    \begin{tikzpicture}[thick, smooth, scale = 1.45]
    \draw[blue, fill = blue, fill opacity = 0.3] (-2, 2) -- (-2, 1.7) -- (-1.7, 2) -- cycle;
    \draw[blue, fill = blue, fill opacity = 0.3] (-2, -2) -- (-2, -1.7) -- (1.7, 2) -- (2, 2) -- (2, 1.7) -- (-1.7, -2) -- cycle;
    \draw[blue, fill = blue, fill opacity = 0.3] (2, -2) -- (2, -1.7) -- (1.7, -2) -- cycle;

    \draw[red, fill = red, fill opacity = 0.3] (-2, -2) -- (-2, -1.7) -- (-1.7, -2) -- cycle;
    \draw[red, fill = red, fill opacity = 0.3] (-2, 2) -- (-2, 1.7) -- (1.7, -2) -- (2, -2) -- (2, -1.7) -- (-1.7, 2) -- cycle;
    \draw[red, fill = red, fill opacity = 0.3] (2, 2) -- (2, 1.7) -- (1.7, 2) -- cycle;

    \draw[dashed, fill = white] (0.3, 0.3) -- (0.3, -0.3) -- (-0.3, -0.3) -- (-0.3, 0.3) -- cycle;
    \node[above](n0) at (-1.5, 0.25) {$R''_{\theta}$};
    \draw [-stealth, smooth, thick] plot coordinates {(-0.35, 0.25) (n0.east)};

    \draw [dashed] (-2, -2) rectangle (2, 2);
    \draw[->] (-2.5, 0) -- (2.5, 0);
    \draw[->] (0, -2.5) -- (0, 2.5);
    \node[right] at (2.5, 0) {$z_{1}$};
    \node[above] at (0, 2.5) {$z_{2}$};
    \draw[dashed] (2, 2) -- (-2, 2) -- (-2, -2) -- (2, -2) -- (2, 2);

    \node[anchor = 45] at (0, 0) {$0$};
    \node[anchor = 45] at (-2, 0) {$-\pi$};
    \node[anchor = 45] at (2, 0) {$\pi$};
    \node[anchor = 45] at (0, -2) {$-\pi$};
    \node[anchor = 45] at (0, 2) {$\pi$};
    
    \draw[densely dotted] (2, 2) circle (0.3);
    \draw[densely dotted] (-2, 2) circle (0.3);
    \draw[densely dotted] (2, -2) circle (0.3);
    \draw[densely dotted] (-2, -2) circle (0.3);
    \end{tikzpicture}
    }
    \hfill
    \subfloat[When $(c_{1}, c_{2}) = (2, 1)$
    \label{subfig:discard-sub2-3:2}]{
    \begin{tikzpicture}[thick, smooth, scale = 1.45]
    \draw[blue, fill = blue, fill opacity = 0.3] (-2, -0.3) -- (-2, 0.3) -- (-0.85, -2) -- (-1.15, -2) -- cycle;
    \draw[blue, fill = blue, fill opacity = 0.3] (-1.15, 2) -- (-0.85, 2) -- (1.15, -2) -- (0.85, -2) -- cycle;
    \draw[blue, fill = blue, fill opacity = 0.3] (0.85, 2) -- (1.15, 2) -- (2, 0.3) -- (2, -0.3) -- cycle;

    \draw[red, fill = red, fill opacity = 0.3] (-2, 0.85) -- (-2, 1.15) -- (-0.3, 2) -- (0.3, 2) -- cycle;
    \draw[red, fill = red, fill opacity = 0.3] (-2, -1.15) -- (-2, -0.85) -- (2, 1.15) -- (2, 0.85) -- cycle;
    \draw[red, fill = red, fill opacity = 0.3] (-0.3, -2) -- (0.3, -2) -- (2, -1.15) -- (2, -0.85) -- cycle;

    \draw[dashed, fill = white] (0.3, 0.3) -- (0.3, -0.3) -- (-0.3, -0.3) -- (-0.3, 0.3) -- cycle;
    \node[above](n0) at (-1.5, 0.25) {$R''_{\theta}$};
    \draw [-stealth, smooth, thick] plot coordinates {(-0.35, 0.25) (n0.east)};

    \draw [dashed] (-2, -2) rectangle (2, 2);
    \draw[->] (-2.5, 0) -- (2.5, 0);
    \draw[->] (0, -2.5) -- (0, 2.5);
    \node[right] at (2.5, 0) {$z_{1}$};
    \node[above] at (0, 2.5) {$z_{2}$};
    \draw[dashed] (2, 2) -- (-2, 2) -- (-2, -2) -- (2, -2) -- (2, 2);

    \node[anchor = 45] at (0, 0) {$0$};
    \node[anchor = 45] at (-2, 0) {$-\pi$};
    \node[anchor = 45] at (2, 0) {$\pi$};
    \node[anchor = 45] at (0, -2) {$-\pi$};
    \node[anchor = 45] at (0, 2) {$\pi$};
    
    \draw[densely dotted] (1.6, 0.8) circle (0.3);
    \draw[densely dotted] (-0.8, 1.6) circle (0.3);
    \draw[densely dotted] (-1.6, -0.8) circle (0.3);
    \draw[densely dotted] (0.8, -1.6) circle (0.3);
    \end{tikzpicture}
    }
    \caption{The subregion $R_{c_{1}, c_{2}} := \{(\theta_{1}, \theta_{2}) \in [-\pi, \pi]^{2}: \mbox{$(c_{1} \theta_{1} + c_{2} \theta_{2}), (c_{2} \theta_{1} - c_{1} \theta_{2}) \in [-\frac{b_{0}}{M},\; \frac{b_{0}}{M}]_{\bigcirc}$}\}$ is the intersection of the red area (on which $(c_{1} \theta_{1} + c_{2} \theta_{2}) \in [-\frac{b_{0}}{M},\; \frac{b_{0}}{M}]_{\bigcirc}$) and the blue area (on which $(c_{2} \theta_{1} - c_{1} \theta_{2}) \in [-\frac{b_{0}}{M},\; \frac{b_{0}}{M}]_{\bigcirc}$). Given any parameters $t_{R}, t_{B} \in \R$, by construction a (red) straight line $c_{1} \theta_{1} + c_{2} \theta_{2} = t_{R}$ and a (blue) straight line $c_{2} \theta_{1} - c_{1} \theta_{2} = t_{B}$ are vertical. The subregion $R''_{\theta}$ (namely the dashed square) is excluded from $[-\pi, \pi]^{2}$. \Cref{subfig:discard-sub2-3:1,subfig:discard-sub2-3:2} respectively choose $(c_{1}, c_{2}) = (1, 1)$ and $(c_{1}, c_{2}) = (2, 1)$.}
    \label{fig:discard-sub2-3}
\end{figure}
    
    \begin{proof}
    Define the function $f_{c_{1}, c_{2}}(\theta_{1}, \theta_{2}) := \frac{1}{2}(f(c_{1} \theta_{1} + c_{2} \theta_{2}) + f(c_{2} \theta_{1} - c_{1} \theta_{2}))$ for each pair $(c_{1}, c_{2}) \in C^{2}$. Because $|f(\theta)| \leq 1$ for any $\theta \in \R$ (see \eqref{eq:f-def}), we also have $|f_{c_{1}, c_{2}}(\theta_{1}, \theta_{2})| \leq 1$ for any $(\theta_{1}, \theta_{2}) \in \R^{2}$. Recall that $|f(\theta)| \leq \frac{1}{a}$ for any $|\theta| \in [\frac{b_{0}}{M},\; \pi]$ (see \eqref{eq:fbound1}) and $\frac{b_{0}}{M} \leq \frac{(\pi / 2) \cdot b}{8d \cdot b} \leq \frac{\pi}{16}$ (see \eqref{eq:M-bound} and \eqref{eq:b0}). Since $f$ is $2\pi$-periodic (see \eqref{eq:f-def}), only on the $2\pi$-periodic range $[-\frac{b_{0}}{M},\; \frac{b_{0}}{M}]_{\bigcirc} := \bigcup_{k \in \Z} [-\frac{b_{0}}{M} + 2k\pi,\; \frac{b_{0}}{M} + 2k\pi]$ can we possibly have $|f(\theta)| > \frac{1}{a}$. Regarding the function $f_{c_{1}, c_{2}}$, we consider the analogous subregion
    \[
        R_{c_{1}, c_{2}} := \Big\{(\theta_{1}, \theta_{2}) \in [-\pi, \pi]^{2}: \mbox{$(c_{1} \theta_{1} + c_{2} \theta_{2}), (c_{2} \theta_{1} - c_{1} \theta_{2}) \in [-\frac{b_{0}}{M},\; \frac{b_{0}}{M}]_{\bigcirc}$}\Big\}.
    \]
    For any point $(\theta_{1}, \theta_{2}) \in [-\pi, \pi]^{2} \setminus R_{c_{1}, c_{2}}$, we have either $|f(c_{1} \theta_{1} + c_{2} \theta_{2})| \leq \frac{1}{a}$ or $|f(c_{2} \theta_{1} - c_{1} \theta_{2})| \leq \frac{1}{a}$ or both, which (together with the fact that $|f(\theta)| \leq 1$ for any $\theta \in \R$) implies
    \begin{align}
        \label{eq:discard-sub2-3}
        |f_{c_{1}, c_{2}}(\theta_{1}, \theta_{2})|
        \leq \frac{|f(c_{1} \theta_{1} + c_{2} \theta_{2})| + |f(c_{2} \theta_{1} - c_{1} \theta_{2})|}{2} \leq \frac{1 + 1 / a}{2}.
    \end{align}
    
    Because $C = C(d) = \{\pm 1, \pm 2, \dots, \pm d\}$ for a fixed integer $d > 1$, the considered pair-indices $(c_{1}, c_{2}) \in C^{2}$ include both $(c_{1}, c_{2}) = (1, 1)$ and $(c_{1}, c_{2}) = (2, 1)$.
    % \rnote{rephrase? \yj{Is the current argument easier to understand?} }
    As shown in \Cref{subfig:discard-sub2-3:1}, the union of the $\frac{b_{0}}{M}$-neighborhoods of four points $(\pm \pi, \pm \pi)$ covers $R_{1, 1}$. Also, as shown in \Cref{subfig:discard-sub2-3:2}, the union of the $\frac{b_{0}}{M}$-neighborhoods of four points $(\frac{4\pi}{5}, \frac{2\pi}{5})$, $(-\frac{2\pi}{5}, \frac{4\pi}{5})$, $(-\frac{4\pi}{5}, -\frac{2\pi}{5})$ and $(\frac{2\pi}{5}, -\frac{4\pi}{5})$ covers $R_{2, 1}$. We must have $R_{1, 1} \cap R_{2, 1} = \emptyset$, because the $\frac{b_{0}}{M} \leq \frac{\pi}{16}$ is small enough.

    Given the above arguments, we deduce that for any $(\theta_{1}, \theta_{2}) \in [-\pi, \pi]^{2} \setminus R''_{\theta}$,
    \begin{align*}
        |G(\theta_{1}, \theta_{2})|
        & = \Bigg|\sum_{(c_{1}, c_{2}) \in C^{2}} f(c_{1} \theta_{1} + c_{2} \theta_{2})\Bigg| \\
        & = \Bigg|\sum_{(c_{1}, c_{2}) \in C^{2}} f_{c_{1}, c_{2}}(\theta_{1}, \theta_{2})\Bigg| \\
        & \leq \frac{1 + 1 / a}{2} + |C|^{2} - 1 \\
        & = |C|^{2} - \frac{1 - 1 / a}{2}.
    \end{align*}
    Here the second step holds since, regarding the definition of $f_{c_{1}, c_{2}}$, the mapping $(c_{1}, c_{2}) \mapsto (c_{2}, -c_{1})$ is one-to-one. The third step uses \eqref{eq:discard-sub2-3}, the fact that $|f_{c_{1}, c_{2}}(\theta_{1}, \theta_{2})| \leq 1$ for each pair $(c_{1}, c_{2}) \in C^{2}$ and any $(\theta_{1}, \theta_{2}) \in \R^{2}$, and the fact that $R_{1, 1} \cap R_{2, 1} = \emptyset$.
    
    Because $a > |C|^{2} / \eps > 1$ (see \eqref{eq:M-bound}), there must exist some constant $\eps_{2} := \eps_{2}(|C|, a) > 0$ such that $|C|^{2} - (1 - 1 / a) / 2 \leq |C|^{(2 - \eps_{2})}$. This finishes the proof of \Cref{cla:discard-sub2-3}.
    \end{proof}

    \begin{claim}
        \label{cla:discard-sub2-5}
        The subregion $Q' \subseteq [-\pi, \pi]^{2}$ defined below has measure at most $|Q'| \leq |C|^{2}  \cdot \frac{4\pi b_{0}}{M}$.
        \[
            Q' := \Big\{(\theta_{1}, \theta_{2}) \in [-\pi, \pi]^{2}: |G(\theta_{1}, \theta_{2})| > \frac{|C|^{2}}{a}\Big\}.
        \]
    \end{claim}
    
    \begin{proof}
        To begin with, we introduce an auxiliary set $A_{c_{1}, w}$ for each $c_{1} \in C$ and any offset $w \in \R$.
        \begin{align}
            \label{eq:setA-def}
            A_{c_{1}, w} := \Big\{\theta_{1} \in [-\pi, \pi] : |f(c_{1} \theta_{1} + w)| > \frac{1}{a} \Big\}.
        \end{align}
        Recall \eqref{eq:fbound2} that $|A_{c_{1}, 0}| \leq \frac{2b_{0}}{M}$. Because $f$ is $2\pi$-periodic, regarding any integer $c_{1} \in C$, the function $f_{c_{1}}(\theta_{1}) := f(c_{1}\theta_{1})$ is also $2\pi$-periodic. Thus, we also have 
        \begin{align}
            \label{eq:Abound1}
            |A_{c_{1}, w}| \leq \frac{2b_0}{M}
        \end{align}
        for any $c_1 \in C$, $w \in \R$. It follows that, for any pair $(c_{1}, c_{2}) \in C^{2}$, we have
        \begin{align}
            \label{eq:fbound5}
            \Big|\Big\{ (\theta_{1}, \theta_{2}) \in [-\pi, \pi]^{2} : |f(c_{1} \theta_{1} + c_{2} \theta_{2})| > \frac{1}{a} \Big\}\Big| \leq 2\pi \cdot \frac{2b_0}{M} = \frac{4\pi b_0}{M}.
        \end{align}
        Recall that $G(\theta_{1}, \theta_{2}) = \sum_{(c_{1}, c_{2}) \in C^{2}} f(c_{1} \theta_{1} + c_{2} \theta_{2})$. Suppose $|G(\theta_{1}, \theta_{2})| > |C|^{2} / a$, then we must have $|f(c_{1} \theta_{1} + c_{2} \theta_{2})| > 1 / a$ for at least one pair $(c_{1}, c_{2}) \in C^{2}$. Union-bounding the measure \eqref{eq:fbound5} over all pairs $(c_{1}, c_{2}) \in C^{2}$ concludes the proof of \Cref{cla:discard-sub2-5}.
    \end{proof}

    \begin{claim}
        \label{cla:discard-sub2-4}
        The subregion $Q'' \subseteq [-\pi, \pi]^{2}$ defined below has measure at most $|Q''| \leq |C|^{4} \cdot (\frac{2b_{0}}{M})^{2}$.
        \[
            Q'' := \Big\{(\theta_{1}, \theta_{2}) \in [-\pi, \pi]^{2}: |G(\theta_{1}, \theta_{2})| > |C| + \eps \Big\}.
        \]
    \end{claim}
    
    \begin{proof}
        For this proof, we introduce a convenient decomposition of the function $G$. Observe that $G(\theta_{1}, \theta_{2}) = \sum_{c_{2} \in C} q(\theta_{1}, c_{2} \theta_{2})$, where 
        \[
            q(\theta_{1}, w) := \sum_{c_{1} \in C} f(c_{1} \theta_{1} + w).
        \]
        This allows us to rewrite $Q''$ as
        \[
            Q'' = \Big\{(\theta_{1}, \theta_{2}) \in [-\pi, \pi]^{2}: \Big|\sum_{c_{2} \in C} q(\theta_{1}, c_{2} \theta_{2})\Big| > |C| + \eps\Big\}.
        \]
        As $\max_{x} q(x) \leq |C| \cdot \max_x f(x) \leq |C|$, $q$ is upper-bounded by $|C|$. Thus for any point $(\theta_1, \theta_2) \in Q''$, there must exist two distinct coefficients $c_2, c'_{2} \in C$ such that $|q(\theta_1, c_2\theta_2)|, |q(\theta_1, c'_{2}\theta_2)| > \frac{\eps}{|C|}$. In other words, the considered set $Q''$ is covered by
        \begin{align}
            \label{eq:restate-Q''}
            Q'' \subseteq \bigcup_{c_2 \neq c'_{2} \in C} Q_{c_2, c'_{2}},
        \end{align}
        where
        \begin{align*}
            Q_{\alpha, \beta} := \Big\{(\theta_{1}, \theta_{2}) \in [-\pi, \pi]^{2}: |q(\theta_1, \alpha\theta_2)|, |q(\theta_1, \beta\theta_2)| \geq \frac{|C|}{a} \Big\}.
        \end{align*}
        This definition of $Q_{\alpha, \beta}$ suffices because $a > |C|^2 / \eps$ (see \eqref{eq:M-bound}) and thus $\frac{|C|}{a} < \frac{\eps}{|C|}$.
        
        To complete the proof, it suffices to show that
        \begin{align}
            \label{eq:Qc2-final-todo}
            |Q_{c_2, c'_{2}}| \leq |C|^2 \cdot \Big(\frac{2b_0}{M}\Big)^{2}
        \end{align}
        for each pair $(c_2, c'_{2}) \in C^2$ with $c_2 \neq c'_{2}$. (Suppose so, then the measure $|Q''|$ can be upper-bounded by $|C|^2 \cdot \max_{c_2 \neq c'_{2}} |Q_{c_2, c'_{2}}| \leq |C|^{4} \cdot (\frac{2b_{0}}{M})^{2}$ using \eqref{eq:restate-Q''}.)
        
        To show \eqref{eq:Qc2-final-todo}, fix a pair $(c_2, c'_2) \in C^{2}$ such that $c_2 \neq c'_{2}$. We observe that 
        \begin{align*}
             Q_{c_2, c'_2} = \Big\{(\theta_{1}, \theta_{2}) \in [-\pi, \pi]^{2}: \theta_{1} \in \big( B_{c_{2} \theta_{2}} \cap B_{c'_{2} \theta_{2}} \big) \Big\}, 
        \end{align*}
        for $B_{w} := \{\theta_{1} \in [-\pi, \pi] : |q(\theta_{1}, w)| > \frac{|C|}{a}\}$. Intuitively, $B_w$ captures the values for which $|q(\theta_1, w)|$ is large. Because $B_w$ is determined by the function $q(\theta_1, w) = \sum_{c \in C} f(c_1\theta_1 + w)$, it has a convenient property: for any $w, z \in \R$, $B_{w + z}$ is just the set $B_w$ translated by $z \pmod{2\pi}$. \Cref{fig:discard-sub2-5} provides a visual aid for the structure of this set and the remainder of the proof.
        
        It will be convenient to write $|Q_{c_2, c'_{2}}|$ in terms of the probability that a random point in $[-\pi, \pi]^2$ is contained in $Q_{c_2, c'_{2}}$. We have
        \begin{align}
            |Q_{c_2, c'_2}|
            & = 4\pi^2 \cdot \Prx_{\bm{\theta}_1, \bm{\theta_2} \sim [-\pi, \pi]} \Big[\bm{\theta}_1 \in \big( B_{c_{2} \bm{\theta}_{2}} \cap B_{c'_{2} \bm{\theta}_{2}} \big)\Big]
            \nonumber \\
            &= 4\pi^2 \cdot\Prx_{\bm{\theta}_1, \bm{\theta_2} \sim [-\pi, \pi]} \Big[\bm{\theta}_1 \in \big( B_{0} \cap B_{(c'_{2} - c_2) \bm{\theta}_{2}} \big)\Big]
            \nonumber \\
            &= 4\pi^2 \cdot \Prx_{\bm{\theta}_1 \sim [-\pi, \pi]} \Big[\bm{\theta}_1 \in B_0\Big] \cdot \Prx_{\bm{\theta}_1, \bm{\theta_2} \sim [-\pi, \pi]} \Big[\bm{\theta}_1 \in B_{(c'_{2} - c_2)\bm{\theta}_2} ~ \Big| ~ \bm{\theta}_1 \in B_0\Big].
            \label{eq:discard-sub2-4:5}
        \end{align}
        Here the second step follows from translating both sets by $c_2 \bm{\theta}_2 \pmod{2\pi}$, which does not change the size of the intersection. And the last step uses the identity $\Pr[A \cap B] = \Pr[A] \cdot \Pr[B \mid A]$.
        
        Given any nonzero integer $k$ and any fixed offset $w \in [-\pi, \pi]$, we observe that
        \begin{align}
            \label{eq:discard-sub2-4:6}
            \Prx_{\bm{\theta}_2\sim[-\pi, \pi]} [w \in B_{k \bm{\theta}_2}] = \Prx_{\by \sim[-\pi, \pi]} [\by \in B_0] = \frac{|B_0|}{2\pi},
        \end{align}
        because sampling $\bm{\theta_2}$ uniformly from $[-\pi, \pi]$ distributes $k \bm{\theta}_2 \pmod{2\pi}$ uniformly over $[-\pi, \pi]$. \eqref{eq:discard-sub2-4:6} holds for any nonzero integer $k \neq 0$ and any offset $w \in [-\pi, \pi]$, so both (conditional) probabilities in \eqref{eq:discard-sub2-4:5} can be substituted with $\eqref{eq:discard-sub2-4:6} = |B_0| / (2\pi)$.
        
        Recall the definition of the set family $A_{c_{1}, w}$ in \eqref{eq:setA-def}. Using \eqref{eq:Abound1} to union-bound \eqref{eq:discard-sub2-4:5} yields \eqref{eq:Qc2-final-todo}:
        \[
            |Q_{c_2, c'_2}| = 4\pi^2 \cdot \Big(\frac{|B_0|}{2\pi}\Big)^2 \leq \Big(\sum_{c \in C} |A_{c, 0}|\Big)^2 \leq |C|^2 \cdot \Big(\frac{2b_0}{M}\Big)^2.
            \qedhere
        \]
    \end{proof}
    
\begin{figure}[t]
    \centering
    \begin{tikzpicture}[thick, smooth, scale = 2.25]
    \path[fill = gray!30, domain = -2: 2] (-3.30, 0.60) -- (-3.30, 0.45) -- (-2.75, 1) -- (-2.9, 1) -- cycle;
    \draw[dotted] (-3.30, 0.60) -- (-2.9, 1) (-3.30, 0.45) -- (-2.75, 1);
    \path[fill = gray!30] (-3.30, -0.05) -- (-3.30, -0.50) -- (-1.8, 1) -- (-2.25, 1) -- cycle;
    \draw[dotted] (-3.30, -0.05) -- (-2.25, 1) (-3.30, -0.50) -- (-1.80, 1);
    \path[fill = gray!30] (-3.30, -0.85) -- (-3.30, -0.95) -- (-1.35, 1) -- (-1.45, 1) -- cycle;
    \draw[dotted] (-3.30, -0.85) -- (-1.45, 1) (-3.30, -0.95) -- (-1.35, 1);
    
    \path[fill = gray!30, domain = -2: 2] (-2.9, -1) -- (-2.75, -1) -- (-0.75, 1) -- (-0.9, 1) -- cycle;
    \draw[dotted] (-2.9, -1) -- (-0.9, 1) (-2.75, -1) -- (-0.75, 1);
    \path[fill = gray!30] (-2.25, -1) -- (-1.80, -1) -- (0.20, 1) -- (-0.25, 1) -- cycle;
    \draw[dotted] (-2.25, -1) -- (-0.25, 1) (-1.80, -1) -- (0.20, 1);
    \path[fill = gray!30] (-1.45, -1) -- (-1.35, -1) -- (0.65, 1) -- (0.55, 1) -- cycle;
    \draw[dotted] (-1.45, -1) -- (0.55, 1) (-1.35, -1) -- (0.65, 1);
    
    \path[fill = gray!30, domain = -2: 2] (-0.9, -1) -- (-0.75, -1) -- (1.25, 1) -- (1.1, 1) -- cycle;
    \draw[dotted] (-0.9, -1) -- (1.1, 1) (-0.75, -1) -- (1.25, 1);
    \path[fill = gray!30] (-0.25, -1) -- (0.20, -1) -- (2.20, 1) -- (1.75, 1) -- cycle;
    \draw[dotted] (-0.25, -1) -- (1.75, 1) (0.20, -1) -- (2.20, 1);
    \path[fill = gray!30] (0.55, -1) -- (0.65, -1) -- (2.65, 1) -- (2.55, 1) -- cycle;
    \draw[dotted] (0.55, -1) -- (2.55, 1) (0.65, -1) -- (2.65, 1);
    
    \path[fill = gray!30, domain = -2: 2] (1.1, -1) -- (1.25, -1) -- (3.25, 1) -- (3.1, 1) -- cycle;
    \draw[dotted] (1.1, -1) -- (3.1, 1) (1.25, -1) -- (3.25, 1);
    \path[fill = gray!30] (1.75, -1) -- (2.20, -1) -- (3.30, 0.1) -- (3.3, 0.55) -- cycle;
    \draw[dotted] (1.75, -1) -- (3.3, 0.55) (2.20, -1) -- (3.30, 0.1);
    \path[fill = gray!30] (2.55, -1) -- (2.65, -1) -- (3.3, -0.35) -- (3.3, -0.25) -- cycle;
    \draw[dotted] (2.55, -1) -- (3.3, -0.25) (2.65, -1) -- (3.3, -0.35);
    
    \path[fill = gray!30, domain = -2: 2] (3.1, -1) -- (3.25, -1) -- (3.30, -0.95) -- (3.30, -0.80) -- cycle;
    \draw[dotted] (3.1, -1) -- (3.30, -0.80) (3.25, -1) -- (3.30, -0.95);

    \draw[color = blue] (-3.3, -1) -- (3.3, -1) (-3.3, 1) -- (3.3, 1);
    \draw[color = blue] (-3, -1) -- (-3, 1) (-1, -1) -- (-1, 1) (1, -1) -- (1, 1) (3, -1) -- (3, 1);

    \draw[red, dotted] (-1, -1.2) -- (-1, 1.2);
    \draw[red, dotted] (-0.80, -1.2) -- (-0.80, 1.2);
    \draw[red, dotted] (-0.45, -1.2) -- (-0.45, 1.2);
    \draw[red, dotted] (-0.35, -1.2) -- (-0.35, 1.2);
    \draw[red, dotted] (0.1, -1.2) -- (0.1, 1.2);
    \draw[red, dotted] (0.25, -1.2) -- (0.25, 1.2);
    \draw[red, dotted] (0.75, -1.2) -- (0.75, 1.2);
    \draw[red, dotted] (1, -1.2) -- (1, 1.2);

    \path[draw, red, fill = red!30] (-1, 0.75) -- (-0.8, 0.95) -- (-0.8, 1) -- (-0.9, 1) -- (-1, 0.9) -- cycle;
    \path[draw, red, fill = red!30] (-1, -0.2) -- (-0.8, 0) -- (-0.8, 0.45) -- (-1, 0.25) -- cycle;
    \path[draw, red, fill = red!30] (-1, -0.65) -- (-0.8, -0.45) -- (-0.8, -0.35) -- (-1, -0.55) -- cycle;
    \path[draw, red, fill = red!30] (-0.9, -1) -- (-0.8, -1) -- (-0.8, -0.9) -- cycle;
    
    \path[draw, red, fill = red!30] (-0.45, 0) -- (-0.45, -0.1) -- (-0.35, 0) -- (-0.35, 0.1) -- cycle;
    \path[draw, red, fill = red!30] (-0.45, 0.35) -- (-0.35, 0.45) -- (-0.35, 0.9) -- (-0.45, 0.8) -- cycle;
    \path[draw, red, fill = red!30] (-0.45, -0.55) -- (-0.45, -0.7) -- (-0.35, -0.6) -- (-0.35, -0.45) -- cycle;
    
    \path[draw, red, fill = red!30] (0.1, 0.9) -- (0.2, 1) -- (0.1, 1) -- cycle;
    \path[draw, red, fill = red!30] (0.1, 0.45) -- (0.25, 0.6) -- (0.25, 0.7) -- (0.1, 0.55) -- cycle;
    \path[draw, red, fill = red!30] (0.1, 0) -- (0.1, -0.15) -- (0.25, 0) -- (0.25, 0.15) -- cycle;
    \path[draw, red, fill = red!30] (0.1, -0.65) -- (0.1, -1) -- (0.2, -1) -- (0.25, -0.95) -- (0.25, -0.5) -- cycle;
    
    \path[draw, red, fill = red!30] (0.75, 0.65) -- (0.75, 0.5) -- (1, 0.75) -- (1, 0.9) -- cycle;
    \path[draw, red, fill = red!30] (0.75, 0) -- (0.75, -0.45) -- (1, -0.2) -- (1, 0.25) -- cycle;
    \path[draw, red, fill = red!30] (0.75, -0.8) -- (0.75, -0.9) -- (1, -0.65) -- (1, -0.55) -- cycle;

    \draw[red] (-1, -1.33) -- (-1, -1.27);
    \draw[red] (-0.8, -1.27) -- (-0.8, -1.33);
    \draw[red] (-1, -1.3) -- (-0.8, -1.3);
    
    \draw[red] (-0.45, -1.33) -- (-0.45, -1.27);
    \draw[red] (-0.35, -1.27) -- (-0.35, -1.33);
    \draw[red] (-0.45, -1.3) -- (-0.35, -1.3);
    
    \draw[red] (0.1, -1.33) -- (0.1, -1.27);
    \draw[red] (0.25, -1.27) -- (0.25, -1.33);
    \draw[red] (0.1, -1.3) -- (0.25, -1.3);
    
    \draw[red] (0.75, -1.33) -- (0.75, -1.27);
    \draw[red] (1, -1.27) -- (1, -1.33);
    \draw[red] (0.75, -1.3) -- (1, -1.3);
    
    \node(n0)[left, red] at (-0.3, -1.65) {$B_{0}$};
    \draw [-stealth, smooth, red] plot coordinates {(n0.north west) (-0.9, -1.32)};
    \draw [-stealth, smooth, red] plot coordinates {(n0.north) (-0.4, -1.32)};
    \draw [-stealth, smooth, red] plot coordinates {(n0.north east) (0.175, -1.32)};
    \draw [-stealth, smooth, red] plot coordinates {(n0.east) (0.7, -1.55) (0.875, -1.32)};

    \draw[<->] (-3.5, 0) -- (3.5, 0);
    \draw[<->] (0, -1.5) -- (0, 1.5);
    \node[above] at (0, 1.5) {$\theta_2$};
    \node[right] at (3.5, 0) {$\theta_1$};
    
    \draw[thick] (-2, -0.03) -- (-2, 0.03);
    \draw[thick] (2, -0.03) -- (2, 0.03);
    
    \node[anchor = 45] at (0, 0) {$0$};
    \node[anchor = 30] at (-1, 0) {$-\pi$};
    \node[anchor = 30] at (-2, 0) {$-2\pi$};
    \node[anchor = 30] at (-3, 0) {$-3\pi$};
    \node[anchor = 45] at (1, 0) {$\pi$};
    \node[anchor = 30] at (2, 0) {$2\pi$};
    \node[anchor = 30] at (3, 0) {$3\pi$};
    \end{tikzpicture}
    \caption{Visual aid for the proof of \Cref{cla:discard-sub2-4}. For any offset $\theta_2 \in [-\pi, \pi]$, the range $B_{(c'_{2} - c_2)\theta_2}$ is the intersection of the line segment $(-\pi, \theta_2)$--$(\pi, \theta_2)$ with the gray regions. The red regions indicate $\{(\theta_1, \theta_2) \in [-\pi, \pi]^{2}: \theta_1 \in B_0 \cap B_{(c'_{2} - c_2)\theta_2}\}$.}
    \label{fig:discard-sub2-5}
\end{figure}
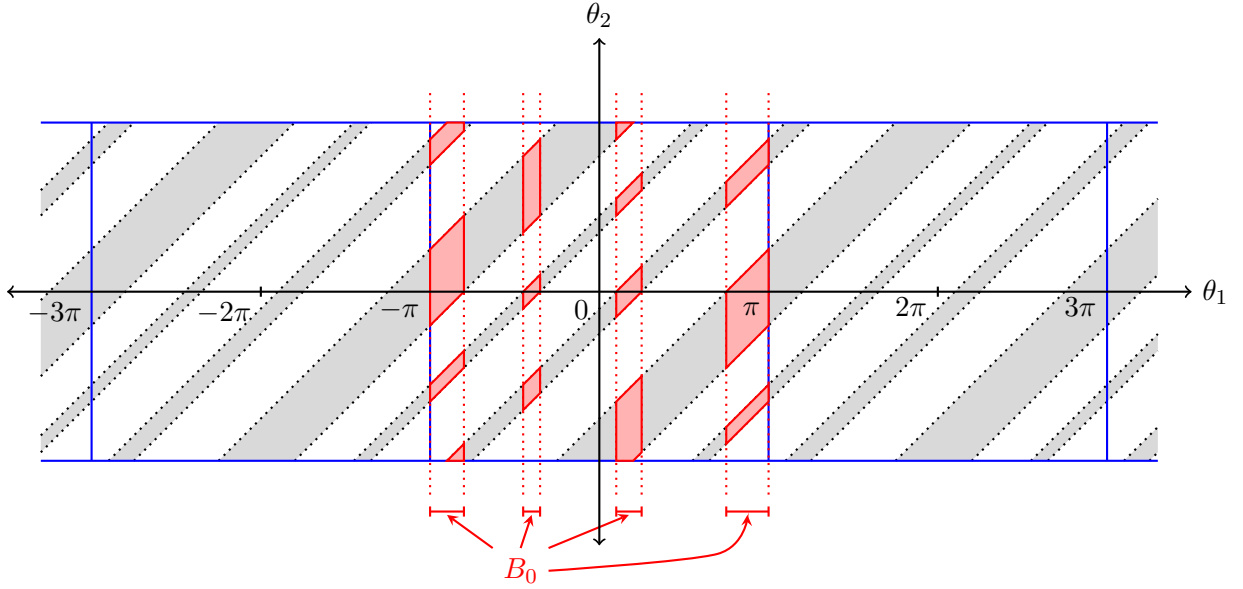

    \begin{corollary}
        \label{cor:discard-sub2}
        $\eqref{eq:ez2-part3} = o(\rho_{n}^{2})$.
    \end{corollary}
    
    \begin{proof}
    For ease of reference, let us restate the results in \Cref{cla:discard-sub2-3,cla:discard-sub2-4,cla:discard-sub2-5}:
    \begin{itemize}
        \item \Cref{cla:discard-sub2-3}: $|G(\theta_{1}, \theta_{2}))| \leq |C|^{(2 - \eps_{2})}$ for any point $(\theta_{1}, \theta_{2}) \in [-\pi,\; \pi]^{2} \setminus R''_{\theta}$.

        \item \Cref{cla:discard-sub2-5}: The subregion $Q'$ on which $|G(\theta_{1}, \theta_{2})| > |C|^{2} / a$ has measure at most $|C|^{2} \cdot \frac{4\pi b_{0}}{M}$.

        \item \Cref{cla:discard-sub2-4}: The subregion $Q''$ on which $|G(\theta_{1}, \theta_{2})| > |C| + \eps$ has measure at most $|C|^{4} \cdot (\frac{2b_{0}}{M})^{2}$.
    \end{itemize}
    Obviously, the three subregions $(Q'' \setminus R''_{\theta})$ and $(Q' \setminus Q'')$ and $[-\pi,\; \pi]^{2} \setminus Q'$ together covers the domain of integration for \eqref{eq:ez2-part3}, namely $[-\pi,\; \pi]^{2} \setminus R''_{\theta}$. As a consequence, we have
    \begin{align*}
        \eqref{eq:ez2-part3}
        & \leq \frac{1}{4\pi^{2}} \bigg(\iint\limits_{(\theta_{1}, \theta_{2}) \in [-\pi,\; \pi]^{2} \setminus Q'}
        + \iint\limits_{(\theta_{1}, \theta_{2}) \in Q' \setminus Q''}
        + \iint\limits_{(\theta_{1}, \theta_{2}) \in Q'' \setminus R''_{\theta}}\bigg)
        |G(\theta_{1}, \theta_{2})|^{n} \cdot \d \theta_{1} \d \theta_{2} \\
        & \leq \frac{1}{4\pi^{2}} \bigg( \Big(\frac{|C|^{2}}{a}\Big)^{n} \cdot 4\pi^{2} + (|C| + \eps)^n \cdot |Q'| + |C|^{(2 - \eps_{2})n} \cdot |Q''| \bigg) \\
        & \leq \frac{1}{4\pi^{2}} \bigg( \Big(\frac{|C|^{2}}{a}\Big)^{n} \cdot 4\pi^{2} + (|C| + \eps)^n \cdot |C|^{2}  \cdot \frac{4\pi b_{0}}{M} + |C|^{(2 - \eps_{2})n} \cdot |C|^{4} \cdot \Big(\frac{2b_{0}}{M}\Big)^{2} \bigg) \\
        & = O(\eps^{n}) + O\Big(\frac{(|C| + \eps)^{n}}{M}\Big) + O\Big(\frac{|C|^{(2 - \eps_{2})n}}{M^{2}}\Big).
    \end{align*}
    Here the second step follows from \Cref{cla:discard-sub2-3} and the definitions of $Q'$ and $Q''$. The third step applies \Cref{cla:discard-sub2-4,cla:discard-sub2-5}. The fourth step holds since (for the first term) $a > |C|^{2} / \eps$ using \eqref{eq:M-bound} and (for the second/third terms) both $|C|$ and $b_{0}$ are constants. Since $M = O^{*}(|C|^{(1 - \eps)n})$ for a constant $\eps > 0$, we know from \eqref{eq:rho} that $\rho_{n}^{2} = \Theta(|C|^{2n} / (M^{2} n)) = \Omega(|C|^{\eps n})$. Observe that:
    \begin{itemize}
        \item The first term $O(\eps^{n}) = o_{n}(1) = o(\rho_{n}^{2})$.
        
        \item The second term $O((|C| + \eps)^{n} \cdot M^{-1}) = O(\rho_{n}^{2}) \cdot (|C| + \eps)^{n} \cdot \frac{M n}{|C|^{2n}}$, by the definition of $\rho_{n}$ \eqref{eq:rho}. Then an upper bound $O(\rho_{n}^{2}) \cdot |C|^{-(\eps / 4)n} = o(\rho_{n}^{2})$ can be inferred from the calculation below.
        \begin{align*}
            (|C| + \eps)^{n} \cdot \frac{M n}{|C|^{2n}}
            & = O\Big((|C| + \eps)^{n} \cdot \frac{1}{|C|^{(1 + \eps / 2)n}}\Big) \\
            & = O\Big(|C|^{(1 + \eps / 4)n} \cdot \frac{1}{|C|^{(1 + \eps / 2)n}}\Big) \\
            & = O(|C|^{-(\eps / 4)n}).
        \end{align*}
        Here the first step holds because $M = O^{*}(|C|^{(1 - \eps)n})$ and then (given that $\eps > 0$ is a constant) $M n = O(|C|^{(1 - \eps / 2)n})$. The second step holds since $|C| + \eps \leq |C| \cdot (1 + \eps / 4) < |C|^{1 + \eps / 4}$, given that $|C| \geq 4$ (namely $C = C(d) = \{\pm 1, \dots, \pm d\}$ for some integer $d > 1$).
        
        \item The third term $O(|C|^{(2 - \eps_{2})n} \cdot M^{-2}) = O(\rho_{n}^{2}) \cdot n \cdot |C|^{-\eps_{2}n} = o(\rho_{n}^{2})$, given that $\eps_{2} = \eps_{2}(|C|, n) > 0$ is a constant (see the statement of \Cref{cla:discard-sub2-3}).
    \end{itemize}
    Putting everything together completes the proof of \Cref{cor:discard-sub2}.
    \end{proof}

    \begin{proof}[Proof of \Cref{lem:discard-sub2}]
    Putting \Cref{cla:discard-sub2-1,cla:discard-sub2-2,cor:discard-sub2} together, we have 
    \begin{align*}
        \E_{\vec{\bx}}[\bcalZ^{2}]
        ~ \leq ~ \eqref{eq:ez2-part1} ~ + ~ \eqref{eq:ez2-part2} ~ + ~ \eqref{eq:ez2-part3}
        ~ \leq ~ \rho_{n} \cdot (1 + o_{n}(1)).
\end{align*}
This finishes the proof of \Cref{lem:discard-sub2}.
    \end{proof}

\end{document}